\documentclass[12pt]{iopart}
\usepackage{iopams}  

\usepackage[dvips]{graphicx}
\usepackage{multirow}

\usepackage{bm}

\makeatletter
\newcommand{\figcaption}[1]{\def\@captype{figure}\caption{#1}}
\newcommand{\tblcaption}[1]{\def\@captype{table}\caption{#1}}
\makeatother

\def\bx{\bm{x}}
\def\bs{\bm{s}}

\def\F{\mathbb{F}}
\def\C{\mathbb{C}}
\def\Pr{{\rm Pr}}
\def\rE{{\rm E}}

\def\mix{\mathop{\hbox{\rm mix}}\nolimits}
\def\ideal{\mathop{\hbox{\rm ideal}}\nolimits}

\def\per{\mathop{\hbox{\rm per}}\nolimits}
\def\est{\mathop{\hbox{\rm est}}\nolimits}

\newtheorem{condition}    {Condition}%[section]
\newtheorem{theorem}{Theorem}
\newtheorem{definition}{Definition}
\newtheorem{rem}     {Remark}%[section]56

 \newenvironment{proof}{\par \noindent
            {\bf Proof. \hspace{2mm}}}{\hfill$\Box$ \vspace*{3mm}}

\def\Label#1{\label{#1} \hbox{[ #1 ]} }
\def\Label{\label}

\begin{document}

\title[Finite-key analysis of the decoy method]{Security analysis of the decoy method with the Bennett-Brassard 1984 protocol for finite key lengths}

\author{Masahito Hayashi$^{1,2}$ and Ryota Nakayama$^{3,1}$}

\address{$^{1}$ Graduate School of Mathematics, Nagoya University, Furocho, Chikusa-ku, Nagoya, 464-860 Japan
}
%\ead{custserv@iop.org}
\address{$^{2}$
Centre for Quantum Technologies, National University of Singapore, 3 Science Drive 2, Singapore 117542
}
\address{$^{3}$
Graduate School of Information Sciences, Tohoku University, Aoba-ku, Sendai, 980-8579, Japan}

\begin{abstract}
This paper provides a formula for the sacrifice bit-length for privacy amplification with the Bennett-Brassard
1984 protocol for finite key lengths when we employ the decoy method. 
Using the formula, we can guarantee the security parameter for realizable quantum key distribution system. 
The key generation rates with finite key lengths are numerically evaluated. 
The proposed method improves the existing key generation rate even in the asymptotic setting.
\end{abstract}

%02.20.-a 	Group theory 
%(for algebraic methods in quantum mechanics, see 03.65.Fd; 
%for symmetries in elementary particle physics, see 11.30.-j)
%03.65.Ta Foundations of quantum mechanics; measurement theory 
%(for optical tests of quantum theory, see 42.50.Xa)
%03.65.Ud Entanglement and quantum nonlocality 
%(e.g. EPR paradox, Bell's inequalities, GHZ states, etc.) 
%(for entanglement production in quantum information, see 03.67.Mn; 
%for entanglement in Bose-Einstein condensates, see 03.75.Gg)
%03.65.Wj State reconstruction, quantum tomography
%03.67.-a 	Quantum information
%03.67.Dd 	Quantum cryptography
%03.67.Hk 	Quantum communication
%05.30.Jp 	Boson systems 
%02.30.Hq 	Ordinary differential equations
%02.30.Nw 	Fourier analysis

\vspace{2pc}

%\submitto{\NJP}

\maketitle

%\tableofcontents

\section{Introduction}\Label{s0}
\subsection{Background}
Quantum key distribution (QKD) protocol proposed by 
Bennett-Brassard \cite{BB84} is one of the most applicable protocols in quantum information.
The conventional BB84 QKD protocol generates keys with the 
matched bases\footnote{In this paper, when Alice's basis 
is the same as Bob's basis, the basis is called matched.}, 
which are called raw keys and are trivially shown to be secure 
with the noiseless channel and the perfect single photon source.
However, in the realistic setting, there are two obstacles for security.
One is the noise of the communication quantum channel.
Due to the presence of the noise, the eavesdropper can obtain a part of information of raw keys behind the noise.
The second one is the imperfection of the photon source.
If the sender sends the two-photon state instead of the single photon state,
the eavesdropper can obtain one photon so that she can obtain information perfectly.
Many realized QKD systems have been realized with weak coherent pulses.
In this case, the photon number of transmitted pulses 
obeys the Poisson distribution, whose average is given by the intensity $\mu$ of the pulse.
The first problem can be resolved by the application of 
the error correction and the random privacy amplification to raw keys \cite{SP00,M01,WMU06,Hayashi3}.
In the privacy amplification stage, we amplify the security of our raw keys by sacrificing a part of our raw keys.
The security of final keys depends on the decreasing number of keys in the privacy amplification stage,
which is called the {\it sacrifice bit-length}.
Shor-Preskill \cite{SP00} and Mayers \cite{M01} showed 
that this method gives the secure keys asymptotically
when the rate of the sacrifice bit-length is greater than a certain amount.
In order to solve the second problem, 
Gottesman-Lo-L\"{u}tkenhaus-Preskill (GLLP)\cite{GLLP} extended their result to the case when the photon source has the imperfection.
However, GLLP's result assumes 
the fractions of respective photon number pulses among received pulses.
Indeed, there is a possibility that the eavesdropper 
can control the receiver's detection rate dependently of the photon number
because pulses with the different photon number can be distinguished by 
the eavesdropper.
In order to solve this problem, we need to estimate the detection rate of the single photon pulses.
Hwang proposed the decoy method to estimate the detection rate \cite{decoy1}.
This method has been improved by many researchers\cite{decoy2,decoy3,Ma05,Wang05,H1,decoy4,Lo2,wang2,wang3}.
In this method, in order to estimate the detection rates,
the sender randomly chooses several kinds of pulses with different intensities.
The first kind of pulses are the signal pulses, which generate raw keys.
The other kind of pulses are the decoy pulses, which are used for estimating the operation by the eavesdropper and have a different intensity from the signal pulses.

However, we still cannot realize a truly secure QKD system in the real world
due to the finiteness of the coding length.
Most of the above results assume the asymptotic setting except for Mayers\cite{M01}.
Also, their privacy amplification requires many calculation times.
Renner \cite{Renner} proposed to use universal$_2$ hash functions for 
privacy amplification and 
showed the security under this kind of hash functions.
Universal$_2$ hash functions have been recognized as 
a fundamental tool for information theoretical security \cite{Hayashi2,Hayashi5,cq-security,WC81}.
His security proof is quite different from 
the traditional Shor-Preskill formalism
in the following points.
He focused on the trace norm of the difference between the true state and the ideal state as the security parameter
because the trace norm is universally composable \cite{uni2}.
In the following, we call the trace norm the universal composability criterion.
As another different point, he employed 
the left over hashing lemma (privacy amplification)
while the traditional Shor-Preskill formalism employs error correction.
On the other hand, in the context of the traditional Shor-Preskill formalism,
it was shown that
the leaked information can be evaluated only by the phase error probability\cite{Hayashi3,Miyadera,H2,Koashi,Renes10},
which implies that the phase error correction guarantees the security.
Using this fact,
a previous paper \cite{TH11} showed that 
the security under a wider class of hash function, which is called
$\varepsilon$-almost dual universal$_2$ hash function.
%This class is also larger than the class 

In order to treat the finiteness problem in the single photon case,
when $n$ is the block length of our code,
another previous paper\cite{Hayashi3} 
considers the asymptotic expansion of the coding length up to 
%the lower order than 
the order $\sqrt{n}$\footnote{Analysis of this type of asymptotic expansion is called the second order analysis
and has attracted attention among information theory community due to the relation with analysis of finite coding length\cite{strassen,Hayashi4,Hay1,Pol}.}
with Gaussian approximation by using the above phase error correction formalism.
Scarani et al.\cite{SR08}
and Sano et al. \cite{MU10}
also treated the finiteness problem only for collective attack.
Recently, using Renner's formalism,
Tomamichel et al \cite{TWGR} derived an upper bound formula for the security parameter with the finite coding length.
However, these results assume the single photon source.
Furrer et al.\cite{FFBLSTW} gave a finite-length analysis with 
continuous variable quantum key distribution, which works with weak coherent pulses.
While continuous variable quantum key distribution can be implemented with an inexpensive Homodyne detection,
the decoy method with BB84 protocol can achieve the longest distance 
with the current technology\cite{Sasa,Stucki}.
%Since BB84 protocol with the decoy method realized longer distance than continuous variable quantum key distribution among existing experiments\cite{FDDVTG},
Hence, we treat the security of finite coding length of BB84 protocol
when we use weak coherent pulses and the decoy method.

In the single photon case,
using the phase error correction formalism,
another previous paper \cite{finite} derived better upper bound formulas 
for the security with the finite coding length,
which attain the key generation rate given in \cite{Hayashi3} up to the order $\sqrt{n}$.
They also treated the security with the universal composability 
based on the phase error correction formalism
when the coding length depends on the outcomes of Alice and Bob.
The phase error correction formalism provides an upper bound of the leaked information only from 
the decoding phase error probability.
Hence, we employ the phase error correction formalism
for our security analysis of finite coding length of BB84 protocol
when we use weak coherent pulses and the decoy method.

\subsection{Our formula for sacrifice bit-length with the finite-length setting}
When the raw keys are generated by BB84 protocol with the weak coherent pulses by the decoy method, 
we apply 
the error correction and the privacy amplification
to the raw keys.
The security of final keys can be evaluated by the amount of the sacrifice bit-length.
The aim of this paper is to provide 
a calculation formula of the sacrifice bit-length
guaranteeing a given security level 
with the universal composability.
Since the generated pulses contain the vacuum pulses, 
the single-photon pulses, and the multi-photon pulses,
we need to estimate these ratios among the pulses generating the raw keys.
Note that the vacuum pulses also generate a part of raw keys.
The flow of our analytical framework is illustrated as Fig. \ref{outline-zu2}.
First, %we recall the observation for the decoding phase error probabilities in imperfect photon source given in \cite{H2}
using the relation between phase error and the security,
we give a formula of the sacrifice bit-length based on
the numbers of the detected pulses originated from the vacuum emissions by Alice, 
the detected pulses from the single-photon emissions, and 
the detected pulses from the multi-photon emissions 
among the detected pulses consisting of the raw keys.
In the following, we call these numbers the {\it partition} of the detected pulses generating the raw keys.
When a component of the partition are divided by total pulse number, we obtain the fractions.
For the finite-length analysis, we need the partition instead of the fractions.

In order to estimate the partition of the detected pulses generating the raw keys,
we need to estimate the detection rates of respective kinds of pulses
and the phase error probability of single photon pulses,
which 
characterize Eve's operations and can be regarded as parameters of the quantum communication channel.
For this purpose, Alice sends the pluses with different intensities.
This method is called the decoy method, and 
enables us to estimate the above detection rates and the phase error probability of single photon pulses.
This estimation part can be divided into two parts.
The first part is the derivation of channel parameters from 
the detection rates, the phase error rates, 
and the partition of respective transmitted pulses by solving joint inequalities, which are given from non-negativity of several channel parameters. 
The second part is the treatment of statistical fluctuation.
If we could treat infinite number of pulses, 
we had not had to deal with the statistical fluctuation.
However, our finite-length setting requires the treatment of the statistical fluctuation.
In contrast with the previous papers \cite{finite,Hayashi3},
this paper deals with the statistical fluctuation by
{\it interval estimation}\footnote{Interval estimation is a statistical method to give an interval of possible (or probable) values of 
an unknown parameter from  sample data, in contrast to point estimation, which is a single number.
The method of the binomial case is explained in \ref{as2}.} and 
{\it percent point}\footnote{Precisely, the percent point means the lower percent point or the upper percent point
dependently of the context.
When we focus on the $\varepsilon$ percent, 
the lower percent point of the random variable $X$ is the value $x_1$ satisfying the following.
The probability that the random variable $X$ is less than $x_1$ is $\varepsilon/100$.
For example, the lower 5\% point of a standard normal distribution is -1.645.}. % in our formula.
The interval estimation is employed for deriving the detection rates and the phase error rates of transmitted pulses with respective intensities from the observed detection rates and the observed phase error rate.
The percent points are employed for deriving the partitions of transmitted pulses with respective intensities.
Similarly, we employ percent points for deriving the partition of the detected pulses generating the raw keys from the channel parameters.

\begin{figure}[h]
\begin{center}
\includegraphics[scale=0.6]{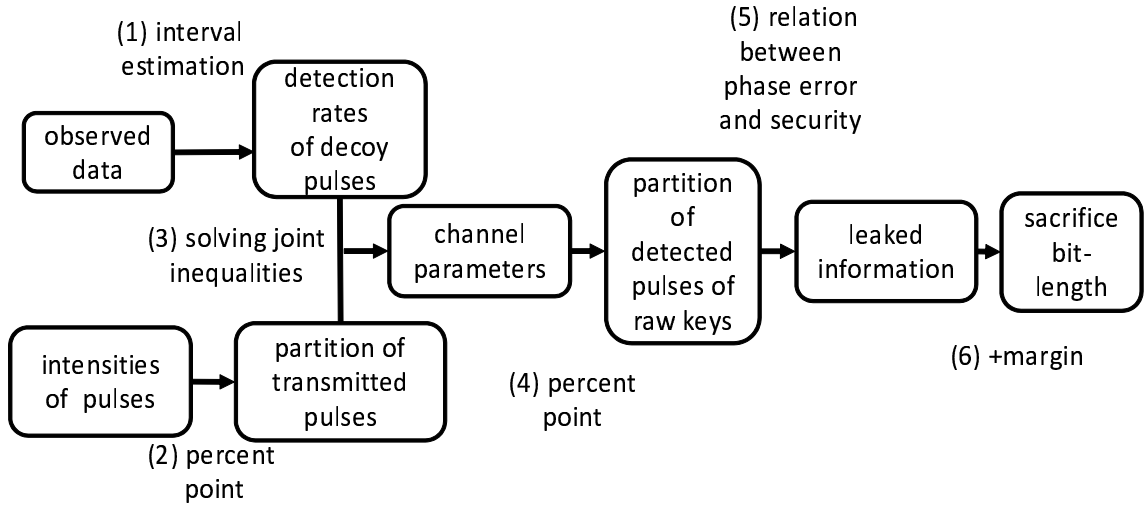}
\caption{Roles of percent points and interval estimation}
\Label{outline-zu2}
\end{center}
\end{figure}

In our analysis, we focus on the universal composability criterion.
Our calculation formula for the sacrifice bit-length employs only the basic formulas of
percent points and the interval estimation of the binomial distribution,
whose numerical calculations are possible by many computer software packages.
Hence, it does not contain any optimization process, and then 
it requires a relatively shorter calculation time.
Then, using our formula, we numerically calculate the key generation rate per pulses in several cases.
In our numerical calculations, 
we require that the universal composability criterion is less than $2^{-80}$.
Under this requirement, we have to require too small error probabilities 
%for the above percent point and the above interval estimation 
to calculate the exact percent point and the exact interval estimation.

For this purpose,
we employ Chernoff bound, which is an upper bound of the error probability and requires a quite small amount of calculations, as summarized in Appendix. 
Using Chernoff bound, we can derive upper and lower estimates of the true parameter.
Since Chernoff bound is not the tight bound of the error probability,
these upper and lower estimates are looser than the exact interval estimation.
However,
even when the required error probabilities are very small, 
when the size of obtained data is sufficiently large,
these upper and lower estimates are sufficiently close to the exact interval estimation\footnote{The reason is the following.
The rate of Chernoff bound to the true error probability behaves polynomially with respect to the size of data.
In particular, in the binary case, the rate behaves linearly with respect to the size of data.
Hence, even when the required error probabilities are very small, 
when the size of obtained data is sufficiently large,
these upper and lower estimates are sufficiently close to the exact interval estimation.}.

Further, similar to Wang et al.\cite{wang2,wang3},
in Section \ref{s12},
we discuss our key generation rate with finite-length
when 
the intensities are not fixed
and obey certain probability distributions.
In Subsection \ref{s13-2},
we numerically calculate the above key generation rate 
when the intensities obey Gaussian distributions
because the fluctuations of intensities are usually caused by the thermal noise.

Here, we summarize the physical assumption.
The photon source generates the coherent state, and
the phase factor of the coherent state is completely randomized.
The receiver uses the threshold detector.
We do not care about other types of imperfection of devices.
In particular,
we assume no side-channel-attack, i.e.,
Eve cannot directly see the phase modulator in Alice's side.
Further, we do not assume the perfect vacuum pulses.
That is, we allow that 
a non-vacuum state comes to be mixed in the vacuum pulses
if the probability of erroneous emission of a non-vacuum state
is sufficiently small.
We do not assume the collective attack
while we employ the binary distribution.
That is, our security proof well works for the coherent attack.
The reason why the binary distribution can be used instead of
the hypergeometric distribution is given in Section \ref{s6}.

\subsection{Organization}
The organization of the remaining part is the following.
As a preparation, Section \ref{s1} reviews the result for 
the universal composability criterion of the final keys 
when we know the partition of the received pulses and the phase error probability among single photon pulses.
Then, Section \ref{s1} derives
the leaked information from the partition of detected pulses of raw keys
by using the relation between the phase error and security, i.e., 
Step (5) in Fig. \ref{outline-zu2}. 
Section \ref{s2} describes a concrete protocol of the decoy method.
Section \ref{s3} explains how eavesdropper's operation can be described.
Section \ref{s5-5} gives two formulas of the sacrifice bit-length.
Subsection \ref{s-imp} gives a shorter sacrifice bit-length by improving 
the formula given in Subsection \ref{s-non-imp}.
So, we call the formula given in Subsection \ref{s-imp} the improved formula
and we call the formula given in Subsection \ref{s-non-imp} the non-improved formula.
%In the improved formula, we put out several probabilities from the square root.
Since the improved formula is too complicated,
we give the non-improved formula in Subsection \ref{s-non-imp}.
After describing the whole structure of the non-improved formula,
we give the improved formula in Subsection \ref{s-imp}.
Then, we present a numerical result with the improved formula.
In Section \ref{s12}, we treat the finite sacrifice bit-length
when the source intensity is not fixed.
Then, we present a numerical result with Gaussian distribution.

%In Section \ref{s13}, we give the several numerical calculations with the finite-length setting.

The remaining sections are devoted to the security proofs of the formulas given in
Section \ref{s5-5}.
For this purpose,
Section \ref{s6} summarizes a fundamental knowledge for random variables
because the notation explained in Section \ref{s6} will be used in latter sections.
Since the improved formula is too complicated,
we first show the security proof of the non-improved formula
in Sections \ref{s4}, \ref{s7}, and \ref{s9}.
After the security proof of the non-improved formula,
we give the security proof of the improved formula in Section \ref{s9-5}.
Section \ref{s4} briefly describes our security proof and the outline of 
discussions in the latter sections.
It also gives the sacrifice bit-length from the leaked information, i.e.,
Step (6) in Fig. \ref{outline-zu2}. 
Section \ref{s7} gives the estimate of channel parameters when 
the partition of the generated sources is given.
Subsection \ref{s7-1} gives the partition of detected pulses of raw keys
from the channel parameters by using percent points
i.e., Step (4) in Fig. \ref{outline-zu2}.
Subsection \ref{s5} estimates 
the channel parameters from the partitions and the detection rates of
several kinds of transmitted pulses
by solving joint inequalities,
i.e.,
Step (3) in Fig. \ref{outline-zu2},
and
Subsection \ref{s7-3} derives the detection rates of decoy pulses from the observed data
based on interval estimation, i.e., 
Step (1) in Fig. \ref{outline-zu2}.
In Section \ref{s9}, we treat statistical fluctuation of 
the photon number of the sources.
In particular, Subsection \ref{s93} gives 
the partitions of several kinds of transmitted pulses 
by using percent points, i.e.,
Step (2) in Fig. \ref{outline-zu2}.
Then, combining the discussions in Sections \ref{s4}, \ref{s7}, and \ref{s9},
we show the security under the sacrifice-length given by non-improved formula given in Subsection \ref{s-non-imp}.
In Section \ref{s9-5},
we give the security proof of the improved formula given in Subsection \ref{s-imp}
by putting out several probabilities from the square root.
%modifying the steps (1), (2) and (4) in Fig. \ref{outline-zu2}. 

In Appendices A and B, we summarize the basic knowledge for the tail probability and the interval estimation under the binary distribution.
In \ref{as3}, we summarize calculations required for the numerical calculation in Subsection \ref{s13-2}.

\section{Relation between security evaluation and decoding phase error probability}\Label{s1}
An evaluation method to use the trace norm of the difference 
between the true state and the ideal state is known as a universally composable security criterion in QKD \cite{uni2}.
Hence, we call it the universal composability criterion.
When the length $m$ of the final keys is not fixed,
we need a more careful treatment.
We denote the final state 
and Eve's final state
by $\rho_{AE|m}$ and $\rho_{E|m}$, respectively when the length of the final keys is $m$.
Our ideal Alice's state is the uniform distribution $\rho_{\mix|m}$ on $m$ bits.
Hence, the ideal composite state is $\rho_{\mix|m} \otimes \rho_{E|m}$.
We denote the state indicating that the length of final keys is $m$,
by $|m \rangle \langle m|$, and its probability by $P(m)$.
Then, the state of the composite system is
$\rho_{AE}:=\sum_{m}P(m) |m \rangle \langle m| \otimes \rho_{AE|m}$,
and its ideal state is
$\rho_{\ideal}:= \sum_{m}P(m) |m \rangle \langle m| \otimes \rho_{\mix|m} \otimes \rho_{E|m}$.
Hence, the averaged universal composability criterion of the obtained keys is written as the trace norm of the difference between the real state 
$\rho_{AE}$ of the composite system and its ideal state $\rho_{\ideal}$ as \cite{uni1}\footnote{The relation of this quantity with Eve's success probability is given in \ref{as4}}
\begin{eqnarray}
\| \rho_{A,E}-\rho_{\ideal} \|_1.
\end{eqnarray}
Thus, a smaller trace norm guarantees more secure final keys.

On the other hand,
when we apply surjective universal$_2$ linear hash functions as the privacy amplification
\cite{H2}, \cite[(10)]{finite} 
the above value is bounded by the averaged virtual decoding phase error probability $P_{ph}$ as
\begin{eqnarray}
\| \rho_{A,E}-\rho_{\ideal} \|_1
\le 2\sqrt{2}\sqrt{P_{ph}}
\Label{9-17-6}.
\end{eqnarray}
Then, the security analysis of QKD can be reduced to the evaluation of $P_{ph}$.

In the following, we consider the protocol 
containing the privacy amplification with the sacrifice bit-length
$S$ over the raw keys with length $M$.
When phase error occurs in $E$ bits among $M$-bit raw keys
and we apply the minimum length decoding,
the averaged virtual decoding phase error probability $P_{ph}$ is evaluated as\footnote{It is easy to see that Inequality 
(\ref{10-15-1}) 
holds when the completely
random matrices (a type of universal$_2$ hash functions) are used for PA, as in Koashi's case \cite{Koashi}.}
\begin{eqnarray}
P_{ph} \le 2^{M h(\min(\frac{E}{M},\frac{1}{2}) )-S}.
\end{eqnarray}
Hence,
we can guarantee the security of the final keys
when the sacrifice bit-length $S$ is sufficiently larger than
$M h(\min(\frac{E}{M},\frac{1}{2}))$. 
However, the number $E$ of bits having the phase error
does not take a deterministic value, and
it obeys a probability distribution $Q(E)$.
Then, when we apply the minimum length decoding,
the averaged virtual decoding phase error probability $P_{ph}$ is evaluated as
\begin{eqnarray}
P_{ph} \le \sum_E Q(E) \min (2^{M h(\min(\frac{E}{M},\frac{1}{2}))-S},1).
\end{eqnarray}

When we use an imperfect photon source, 
the $M$ transmitted pulses generate $M$-bit raw keys.
Then, each of the $M$ transmitted pulses takes the following three types of states. 
The first is the vacuum state,
the second is the single-photon state,
and the third is the multi-photon state.
In the following,
we assume that the $M$ transmitted pulses
consist of
$J^{(0)}$ pulses with the vacuum state,
$J^{(1)}$ pulses with the single-photon state,
and
$J^{(2)}$ pulses with the multi-photon state.
This assumption guarantees the relation $M=J^{(0)}+J^{(1)}+J^{(2)}$.
That is, the triplet $(J^{(0)},J^{(1)},J^{(2)})$ gives the partition of the $M$ transmitted pulses.
When we send the pulse with the vacuum state, no information can be leaked to Eve.
That is, the leaked information in this case equals the leaked information to Eve 
when we send single-photon pulses with phase error probability $0$.
On the other hand, in the multi-photon case,
we have to consider that all information is leaked to Eve.
Hence, the leaked information in the multi-photon case equals the leaked information to Eve 
when we send single-photon pulses with phase error probability $1/2$.
In the following, we assume that
the phase error occurs in $J^{(1)}_e$ bits among $J^{(1)}$ bits.
As is shown in \cite[(19)]{H2} and \cite{TH11},
when we apply a proper class of hash functions in the privacy amplification\footnote
{More precisely,
when we apply $\varepsilon$-almost dual universal$_2$ hash functions,
$P_{ph}$ is evaluated as
$P_{ph} \le \varepsilon \cdot 2^{\phi(J^{(0)},J^{(1)},J^{(1)}_e)-S}$.
As is explained in \cite{TH11}, 
several practical hash functions, e.g.,
the concatenation of 
Toeplitz matrix and the identity matrix, 
are $1$-almost dual universal$_2$.}, 
the averaged virtual decoding phase error probability $P_{ph}$ 
is evaluated as 
\footnote[1]{In the derivation\cite[(19)]{H2},
we considered that
the $J^{(1)}$ qubits have the phase error rate 
$\min(\frac{J^{(1)}_e}{J^{(1)}},\frac{1}{2})$
and 
the $J^{(2)}(=M-J^{(0)}-J^{(1)})$ qubits have the phase error rate $1/2$.}
\begin{eqnarray}
P_{ph} \le 2^{\phi(J^{(0)},J^{(1)},J^{(1)}_e)-S}\Label{10-15-1}
\end{eqnarray}
because $J^{(2)}=M-J^{(0)}-J^{(1)}$,
where we define
\begin{eqnarray}
\phi(J^{(0)},J^{(1)},J^{(1)}_e)
&:= J^{(1)} h(\min(\frac{J^{(1)}_e}{J^{(1)}},\frac{1}{2}))+(M-J^{(0)}-J^{(1)})h(\frac{1}{2}) 
\nonumber \\
&= J^{(1)} h(\min(\frac{J^{(1)}_e}{J^{(1)}},\frac{1}{2}))+(M-J^{(0)}-J^{(1)}),
\Label{1-23-5}
\end{eqnarray}
which provides Step (5) in Fig. \ref{outline-zu2}.
Due to Eq. (\ref{10-15-1}), we can regard $\phi(J^{(0)},J^{(1)},J^{(1)}_e)$ as a leaked information.

In the actual case,
the values $J^{(0)}$, $J^{(1)}$, and $J^{(1)}_e$ 
do not take deterministic values,
and obey a joint distribution $Q(J^{(0)},J^{(1)},J^{(1)}_e)$.
Hence, the averaged virtual decoding phase error probability $P_{ph}$ 
is evaluated by
\begin{eqnarray}
P_{ph} \le \sum_{J^{(0)},J^{(1)},J^{(1)}_e} Q(J^{(0)},J^{(1)},J^{(1)}_e) 
\min( 2^{\phi(J^{(0)},J^{(1)},J^{(1)}_e)-S},1).
\end{eqnarray}
In the general case,
the size of sacrifice bit-length $S$ 
also does not take a deterministic value,
and is stochastically determined.
In such a case, the values $J^{(0)}$, $J^{(1)}$, $J^{(1)}_e$, and $S$ obey 
a joint distribution
$Q(J^{(0)},J^{(1)},J^{(1)}_e,S)$,
and the averaged virtual decoding phase error probability $P_{ph}$ is evaluated by
\begin{eqnarray}
P_{ph} \le \sum_{J^{(0)},J^{(1)},J^{(1)}_e,S} Q(J^{(0)},J^{(1)},J^{(1)}_e,S) 
\min( 2^{\phi(J^{(0)},J^{(1)},J^{(1)}_e)-S},1).
\Label{9-17-5}
\end{eqnarray}
In the following, for a simplicity, 
we employ the notations $\bm{J}=(J^{(0)},J^{(1)},J^{(1)}_e)$ and $\phi(\bm{J}):=\phi(J^{(0)},J^{(1)},J^{(1)}_e)$.

\section{Protocol of decoy method}\Label{s2}
In the following, 
we assume that $M_s$-bit raw keys are 
generated by $N_s$ signal pulses generated by an imperfect photon source.
Now, we assume that
there are $N_s^{(0)}$ vacuum state pulses and $N_s^{(1)}$ single-photon pulses
among $N_s$ transmitted pulses.
Then, the remaining $N_s^{(2)}=N_s-N_s^{(0)}-N_s^{(1)}$ pulses take multi-photon states.
In the following discussion, 
the partition of $N_s$ signal pulses is described by the triplet $(N_s^{(0)},N_s^{(1)},N_s^{(2)})$, 
and plays an important role.

Now, we prepare three parameters $\bar{q}^{(0)}$, $\bar{q}^{(1)}$, and 
$\bar{b}^{(1)}_{\times}$ as follows.
The parameter $\bar{q}^{(0)}$ is the detection rate in the vacuum pulse, i.e.,
the rate of the vacuum pulses detected in Bob's side
to the vacuum pulses transmitted from Alice's side.
The parameter $\bar{q}^{(1)}$
is the detection rate in the single-photon pulse, i.e.,
the rate of the single-photon pulses detected in Bob's side
to the single-photon pulses transmitted from Alice's side.
The parameter $\bar{b}^{(1)}_{\times}$ is 
%the detection ratio with phase error in the single-photon pulse, i.e.,
the rate of the single-photon pulses detected with phase error in Bob's side
%in which, the phase error occurs,
to the single-photon pulses transmitted from Alice's side.
We call the rate $\bar{b}^{(1)}_{\times}$ 
the phase-error detection rate in the single-photon pulse.
Then, the numbers $J^{(0)}$, $J^{(1)}$, and $J^{(1)}_e$ can be estimated as
\begin{eqnarray}
J^{(0)} \sim  N_s^{(0)} \bar{q}^{(0)}, ~
J^{(1)} \sim  N_s^{(1)} \bar{q}^{(1)}, ~
J^{(1)}_e \sim  N_s^{(1)} \bar{b}^{(1)}_{\times} .
\Label{9-17-12}
\end{eqnarray}
However, it is not easy to estimate the partition
of $N_s$ pulses, i.e., $(N_s^{(0)},N_s^{(1)},N_s^{(2)})$. 
Now, we consider the case when the $N_s$ $\mu_1$-intensity weak coherent pulses are transmitted.

Then, we obtain the expansion with respect to the photon-number states.
\begin{eqnarray}
\fl \sum_{n=0}^\infty e^{-\mu_1}\frac{\mu_1^n}{n!}
|n \rangle\langle n|=
e^{-\mu_1}|0\rangle\langle 0|
+e^{-\mu_1}\mu_1|1\rangle\langle 1| 
+e^{-\mu_1}\mu_1^2 \omega_2\rho_2 \Label{1-23-1},
\end{eqnarray}
where
\begin{eqnarray}
\rho_2 := 
\frac{1}{\omega_2}
\sum_{n=2}^\infty 
\frac{\mu_1^{n-2}}{n!}
|n \rangle\langle n| , \quad
\omega_2 :=\frac{1}{\mu_1^2}
(e^{\mu_1}-(1+\mu_1)) \Label{4-9-7}.
\end{eqnarray}
Then, the partition can be estimated as
\begin{eqnarray}
N_s^{(0)} \sim N_s e^{-\mu_1}, ~
N_s^{(1)} \sim N_s e^{-\mu_1}\mu_1.
\end{eqnarray}

Hence, it is needed to estimate 
the parameters $\bar{q}^{(0)}$, $\bar{q}^{(1)}$, and $\bar{b}^{(1)}_{\times}$.
For this purpose, we shuffle $\mu_1$-intensity coherent pulses and $\mu_2$-intensity coherent pulses.
This method is called the decoy method \cite{decoy1,decoy2,decoy3,H1,decoy4} \footnote[2]{In a wider sense, we can regard 
the check bits estimating the phase error probability as another kind of decoy state.} 
because $\mu_2$-intensity pulses 
work as ``decoy'' for estimating the parameters
$\bar{q}^{(0)}$, $\bar{q}^{(1)}$, and $\bar{b}^{(1)}_{\times}$.
Hence, 
the intensity $\mu_1$ to be used to generating the raw keys 
is called the signal pulse,
and the other intensity $\mu_2$ is called the decoy pulse.
In the following, we assume that $\mu_1< \mu_2$.
Then, the 
$\mu_2$-intensity coherent pulse
has the following expansion:
\begin{eqnarray}
\sum_{n=0}^\infty e^{-\mu_2}\frac{\mu_2^n}{n!}
|n \rangle\langle n|
=&
e^{-\mu_2}|0\rangle\langle 0|
+e^{-\mu_2}\mu_2|1\rangle\langle 1| \nonumber \\
& +e^{-\mu_2}\mu_2^2\omega_2\rho_2
+e^{-\mu_2}\mu_2^2(\mu_2-\mu_1)
\omega_3 \rho_3,
\Label{4-9-4}
\end{eqnarray}
where
\begin{eqnarray}
\rho_3 &:= 
\frac{1}{\omega_3}
\sum_{n=3}^\infty \frac{\mu_2^{n-2}-\mu_1^{n-2}}{(\mu_2-\mu_1)n!}
|n \rangle\langle n|\nonumber \\
\omega_3 &:=
\frac{1}{\mu_2^2}(e^{\mu_2}-(1+\mu_2+\frac{\mu_2^2}{2}))
-
\frac{1}{\mu_1^2}
(e^{\mu_1}-(1+\mu_1+\frac{\mu_1^2}{2})).\Label{4-9-8}
\end{eqnarray}
Using the difference between the coefficients in two expansions (\ref{1-23-1}) and (\ref{4-9-4}), 
we can estimate the detection rates
$\bar{q}^{(0)}$ and $\bar{q}^{(1)}$ by the way explain in Section \ref{s12}.

In this paper, we use the superscript numbers and the subscript numbers 
in the following rules.
The superscript expresses the kind of state, i.e.,
the superscripts $0$, $1$, $2$, and $3$ correspond
to $|0\rangle \langle 0|$, 
$|1\rangle \langle 1|$, 
$\rho_2$, and $\rho_3$, respectively.
%except for $J^{(2)}$. The superscript $2$ of $J^{(2)}$ corresponds only to the multi-photon state.
The subscript expresses the intensity except for
$\rho_2$, $\rho_3$, $\omega_2$, and $\omega_3$. 
That is,
the subscripts $0$, $1$, $2$, $3$ and $4$ correspond
to the vacuum pulse,
the $\mu_1$-intensity pulse,
the $\mu_2$-intensity pulse,
the $\mu_1$-intensity pulse with the phase error,
and
the $\mu_2$-intensity pulse with the phase error, respectively.

In the following, we give the detail of our protocol, 
in which, 
both $\mu_1$-intensity pulses with the bit basis
and $\mu_2$-intensity pulses with the bit basis
are used for generating the raw keys.
\begin{description}
\item[(1) Transmission:]
Alice (the sender) sends 
the pulses with the vacuum,
the 
$\mu_1$-intensity coherent pulses
and
the $\mu_2$-intensity coherent pulses,
randomly with a certain rate.
Here, she chooses the bit basis and the phase basis 
with the ratio $1-\lambda:\lambda$
among the $\mu_1$-intensity coherent pulses
and the $\mu_2$-intensity coherent pulses.

\item[(2) Detection:]
Bob (the receiver) chooses the bit basis and the phase basis 
with the ratio $1-\lambda:\lambda$
and measures the pulses in the received side.
Then, he records existence or non-existence of the detection,
his basis, and the measured bit.
For the detail, see Remark \ref{r-det}.

\item[(3) Verification of basis:]
Using the public channel,
Alice sends Bob all information with respect to the basis and the intensity 
for all pulses.
Using the public channel,
Bob informs Alice what pulses has the matched basis.
Then, as is illustrated in Table \ref{zu1},
they decide the numbers $N_0$, $N_1$, $N_2$, $N_{s,1}$ and $N_{s,2}$ as follows.
$N_0$ is the number of vacuum pulses,
$N_1$ is the number of $\mu_1$-intensity pulses with the phase basis in the both sides,
$N_2$ is the number of $\mu_2$-intensity pulses with the phase basis in the both sides,
$N_{s,1}$ is the number of $\mu_1$-intensity pulses with the bit basis in the both sides,
and 
$N_{s,2}$ is the number of $\mu_2$-intensity pulses with the bit basis in the both sides.

\item[(4) Parameter estimation:]
Alice and Bob announce all bit information with respect to 
$N_1+N_2$ pulses with the phase basis in the both sides.
Then, as is illustrated in Table \ref{zu2},
they decide the numbers $M_0$, $M_1$, $M_2$, $M_3$, $M_4$, 
$M_{s,1}$ and $M_{s,2}$ as follows.
$M_0$ is the number of vacuum pulses detected by Bob. 
For $i=1,2$,
$M_{i}$($M_{i+2}$) is
the number of $\mu_i$-intensity coherent pulses those 
are detected by Bob
and have the phase basis in the both sides and the agreement bit values (the disagreement bit values). 
(However, they will not use $M_4$.)
$M_{s,1}$ is
the number of $\mu_1$-intensity coherent pulses 
those are detected by Bob and have the bit basis in the both sides.
$M_{s,2}$ is the number of $\mu_2$-intensity coherent pulses 
those are detected by Bob and have the bit basis in the both sides.
\end{description}

In the following, we describe the key distillation protocol
for $M_{s,1}$-bit raw keys generated by the $\mu_1$-intensity coherent pulses.
The key distillation protocol
for $M_{s,2}$-bit raw keys generated by the $\mu_2$-intensity 
coherent pulses can be obtained 
when $N_{s,1}$ and $M_{s,1}$ are replaced by $N_{s,2}$ and $M_{s,2}$, respectively.

\begin{description}
\item[(5) Error correction:]
First, Alice and Bob choose a suitable $M_{s,1}$-bit classical code 
$C_1$ that can correct errors of the expected bit error rate $p_{+}$.
For decoding, they prepare a set $\{\bs_{[\bs]}^{(2)}\}_{[\bs] \in \F_2^{M_{s,1}}/C_1}$
of representatives for respective cosets $[\bs] \in \F_2^{M_{s,1}}/C_1$.
They also prepare another set $\{\bs_{[\bs]}^{(1)}\}_{[\bs] \in \F_2^{M_{s,1}}/C_1}$
of representatives
for respective cosets $[\bs] \in \F_2^{M_{s,1}}/C_1$.
Then, they exchange
their information $\F_2^{M_{s,1}}/C_2^{\perp}$.
Alice obtains
$\bx:=\bs-\bs_{[\bs]}^{(1)}$in $C_2^{\perp}$,
and Bob obtains $\bx':=\bs'-\bs_{[\bs]}^{(1)}-\bs_{[\bs'-\bs]}^{(2)}$ in $C_1$.

\item[(6) Privacy amplification:]
Using the method explained latter,
Alice and Bob define the sacrifice bit-length $S$ 
in the privacy amplification
from $N_{s,1},N_0,N_1,N_2$, $M_{s,1}, M_0,M_1,M_2,M_3$.
Then, they apply 
$\varepsilon$-almost dual universal$_2$ hash function 
from $C_1 \cong \F_2^{l}$ to $\F_2^{l-S}$\cite{TH11}.
Then, they obtain the final keys.

\item[(7) Error verification:]
Alice and Bob apply a suitable hash function to the final keys.
They exchange the exclusive OR between the above hash value and other prepared secret keys.
If the above exclusive OR agrees, 
their keys agree with a high probability\cite{FMC10,S91}.
\end{description}

\begin{table}[htb]
  \caption{Transmitted pluses}
\Label{zu1}
\begin{center}
  \begin{tabular}{|c|c|c|c|c|} \hline
    Alice's basis & Bob's basis & vacuum & $\mu_1$ & $\mu_2$ \\ \hline
       \protect{\multirow{2}{*}{bit basis}} & bit basis &  \multirow{4}{*}{$N_0$} & $N_{s,1}$ & $N_{s,2}$ \\ \cline{4-5} \cline{2-2}
    　　　　　　　　　　　& phase basis &  &  & \\ \cline{1-2} \cline{4-5} 
     \multirow{2}{*}{phase basis} & bit basis &  &  & \\ \cline{2-2} \cline{4-5}
                               & phase basis &  & $N_1$ & $N_2$\\ \hline
  \end{tabular}
\end{center}
\end{table}

\begin{table}[htb]
  \caption{Detected pluses}
\Label{zu2}
\begin{center}
  \begin{tabular}{|c|c |c|c|c|c|} \hline
    Alice's basis & \multicolumn{2}{|c|}{Bob's basis} & vacuum & $\mu_1$ & $\mu_2$ \\ \hline
    \multirow{2}{*}{bit basis} & \multicolumn{2}{|c|}{bit basis}   &  \multirow{5}{*}{$M_0$} & $M_{s,1}$   & $M_{s,2}$ \\ \cline{2-3}  \cline{5-6}
    　　　　　　　　　　　      & \multicolumn{2}{|c|}{phase basis}     &                         &       & \\  \cline{2-3}      \cline{5-6} \cline{1-1}    
     \multirow{3}{*}{phase basis} & \multicolumn{2}{|c|}{bit basis}    &                         &   　　& 　　　\\ \cline{2-3} \cline{5-6}
                               & \multirow{2}{*}{phase basis} & correct &                         & $M_1$ & $M_2$ \\ \cline{3-3} \cline{5-6}
                               &                           & incorrect &                         & $M_3$ & $M_4$\\ \hline
  \end{tabular}
\end{center}
\end{table}

%\begin{figure}[h]\begin{center}
%\includegraphics[scale=0.3]{clip060.eps}\caption{Transmitted pluses}\label{zu1}
%\end{center}\end{figure}

%\begin{figure}[h]\begin{center}
%\includegraphics[scale=0.3]{clip002.eps}\caption{Detected pluses}\label{zu2}
%\end{center}\end{figure}

In the error correction,
we lose more than $M_{s,1}h(p_{+})$ bits.
When we lose  
$\eta M_{s,1} h(p_{+})$ bits in the error correction,
the final key length is $M_{s,1}- \eta M_{s,1} h(p_{+})-S$.
In a realistic case, we choose $\eta$ to be $1.1$.
In the above protocol, it is possible to 
restrict 
the intensity to generate the raw keys to $\mu_1$ or $\mu_2$.
In this case, we restrict 
the intensity with the bit basis to $\mu_1$ or $\mu_2$.
When we restrict the intensity with the bit basis to $\mu_2$,
the numbers $N_{s,1}$ and $M_{s,1}$ become $0$.

In the following discussion, 
we denote 
the number of transmitted pulses for generation of raw keys,
the number of raw keys,
and the signal intensity 
by $N_s$, $M_s$, and $\mu_s$.
That is, when we discuss the security of final keys generated from raw keys with the intensity $\mu_i$,
the numbers $N_s$, $M_s$, and $\mu_s$ are chosen to be $N_{s,i}$, $M_{s,i}$, and $\mu_i$ for $i=1,2$.

\begin{rem}
In the above protocol, the raw keys are generated from the bit basis.
However, this assumption is not essential.
For example, our analysis can be applied to the case 
when the raw keys are generated from both bases as follows.
First, we replace Step {\bf (3)} by the following Step {\bf (3')}.
\rm

\begin{description}
\item[(3') Verification of basis:]
Using the public channel,
Alice sends Bob all information with respect to the basis and the intensity 
for all pulses.
Using the public channel,
Bob informs Alice what pulses has the matched basis.
Then, as is illustrated in Table \ref{zu1},
they decide the numbers $N_0$, $N_1'$, $N_2'$, $N_{s,1}'$ and $N_{s,2}'$ as follows.
$N_0'$ is the number of vacuum pulses,
$N_1'$ is the number of $\mu_1$-intensity pulses with the phase basis in the both sides,
$N_2'$ is the number of $\mu_2$-intensity pulses with the phase basis in the both sides,
$N_{s,1}'$ is the number of $\mu_1$-intensity pulses with the bit basis in the both sides,
and 
$N_{s,2}'$ is the number of $\mu_2$-intensity pulses with the bit basis in the both sides.

Then, we decide smaller numbers 
$N_1$, $N_2$, $N_{s,1}$, $N_{s,2}$ 
than
$N_1'$, $N_2'$, $N_{s,1}'$, $N_{s,2}'$, respectively.
Next, we randomly choose 
$N_1$, $N_2$, $N_{s,1}$, $N_{s,2}$ pulses
among
$N_1'$, $N_2'$, $N_{s,1}'$, $N_{s,2}'$ pulses, respectively. 
\end{description}

\it
After Step {\bf (7)}, 
we choose numbers 
$N_{s,1}$, $N_{s,2}$, $N_1$, $N_2$
to be $N_1'-N_1$, $N_2'-N_2$, $N_{s,1}'-N_{s,1}$, $N_{s,2}'-N_{s,2}$, 
respectively.
We apply Step {\bf (4)} and the following steps to the remaining 
$N_{s,1}+N_{s,2}+N_1+N_2$ pulses and $N_0$ vacuum pulses
with exchanging the roles of the bit and the phase bases. 
In this case, we may choose the classical error correcting code $C$
based on the observed error rate in Step {\bf (5)}.
\end{rem}

\begin{rem}\Label{r-det}
When the receiver uses the threshold detector,
in Step {\bf (2)} (Detection),
the receiver might detect the both events.
In this case,
%As is mentioned in \cite[Section II]{H2},
we use the following type detector \cite{koashi2}. 
\begin{description}
\item[Detector]
When the receiver detects the both events,
the receiver chooses $0$ as the bit value definitely.
\end{description}
In fact, since the encoding does not depend on the choice of the detector,
the formula (\ref{9-17-6}) holds with the averaged virtual decoding phase error probability $P_{ph}$ 
based on any Bob's virtual decoder employing any Bob's detector
when Bob's detection event does not depend on the choice of the basis.
Hence, our security analysis is still valid even in the above detector.
\end{rem}

\section{Description of Eve}\Label{s3}
In the following, we describe the strategy of Eve.
For this purpose, we treat only the vacuum pulses 
and the pulses with matched bases,
i.e., 
$N_0+N_1+N_2+N_s$ pulses given in Table \ref{zu1}.
We do not treat other kinds of pulses.
Eve cannot distinguish pulses with the intensities $\mu_1$ and $\mu_2$ perfectly.
Alternatively, 
we assume that Eve can choose her strategy depending on the number of photons
because she can distinguish the number of photons.
That is, Eve is assumed to be able to distinguish the states
$|0\rangle \langle 0|$, $|1\rangle \langle 1|$, $\rho_2$, and $\rho_3$.

We assume the following partition of pulses given in Table \ref{zu1}
as follows:
\begin{itemize}
\item
There are $N^{(0)}_1$ pulses with the vacuum state
and $N^{(1)}_1$ pulses with the single-photon state
among $N_1$ $\mu_1$-intensity pulses with the phase basis.
\item
There are $N^{(0)}_2$ pulses with the vacuum state,
$N^{(1)}_2$ pulses with the single-photon state,
and
$N^{(2)}_2$ pulses with the state $\rho_2$
among $N_2$ $\mu_2$-intensity pulses with the phase basis.

\item
There are $N_s^{(0)}$ pulses with the vacuum state
and $N_s^{(1)}$ pulses with the single-photon state
among $N_s$ $\mu_s$-intensity pulses with the bit basis.
\end{itemize}

For a simplicity, we employ the notations 
$\bm{N}_s:=({N^{(0)}_s},{N^{(1)}_s},{N^{(2)}_s})$,
$\bm{N}_1:=(N^{(0)}_1,N^{(1)}_1)$, $\bm{N}_2:=(N^{(0)}_2,N^{(1)}_2,N^{(2)}_2)$,
and 
$\vec{\bm{N}}:=(\bm{N}_1,\bm{N}_2)$.
In the above partition, there are
$N_0+N^{(0)}_1+N^{(0)}_2+N_s^{(0)}$ pulses with the vacuum state,
$N^{(1)}_1+N^{(1)}_2+N_s^{(1)}$ pulses with the single-photon state,
$N^{(2)}_1+N^{(2)}_2$ pulses with the state $\rho_2$ and the phase basis,
and 
$N^{(3)}_2$ pulses with the state $\rho_3$ and the phase basis,
where
$N^{(2)}_1:=N_1-N^{(0)}_1-N^{(1)}_1$ and 
$N^{(3)}_2:=N_2-N^{(0)}_2-N^{(1)}_2-N^{(2)}_2$.
Note that 
the average state with the bit basis is not the same as
the average state with the phase basis
in the case of the multi-photon state. 

Then, Eve is assumed to be able to control
the detection rates 
$\bar{q}^{(0)}$,
$\bar{q}^{(1)}$,
$\bar{q}^{(2)}_{\times}$,
and 
$\bar{q}^{(3)}_{\times}$
in Bob's side
among
$N_0+N^{(0)}_1+N^{(0)}_2+N_s^{(0)}$ vacuum pulses,
$N^{(1)}_1+N^{(1)}_2+N_s^{(1)}$ single-photon pulses,
$N^{(2)}_1+N^{(2)}_2$ pulses of the state $\rho_2$ with the phase basis,
and $N^{(3)}_2$ pulses of the state $\rho_3$ with the phase basis, respectively.
Similarly,
Eve is assumed to be able to control
the phase-error detection rates
$\bar{b}^{(1)}_{\times}$, 
$\bar{b}^{(2)}_{\times}$, 
and
$\bar{b}^{(3)}_{\times}$
in Bob's side
among
$N^{(1)}_1+N^{(1)}_2+N_s^{(1)}$ single-photon pulses,
$N^{(2)}_1+N^{(2)}_2$ pulses of the state $\rho_2$ with the phase basis,
and $N^{(3)}_2$ pulses of the state $\rho_3$ with the phase basis, respectively.
In the following discussion, we use the parameters 
$\bar{a}^{(1)}_{\times}:=\bar{q}^{(1)}-\bar{b}^{(1)}_{\times}$,
$\bar{a}^{(2)}_{\times}:=\bar{q}^{(2)}_{\times}-\bar{b}^{(2)}_{\times}$,
$\bar{a}^{(3)}_{\times}:=\bar{q}^{(3)}_{\times}-\bar{b}^{(3)}_{\times}$,
instead of
$\bar{q}^{(1)}$,
$\bar{q}^{(2)}_{\times}$,
$\bar{q}^{(3)}_{\times}$.
For a simplicity,
we employ the notations 
$\bar{\bm{a}}:=(\bar{a}^{(1)}_{\times},\bar{a}^{(2)}_{\times},\bar{a}^{(3)}_{\times})$
and
$\bar{\bm{b}}:=(\bar{b}^{(1)}_{\times},\bar{b}^{(2)}_{\times},\bar{b}^{(3)}_{\times})$.
Eve is also assumed to be able to control the parameters
$\bar{q}^{(0)}$, $\bar{\bm{a}}$
and
$\bar{\bm{b}}$
dependently on the partition of the total $N_0+N_1+N_2+N_s$ pulses.
Further, Eve is assumed to choose these values stochastically.
Hence, the joint distribution conditioned with $\vec{\bm{N}}$ and $\bm{N}_s$
can be written as
$Q_e(\bar{q}^{(0)}, \bar{\bm{a}},\bar{\bm{b}}|\vec{\bm{N}},\bm{N}_s)$.
Since our analysis depends only on $\vec{\bm{N}}$,
we use the conditional distribution
$Q_e(\bar{q}^{(0)}, \bar{\bm{a}},\bar{\bm{b}}|\vec{\bm{N}})
:=
\sum_{\bm{N}_s}
P_s(\bm{N}_s)
Q_e(\bar{q}^{(0)}, \bar{\bm{a}},\bar{\bm{b}}|\vec{\bm{N}},\bm{N}_s)$,
where
$P_s$ is the distribution of $\bm{N}_s$
and cannot be controlled by Eve.
%Thus, we have to analyze the security dependently of partitions $\vec{\bm{N}}$ of $N_0$ pulses, $N_1$ pulses, and $N_2$ pulses.

\section{Formulas of sacrifice bit-length}\Label{s5-5}
\subsection{Non-improved formula}\Label{s-non-imp}
The aim of this section is 
to give formulas of the sacrifice bit-length $S$ satisfying 
\begin{eqnarray}
\| \rho_{A,E}-\rho_{\ideal} \|_1 \le 2^{-\beta}\Label{35i}
\end{eqnarray}
as a function of $\beta,\mu_s,\mu_1,\mu_2,N_s,N_0,N_1,N_2$, and $\bm{M}$,
where $\rho_{A,E}$ is the final state and $\rho_{\ideal}$ is the ideal state.
This section gives two formulas, 
the non-improved formula and the improved formula.
While the improved formula gives a shorter sacrifice bit-length than the non-improved formula,
the non-improved formula is simpler than the improved formula.
Hence, we give the non-improved formula firstly.
In the next subsection, we give the improved formula.
For this purpose, we prepare fundamental definition for behavior of random variables.

\begin{definition}
When the random variable $k$ is subject to the distribution $P$,
we denote $k \sim P$.
When the true distribution is 
the $N$-trial binary distribution with success probability $p$, 
which is denoted by $Bin(N,p)$,
we denote 
the upper percent point with probability $\alpha$
by $X_{\per}^+(N,p,\alpha)$,
and 
denote the lower percent point with probability $\alpha$
by $X_{\per}^-(N,p,\alpha)$.
Then, we define
$p_{\per}^+(N,p,\alpha):=X_{\per}^+(N,p,\alpha)/N$,
and
$p_{\per}^-(N,p,\alpha):=X_{\per}^-(N,p,\alpha)/N$.
When we observe the value $k$ subject to the binomial distribution $Bin(N,p)$ with
$N$ trials  and probability $p$,
we denote the lower confidence limit of 
the lower one-sided interval estimation with the confidential level $1-\alpha$
by $p_{\est}^-(N,k,\alpha)$.
Similarly,
we denote 
the upper confidence limit of 
the upper one-sided interval estimation with the confidential level $1-\alpha$
by $p_{\est}^+(N,k,\alpha)$.
Then, we define
$X_{\est}^-(N,k,\alpha):=p_{\est}^-(N,k,\alpha)N$,
and
$X_{\est}^+(N,k,\alpha):=p_{\est}^+(N,k,\alpha)N$.
\end{definition}

When $N$ is not so large (e.g., 10,000) or $\alpha$ is not so small (e.g., 0.001),
the percent point $X_{\per}^{\pm}(N,p,\alpha)$
can be calculated by mathematical package in software (e.g., Mathematica).
As is summarized in Appendix B.1, 
the interval estimation $p_{\est}^{\pm}(N,k,\alpha)$ is described by F distribution, and 
can be calculated by mathematical package in software in this case, similarly.
However, when $N$ is too large and $\alpha$ is too small,
these calculation cannot be done by a usual mathematical package in software.
However, since $N$ is large enough, using formulas given in Appendices A and B,
we can calculate good lower and upper bounds of these values, which is enough close to the exact values for our purpose.
The calculation formulas can be implemented with small calculation amounts.

Indeed, in order to guarantee the unconditional security,
we have to use the hypergeometric distribution instead of the binomial distribution.
However, the hypergeometric distribution can be partially replaced by the binomial distribution.
Section \ref{s6} explains which case allows this replacement.
This replacement greatly simplifies the calculation of sacrifice bit-length.

\begin{figure}[h]
\begin{center}
\includegraphics[scale=0.6]{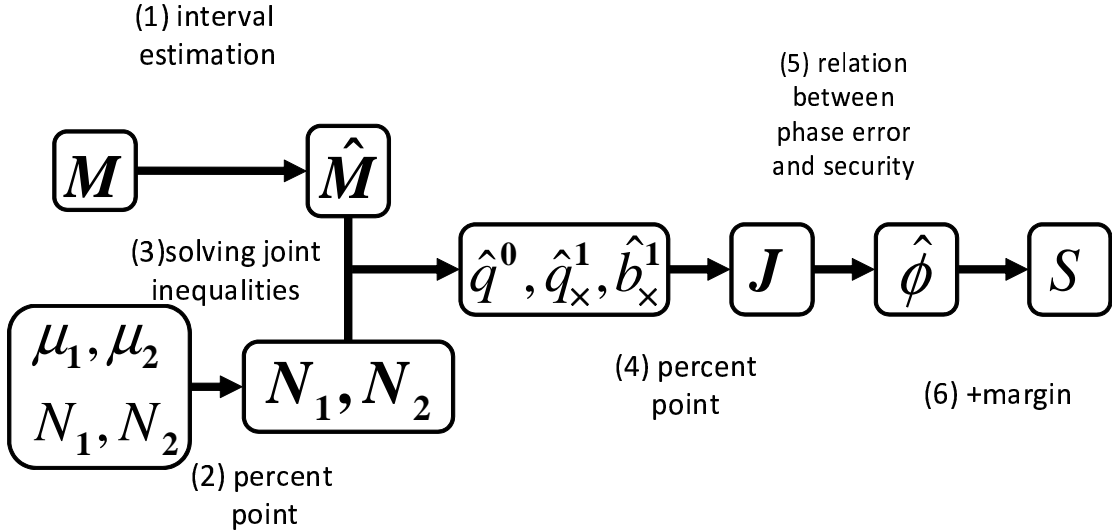}
\caption{Outline of our derivation of the sacrifice bit-length $S$.
%$\hat{\bm{M}}=(\hat{M}_0,\hat{M}_1,\hat{M}_2,\hat{M}_3)$ 
%is the estimate of channel parameter $\bar{\bm{M}}=(\bar{M}_0,\bar{M}_1,\bar{M}_2,\bar{M}_3)$ given in Subsection \ref{s5}.
}
%Step (1) is given in (\ref{10-4-0}), (\ref{10-4-1}), (\ref{10-4-2}), and (\ref{10-4-3}).
%Step (2) is given in (\ref{3-28-1}), (\ref{3-28-2}), (\ref{3-28-3}), (\ref{3-28-4}), and (\ref{3-28-5}).
%Step (3) is given in (\ref{1-23-4}), (\ref{1-23-4-2}), and (\ref{10-8-20}).
%Step (4) is given in (\ref{1-24-1}), (\ref{1-24-2}), and (\ref{1-24-3}).
%Step (5) is given in (\ref{10-26-1}) and (\ref{10-4-a}).
%Step (6) is given in (\ref{3-28-6}).}
\Label{outline-zu}
\end{center}
\end{figure}

Now, we give the non-improved formula of the sacrifice bit-length $S$ 
as a function of $\beta,\mu_s,\mu_1,\mu_2,N_s,N_0,N_1,N_2$, and $\bm{M}=(M_s,M_0,M_1,M_2,M_3)$.
The whole structure of our formulas is summarized as Fig. \ref{outline-zu}.
%Here, we give the sacrifice bit-length $S$ as a function of $\beta,\mu_1,\mu_2,N_s,N_0,N_1,N_2$, $M_s$ and $\bm{M}=(M_0,M_1,M_2,M_3)$ in the following way.
Then, as is shown latter,
when the sacrifice bit-length is given by the following way, the final key satisfies (\ref{35i}).
\begin{description}
\item[Step (1)]
We estimate the detection rates of decoy pulses from the observed data based on interval estimation:
\begin{eqnarray}
\hat{M}_0 &:=X_{\est}^-(N_0,M_0,2^{-2\beta-8}),\Label{10-4-0i} \\
\hat{M}_1 &:=X_{\est}^-(N_1,M_1,2^{-2\beta-8}),\Label{10-4-1i} \\
\hat{M}_2 &:=X_{\est}^+(N_2,M_2,2^{-2\beta-8}),\Label{10-4-2i} \\
\hat{M}_3 &:=X_{\est}^+(N_1,M_3,2^{-2\beta-8}).\Label{10-4-3i} 
\end{eqnarray}
%in (\ref{10-4-0}), (\ref{10-4-1}), (\ref{10-4-2}), and (\ref{10-4-3}).
\item[Step (2)]
We estimate the partitions of several kinds of transmitted pulses 
by using percent points:
\begin{eqnarray}
\hat{N}^{(0)}_1&:=X_{\per}^{-}(N_1,e^{-\mu_1},2^{-2\beta-8})\Label{3-28-1i}\\
\hat{N}^{(1)}_1&:=X_{\per}^{-}(N_1,\mu_1 e^{-\mu_1},2^{-2\beta-8})\Label{3-28-2i}\\
\hat{N}^{(0)}_2&:=X_{\per}^{-}(N_2,e^{-\mu_2},2^{-2\beta-8})\Label{3-28-3i}\\
\hat{N}^{(1)}_2&:=X_{\per}^{-}(N_2,\mu_2 e^{-\mu_2},2^{-2\beta-8})\Label{3-28-4i}\\
\hat{N}^{(2)}_2&:=X_{\per}^{-}(N_2,\omega_2 \mu_2^2 e^{-\mu_2},2^{-2\beta-8}) .\Label{3-28-5i}
\end{eqnarray}
%Step (2) is given in (\ref{3-28-1}), (\ref{3-28-2}), (\ref{3-28-3}), (\ref{3-28-4}), and (\ref{3-28-5}).
\item[Step (3)]
We estimate the channel parameters from the partitions and the detection rates 
of several kinds of transmitted pulses by solving joint inequalities:
\begin{eqnarray}
\fl \hat{q}^{(0)} (\hat{M}_0)
:=& \frac{\hat{M}_0}{N_0} \\
\fl \hat{a}^{(1)}_{\times}%(\vec{\hat{\bm{N}}})
(\hat{\bm{M}},\vec{\hat{\bm{N}}})
:=&
\left[
\frac{ N^{(2)}_2(\hat{M}_1 - \hat{q}^{(0)} (\hat{M}_0) \hat{N}^{(0)}_1 / 2 )
-\hat{N}^{(2)}_1 (\hat{M}_2 - \hat{q}^{(0)} (\hat{M}_0) \hat{N}^{(0)}_2 / 2 )
}{\hat{N}^{(1)}_1 \hat{N}^{(2)}_2- \hat{N}^{(1)}_2 \hat{N}^{(2)}_1} 
\right]_+
\Label{1-23-4-2i} \\
\fl \hat{b}^{(1)}_{\times} %(\vec{\hat{\bm{N}}})
(\hat{\bm{M}},\vec{\hat{\bm{N}}}) 
:=& 
\left[
\frac{\hat{M}_3 - \frac{1}{2} \hat{q}^{(0)} (\hat{M}_0) \hat{N}^{(0)}_1}{\hat{N}^{(1)}_1}
\right]_+
\Label{10-8-20i},
\end{eqnarray}
where $[x]_+:= \max (x,0).$
%Step (3) is given in (\ref{1-23-4}), (\ref{1-23-4-2}), and (\ref{10-8-20}).
\item[Step (4)]
We estimate the partition of detected pulses of raw keys from the channel parameters by using percent points:
\begin{eqnarray}
\fl \hat{J}^{(0)} %(\hat{q}^{(0)},N_s,\mu_s) 
&:= X_{\per}^-(N_s,e^{-\mu_s} \hat{q}^{(0)},2^{-2\beta-8}) \Label{1-24-1i}\\
\fl \hat{J}^{(1)} %(\hat{a}^{(1)}_{\times},\hat{b}^{(1)}_{\times},N_s,\mu_s) 
&:= X_{\per}^-(N_s,e^{-\mu_s} \mu_s (\hat{a}^{(1)}_{\times}+\hat{b}^{(1)}_{\times}),2^{-2\beta-8}) \Label{1-24-2i} \\
\fl \hat{r}^{(1)}_{\times} %(\hat{a}^{(1)}_{\times},\hat{b}^{(1)}_{\times},N_s,\mu_s)  
&:= p_{\per}^+\Bigl(\hat{J}^{(1)} %(\hat{a}^{(1)}_{\times},\hat{b}^{(1)}_{\times},N_s,\mu_s) 
, 
\frac{\hat{b}^{(1)}_{\times}}{\hat{a}^{(1)}_{\times}+\hat{b}^{(1)}_{\times}}
, 2^{-2\beta-8}\Bigr).
\end{eqnarray}
%Step (4) is given in (\ref{1-24-1}), (\ref{1-24-2}), and (\ref{1-24-3}).
\item[Step (5)]
We estimate the leaked information from the partition of detected pulses of raw keys by using the relation between the phase error and the security:
\begin{eqnarray}
\fl \hat{\phi}_2
%(\hat{\bm{M}}(\bm{M}),\vec{\bm{N}})\nonumber \\
:=  M_s-\hat{J}^{(0)}
%(\hat{q}^{(0)},N_s,\mu_s) 
- \hat{J}^{(1)}
%(\hat{a}^{(1)}_{\times},\hat{b}^{(1)}_{\times},N_s,\mu_s) 
(1- h(\min \{\hat{r}^{(1)}_{\times}
%(\hat{a}^{(1)}_{\times},\hat{b}^{(1)}_{\times},N_s,\mu_s) 
,1/2\}) ) .%\nonumber \\
%\fl & 
\Label{10-26-1i}
\end{eqnarray}
%Step (5) is given in (\ref{10-26-1}) and (\ref{10-4-a}).
\item[Step (6)]
We give the sacrifice bit-length from the leaked information:
\begin{eqnarray}
\fl S
 &:=
\left\{
\begin{array}{ll}
\hat{\phi}_2 %(\hat{\bm{M}}(\bm{M}),\vec{\hat{\bm{N}}}) 
+2\beta+ 5& \hbox{if 
%$\hat{\bm{M}}=(\hat{M}_0,\hat{M}_1,\hat{M}_2,\hat{M}_3)$ 
Conditions \ref{c6}, \ref{c5} and \ref{c15} below hold.} \\
\dim C_1
& \hbox{otherwise.}
\end{array}
\right.
\Label{3-28-6i}
\end{eqnarray}
That is, when one of
Conditions \ref{c6}, \ref{c5} and \ref{c15} does not hold,
we abort the protocol.
%Step (6) is given in (\ref{3-28-6}).
\end{description}

Conditions \ref{c6}, \ref{c5}, and \ref{c15} are given as follows.
In order to give these conditions,
we define the set $\Omega_1$ as the set of $\vec{\bm{N}}$ satisfying 
\begin{eqnarray}
%{N}^{(0)}&\in [X_{\per}^{-}(N,e^{-\mu_1},2^{-2\beta-8}),X_{\per}^{+}(N,e^{-\mu_1},2^{-2\beta-8})]
%\Label{10-5-10-d2}
%\\
%{N}^{(1)}&\in [X_{\per}^{-}(N,\mu_1 e^{-\mu_1},2^{-2\beta-8}), X_{\per}^{+}(N,\mu_1 e^{-\mu_1},2^{-2\beta-8})]
%\Label{10-5-11-d2}
%\\
{N}^{(0)}_1&\in [X_{\per}^{-}(N_1,e^{-\mu_1},2^{-2\beta-8}),X_{\per}^{+}(N_1,e^{-\mu_1},2^{-2\beta-8})] 
\Label{11-12-1}\\
{N}^{(1)}_1&\in [X_{\per}^{-}(N_1,\mu_1 e^{-\mu_1},2^{-2\beta-8}),X_{\per}^{+}(N_1,\mu_1 e^{-\mu_1},2^{-2\beta-8})] 
\Label{11-12-2}\\
{N}^{(0)}_2&\in [X_{\per}^{-}(N_2,e^{-\mu_2},2^{-2\beta-8}),X_{\per}^{+}(N_2,e^{-\mu_2},2^{-2\beta-8})] 
\Label{11-12-3}\\
{N}^{(1)}_2&\in [X_{\per}^{-}(N_2,\mu_2 e^{-\mu_2},2^{-2\beta-8}),X_{\per}^{+}(N_2,\mu_2 e^{-\mu_2},2^{-2\beta-8})] 
\Label{11-12-4}\\
{N}^{(2)}_2&\in [X_{\per}^{-}(N_2,\omega_2 \mu_2^2 e^{-\mu_2},2^{-2\beta-8}), X_{\per}^{+}(N_2,\omega_2 \mu_2^2 e^{-\mu_2},2^{-2\beta-8})] .
\Label{11-12-5}
\end{eqnarray}

\begin{condition}\Label{c6}
Any element $\vec{\bm{N}}\in\Omega_1$ satisfies
\begin{eqnarray*}
\fl & X_{\per}^{-}(N_1,\mu_1 e^{-\mu_1},2^{-2\beta-8})
X_{\per}^{-}(N_2,\omega_2 \mu_2^2 e^{-\mu_2},2^{-2\beta-8}) \\
\fl >& 
(N_1- X_{\per}^{-}(N_1,e^{-\mu_1},2^{-2\beta-8}) -X_{\per}^{-}(N_1,\mu_1 e^{-\mu_1},2^{-2\beta-8}))
X_{\per}^{+}(N_2,\mu_2 e^{-\mu_2},2^{-2\beta-8}), \\
\fl & X_{\per}^{-}(N_2,\omega_2 \mu_2^2 e^{-\mu_2},2^{-2\beta-8})
X_{\per}^{-}(N_1,e^{-\mu_1},2^{-2\beta-8}) \\
\fl >& 
(N_1- X_{\per}^{-}(N_1,e^{-\mu_1},2^{-2\beta-8}) -X_{\per}^{-}(N_1,\mu_1 e^{-\mu_1},2^{-2\beta-8}))
X_{\per}^{+}(N_2,e^{-\mu_2},2^{-2\beta-8}), \\
\fl &
\frac{X_{\per}^{-}(N_2,\omega_2 \mu_2^2 e^{-\mu_2},2^{-2\beta-8})}{N_1- X_{\per}^{-}(N_1,e^{-\mu_1},2^{-2\beta-8}) -X_{\per}^{-}(N_1,\mu_1 e^{-\mu_1},2^{-2\beta-8})}
+
\frac{X_{\per}^{-}(N_2,e^{-\mu_2},2^{-2\beta-8})}{X_{\per}^{+}(N_1,e^{-\mu_1},2^{-2\beta-8})} \\
\fl  >&
\frac{2 X_{\per}^{+}(N_2,\mu_2 e^{-\mu_2},2^{-2\beta-8})}{X_{\per}^{-}(N_1,\mu_1 e^{-\mu_1},2^{-2\beta-8})}.
\end{eqnarray*}
\end{condition}

\begin{condition}\Label{c5}
For any $\vec{\bm{N}} \in \Omega_1$, 
all of the following values are positive.
\begin{eqnarray*}
\fl A^{(0)}_1&:=
{ \hat{M}_2 - \hat{q}^{(0)} (N^{(0)}_2 +N^{(2)}_2)/ 2}
- N_2^{(1)} \frac{ N^{(2)}_2(\hat{M}_1 - \hat{q}^{(0)} N^{(0)}_1 / 2)
-N^{(2)}_1 (\hat{M}_2 - \hat{q}^{(0)} N^{(0)}_2 / 2)
}{N^{(1)}_1 N^{(2)}_2- N^{(1)}_2 N^{(2)}_1} ,\\
\fl A^{(1)}_1&:=
{ \hat{M}_2 - \hat{q}^{(0)} N^{(0)}_2 / 2}
- (N_2^{(2)}+N_2^{(1)}) \frac{ N^{(2)}_2(\hat{M}_1 - \hat{q}^{(0)} N^{(0)}_1 / 2)
-N^{(2)}_1 (\hat{M}_2 - \hat{q}^{(0)} N^{(0)}_2 / 2)
}{N^{(1)}_1 N^{(2)}_2- N^{(1)}_2 N^{(2)}_1} , \\
\fl A^{(1)}_2&:=
\frac{
N^{(2)}_1(
N^{(2)}_2
(\hat{M}_1 - \frac{\hat{q}^{(0)}}{2} N^{(0)}_1)
-
N^{(2)}_1 
(\hat{M}_2 - \frac{\hat{q}^{(0)}}{2} N^{(0)}_2) 
)
}{N^{(1)}_1 N^{(2)}_2- N^{(1)}_2 N^{(2)}_1}, \\
\fl A^{(2)}_2&:=
%\fl =&
{ \hat{M}_1 - \frac{\hat{q}^{(0)}}{2} N^{(0)}_1}
-
\frac{
N^{(1)}_1(
N^{(2)}_2
(\hat{M}_1 - \frac{\hat{q}^{(0)}}{2} N^{(0)}_1)
-
N^{(2)}_1 
(\hat{M}_2 - \frac{\hat{q}^{(0)}}{2} N^{(0)}_2) 
)
}{N^{(1)}_1 N^{(2)}_2- N^{(1)}_2 N^{(2)}_1} \\
\fl B^{(1)}_1 &:=
{ \hat{M}_3 - \hat{q}^{(0)} N^{(0)}_1 /2}.
\end{eqnarray*}
\end{condition}

\begin{condition}\Label{c15}
Any element $\vec{\bm{N}} \in \Omega_1$ satisfies
\begin{eqnarray}
\frac{\hat{b}^{(1)}_{\times} (\hat{\bm{M}},%(\bm{M}),
\vec{\bm{N}})
}{\hat{a}^{(1)}_{\times} (\hat{\bm{M}},%(\bm{M}),
\vec{\bm{N}})
+\hat{b}^{(1)}_{\times} (\hat{\bm{M}},%(\bm{M}),
\vec{\bm{N}})
}
\le \frac{1}{8}.
\end{eqnarray}
\end{condition}

\begin{rem}[Adjustment of $\hat{q}^{(0)}$ for non-improved formula]
When the vacuum pulse has a possibility to contain a non-vacuum state,
we cannot apply the above formula $\hat{q}^{(0)}$.
Hence, we need its adjustment.
Assume that 
the vacuum pulse becomes a non-vacuum state with a probability $q$.
In this case, we replace $\hat{q}^{(0)}$ by 
\begin{eqnarray}
\fl
\hat{q}^{(0)}(M_0):=
p_{\est}^+(N_0-X_{\per}^+(N_0, q, 2^{-2\beta-8}), M_0-X_{\per}^+(N_0, q, 2^{-2\beta-8}), 2^{-2\beta-8}).
\Label{5-1-2i}
\end{eqnarray}
\end{rem}

Here, we should remark that 
Condition \ref{c6} is given for 
the initial parameters $\beta,\mu_1,\mu_2,N_1,N_2$
while
Conditions \ref{c5} and \ref{c15}
are given for 
the observed values $\bm{M}=(M_s,M_0,M_1,M_2,M_3)$
as well as 
the initial parameters $\beta,\mu_1,\mu_2,N_0,N_1,N_2$.
Hence, it is required to choose 
the initial parameters $\beta,\mu_1,\mu_2,N_1,N_2$ satisfying Condition \ref{c6}.
Further, we need to choose the initial parameters $\beta,\mu_1,\mu_2,N_0,N_1,N_2$
so that
Conditions \ref{c5} and \ref{c15} hold with high probability.
% $M_s$ and $\bm{M}=(M_0,M_1,M_2,M_3)$

Now, we consider the case when there might exist an eavesdropper.
In this case, even if we choose $\mu_1$, $\mu_2$, $N_0$, $N_1$, $N_2$ suitably,
the eavesdropper might control the channel parameters
$\bar{q}^{(0)}$, $\bar{\bm{a}}$ and $\bar{\bm{b}}$ 
so that
Conditions \ref{c5} and \ref{c15} do not hold.
Hence, we need to prepare a method to smoothly decide whether Conditions \ref{c5} and \ref{c15} hold. 

We will show that 
the non-improved formula satisfies the condition (\ref{35i})
in Sections \ref{s4}, \ref{s7}, and \ref{s9}.
The following table (Table \ref{outline-zu-2}) explains 
which equations in Sections \ref{s4}, \ref{s7}, and \ref{s9}
correspond to the above steps in the non-improved formula.

\begin{table}[htb]
\caption{Detail descriptions of respective steps}
\Label{outline-zu-2}
\begin{center}
  \begin{tabular}{|l|l|l|} \hline
    Step & Subsection & Equation  \\ \hline
Step (1) & \ref{s7-3}&(\ref{10-4-0}), (\ref{10-4-1}), (\ref{10-4-2}), (\ref{10-4-3}) 
\\ \hline
Step (2) & \ref{s93} &(\ref{3-28-1}), (\ref{3-28-2}), (\ref{3-28-3}), (\ref{3-28-4}), (\ref{3-28-5}) 
\\ \hline
Step (3) & \ref{s5} &(\ref{1-23-4}), (\ref{1-23-4-2}), (\ref{10-8-20}) \\ \hline
Step (4) & \ref{s7-1} &(\ref{1-24-1}), (\ref{1-24-2}), (\ref{1-24-3}) 
\\ \hline
Step (5) & \ref{s5} &(\ref{10-26-1})%, (\ref{10-4-a})
 \\ \hline
Step (6) & \ref{s93} &(\ref{3-28-6}) 
\\ \hline
Adjustment & \ref{s7-3} &(\ref{5-1-2}) (Remark \ref{r5-15})
\\ \hline
    \end{tabular}
\end{center}
\end{table}

%\subsection{Adjustment of $\hat{q}^{(0)}$ for non-improved formula}\Label{adjust}

\subsection{Improved formula}\Label{s-imp}
However, the above construction is too restrictive. 
% for realizing the condition (\ref{35i}).
We can replace Steps {\bf (1)}, {\bf (2)}, {\bf (4)}, Condition \ref{c6}, 
and the definition of the set $\Omega_1$
as follows.
That is, Conditions \ref{c5} and \ref{c15} are replaced by the conditions based on 
the improved version of $\Omega$.
The formula given here for the sacrifice bit-length is called
the improved formula.
Section \ref{s9-5} explains why the improvement is possible. 
That is, Section \ref{s9-5} shows that the improved formula also guarantees the condition (\ref{35i}).

\begin{description}
\item[Step (1)]
%(\ref{10-4-0-e}), (\ref{10-4-1-e}), (\ref{10-4-2-e}), (\ref{10-4-3-e}) 
We replace the estimated detection rates of decoy pulses by the following way:
\begin{eqnarray}
\hat{M}_0 &:=X_{\est}^-(N_0,M_0,2^{-\beta-6}),\Label{10-4-0-e} \\
\hat{M}_1 &:=X_{\est}^-(N_1,M_1,2^{-2\beta-7}),\Label{10-4-1-e} \\
\hat{M}_2 &:=X_{\est}^+(N_2,M_2,2^{-2\beta-7}),\Label{10-4-2-e} \\
\hat{M}_3 &:=X_{\est}^+(N_1,M_3,2^{-2\beta-7}).\Label{10-4-3-e} 
\end{eqnarray}
\item[Step (2)]
We replace the estimated partitions of several kinds of transmitted pulses 
by the following way:
%&(\ref{3-28-10}), (\ref{3-28-11}), (\ref{3-28-12}), (\ref{3-28-13}), (\ref{3-28-14})
\begin{eqnarray}
\hat{N}^{(0)}_1&:=X_{\per}^{-}(N_1,e^{-\mu_1},2^{-\beta-6})\Label{3-28-10}\\
\hat{N}^{(1)}_1&:=X_{\per}^{-}(N_1,\mu_1 e^{-\mu_1},2^{-\beta-6})\Label{3-28-11}\\
\hat{N}^{(0)}_2&:=X_{\per}^{-}(N_2,e^{-\mu_2},2^{-\beta-6})\Label{3-28-12}\\
\hat{N}^{(1)}_2&:=X_{\per}^{-}(N_2,\mu_2 e^{-\mu_2},2^{-\beta-6})\Label{3-28-13}\\
\hat{N}^{(2)}_2&:=X_{\per}^{-}(N_2,\omega_2 \mu_2^2 e^{-\mu_2},2^{-\beta-6}) .\Label{3-28-14}
\end{eqnarray}
%Step (2) is given in (\ref{3-28-1}), (\ref{3-28-2}), (\ref{3-28-3}), (\ref{3-28-4}), and (\ref{3-28-5}).
\item[Step (4)]
We replace the estimated partition of detected pulses 
and the estimated phase error rate of the single photon
of raw keys by the following way:
%& (\ref{1-24-1-e}), (\ref{1-24-2-e}),  (\ref{1-24-3-e}) 
\begin{eqnarray}
\fl \hat{J}^{(0)} %(\hat{q}^{(0)},N_s,\mu_s) 
&:= X_{\per}^-(N_s,e^{-\mu_s} \hat{q}^{(0)},2^{-\beta-6}) \Label{1-24-1-e}\\
\fl \hat{J}^{(1)} %(\hat{a}^{(1)}_{\times},\hat{b}^{(1)}_{\times},N_s,\mu_s) 
&:= X_{\per}^-(N_s,e^{-\mu_s} \mu_s (\hat{a}^{(1)}_{\times}+\hat{b}^{(1)}_{\times}),2^{-\beta-6}) \Label{1-24-2-e} \\
\fl \hat{r}^{(1)}_{\times} %(\hat{a}^{(1)}_{\times},\hat{b}^{(1)}_{\times},N_s,\mu_s)  
&:= p_{\per}^+\Bigl(\hat{J}^{(1)} %(\hat{a}^{(1)}_{\times},\hat{b}^{(1)}_{\times},N_s,\mu_s) 
, 
\frac{\hat{b}^{(1)}_{\times}}{\hat{a}^{(1)}_{\times}+\hat{b}^{(1)}_{\times}}
, 2^{-2\beta-7}\Bigr). \Label{1-24-3-e}
\end{eqnarray}
\end{description}
The definition of the set $\Omega_1$ is replaced as the set of $\vec{\bm{N}}$ satisfying 
\begin{eqnarray}
%{N}^{(0)}&\in [X_{\per}^{-}(N,e^{-\mu_1},2^{-2\beta-8}),X_{\per}^{+}(N,e^{-\mu_1},2^{-2\beta-8})]
%\Label{10-5-10-d2}
%\\
%{N}^{(1)}&\in [X_{\per}^{-}(N,\mu_1 e^{-\mu_1},2^{-2\beta-8}), X_{\per}^{+}(N,\mu_1 e^{-\mu_1},2^{-2\beta-8})]
%\Label{10-5-11-d2}
%\\
{N}^{(0)}_1&\in [X_{\per}^{-}(N_1,e^{-\mu_1},2^{-\beta-6}),X_{\per}^{+}(N_1,e^{-\mu_1},2^{-\beta-6})] 
\Label{11-12-1-d}\\
{N}^{(1)}_1&\in [X_{\per}^{-}(N_1,\mu_1 e^{-\mu_1},2^{-\beta-6}),X_{\per}^{+}(N_1,\mu_1 e^{-\mu_1},2^{-\beta-6})] 
\Label{11-12-2-d}\\
{N}^{(0)}_2&\in [X_{\per}^{-}(N_2,e^{-\mu_2},2^{-\beta-6}),X_{\per}^{+}(N_2,e^{-\mu_2},2^{-\beta-6})] 
\Label{11-12-3-d}\\
{N}^{(1)}_2&\in [X_{\per}^{-}(N_2,\mu_2 e^{-\mu_2},2^{-\beta-6}),X_{\per}^{+}(N_2,\mu_2 e^{-\mu_2},2^{-\beta-6})] 
\Label{11-12-4-d}\\
{N}^{(2)}_2&\in [X_{\per}^{-}(N_2,\omega_2 \mu_2^2 e^{-\mu_2},2^{-\beta-6}), X_{\per}^{+}(N_2,\omega_2 \mu_2^2 e^{-\mu_2},2^{-\beta-6})] .
\Label{11-12-5-d}
\end{eqnarray}

Condition \ref{c6} is replaced as follows.
\setcounter{condition}{0}
\begin{condition}
Any element $\vec{\bm{N}}\in\Omega_1$ satisfies
\begin{eqnarray*}
\fl & X_{\per}^{-}(N_1,\mu_1 e^{-\mu_1},2^{-\beta-6})
X_{\per}^{-}(N_2,\omega_2 \mu_2^2 e^{-\mu_2},2^{-\beta-6}) \\
\fl >& 
(N_1- X_{\per}^{-}(N_1,e^{-\mu_1},2^{-\beta-6}) -X_{\per}^{-}(N_1,\mu_1 e^{-\mu_1},2^{-\beta-6}))
X_{\per}^{+}(N_2,\mu_2 e^{-\mu_2},2^{-\beta-6}), \\
\fl & X_{\per}^{-}(N_2,\omega_2 \mu_2^2 e^{-\mu_2},2^{-\beta-6})
X_{\per}^{-}(N_1,e^{-\mu_1},2^{-\beta-6}) \\
\fl >& 
(N_1- X_{\per}^{-}(N_1,e^{-\mu_1},2^{-\beta-6}) -X_{\per}^{-}(N_1,\mu_1 e^{-\mu_1},2^{-\beta-6}))
X_{\per}^{+}(N_2,e^{-\mu_2},2^{-\beta-6}), \\
\fl &
\frac{X_{\per}^{-}(N_2,\omega_2 \mu_2^2 e^{-\mu_2},2^{-\beta-6})}{N_1- X_{\per}^{-}(N_1,e^{-\mu_1},2^{-\beta-6}) -X_{\per}^{-}(N_1,\mu_1 e^{-\mu_1},2^{-\beta-6})}
+
\frac{X_{\per}^{-}(N_2,e^{-\mu_2},2^{-\beta-6})}{X_{\per}^{+}(N_1,e^{-\mu_1},2^{-\beta-6})} \\
\fl  >&
\frac{2 X_{\per}^{+}(N_2,\mu_2 e^{-\mu_2},2^{-\beta-6})}{X_{\per}^{-}(N_1,\mu_1 e^{-\mu_1},2^{-\beta-6})}.
\end{eqnarray*}
\end{condition}
\setcounter{condition}{3}

\begin{rem}[Adjustment of $\hat{q}^{(0)}$ for improved formula]
When the vacuum pulse has a possibility to contain a non-vacuum state,
we cannot apply the above formula $\hat{q}^{(0)}$.
Hence, we need its adjustment.
Assume that 
the vacuum pulse becomes a non-vacuum state with a probability $q$.
In this case, we replace $\hat{q}^{(0)}$ by 
\begin{eqnarray}
\fl
\hat{q}^{(0)}(M_0):=
p_{\est}^+(N_0-X_{\per}^+(N_0, q, 2^{-\beta-6}), M_0-X_{\per}^+(N_0, q, 2^{-\beta-6}), 2^{-\beta-6}).
\Label{5-1-2ii}
\end{eqnarray}
\end{rem}

\subsection{Numerical analysis}\Label{s13-1}
Next, we treat numerical analysis with the improved formula of the sacrifice-bit length.
In the following, we consider only the case when
the perfect vacuum state is available and 
the signal intensity is $\mu_2$, the decoy intensity is $\mu_1$, 
i.e., 
$\mu_s=\mu_2$, $N_s=N_{s,2}$, and $M_s=M_{s,2}$.
This is because this case is better than the opposite case 
in the asymptotic case as is shown in the paper \cite{HM2013}.
We also choose the parameters as
$N_0=N_1=N_2=N_{s,2}/10$, and $\beta=80$, i.e., the trace norm is less than $2^{-80}$.

\begin{figure}[htbp]
	\begin{tabular}{cc}
	\begin{minipage}{.5\textwidth}
		\centering
\includegraphics[width=8.00cm, clip]{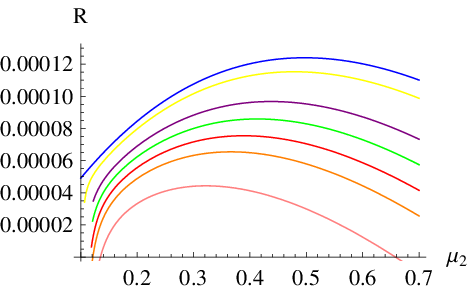}
%\scalebox{0.7}{\includegraphics[scale=1.7]{2by2-Rate.eps}}
\caption{The above graphs describe the key generation rate $R_{2,f}$ given in (\ref{12-8-1}) 
as functions of the signal intensity $\mu_2$ when the decoy intensity $\mu_1$ is $0.1$.
The pink line is the case when
the bit-length of raw keys $M_{s,2}$ is $10^{6}$.
The orange line is the case with $M_{s,2}=2 \times 10^6$.
The red line is the case with $M_{s,2}=3 \times 10^6$.
The green line is the case with $M_{s,2}=5 \times 10^6$.
The purple line is the case with $M_{s,2}=10^7$.
The yellow line is the case with $M_{s,2}=10^8$.
The blue line is the asymptotic case.}
\Label{f1}
	\end{minipage}
		\begin{minipage}{.5\textwidth}
\includegraphics[width=8.00cm, clip]{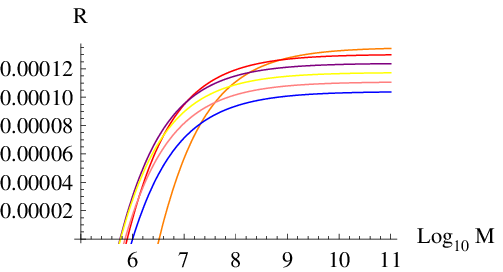}
\caption{The above graphs describe the key generation rate $R_{2,f}$ given in (\ref{12-8-1}) as 
functions of the bit-length of raw keys $M_{s,2}$ when the signal intensity $\mu_2$ is $0.5$.
The orange line is the case when the decoy intensity $\mu_1$ is $0.01$.
The red line is the case with $\mu_1=0.05$.
The purple line is the case with $\mu_1=0.1$.
The yellow line is the case with $\mu_1=0.15$.
The pink line is the case with $\mu_1=0.2$.
The blue line is the case with $\mu_1=0.25$.\vspace{6ex}}
\Label{f2}
	\end{minipage}
	\end{tabular}
\end{figure}

\begin{figure}[h]
\centering
\includegraphics[width=8.00cm, clip]{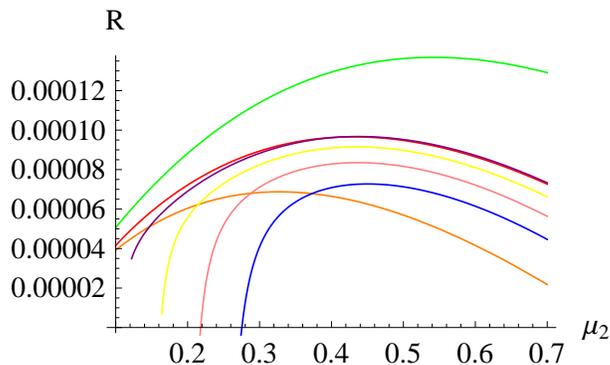}
\caption{The above graphs describe the key generation rate $R_{2,f}$ given in (\ref{12-8-1}) as functions of the signal intensity $\mu_2$ 
when the bit-length of raw keys $M_{s,2}$ is $10^{7}$.
The orange line is the case when the decoy intensity $\mu_1$ is $0.01$.
The red line is the case with $\mu_1=0.05$.
The purple line is the case with $\mu_1=0.1$.
The yellow line is the case with $\mu_1=0.15$.
The pink line is the case with $\mu_1=0.2$.
The blue line is the case with $\mu_1=0.25$.
The green line is the asymptotic case with $\mu_1 \to 0$.}
\Label{f3}
\end{figure}

It is natural to assume that the measured values $M_0$, $M_1$, $M_2$, $M_3$, and $N_{s,2}$
are given as functions of $M_{s,2}$ in the following way
\begin{eqnarray}
M_0&= p_0 N_0, \quad
M_1= (p_{1,\times}-s_{1,\times})N_1 \Label{12-8-2} \\
M_2&= (p_{2,\times}-s_{2,\times})N_2, \quad
M_3= s_{1,\times} N_1 ,\quad N_{s,2}=M_{s,2}/p_{2,+}.
\Label{12-8-3} 
\end{eqnarray}
We also assume that
the channel parameters,
i.e., the detection rates 
$p_{i,\times}$ and $p_{i,+}$
of $\mu_i$-intensity pulses 
with the bases $\times$ and $+$
and the rates $s_{i,\times}$ and $s_{i,+}$
of the detected $\mu_i$-intensity pulses having phase error 
to the transmitted $\mu_i$-intensity pulses 
with the bases $\times$ and $+$ as follows \cite{LB,GRTZ}.
\begin{eqnarray}
&p_{i,+}=p_{i,\times}= 1-e^{-\alpha \mu_i} + p_0 ,\nonumber \\
&s_{i,+} =s_{i,\times} = s(1-e^{-\alpha \mu_i}) + \frac{p_0}{2} ,
\Label{10-8-14-b}
\end{eqnarray}
where $\alpha$ is the total transmission including quantum efficiency of the detector,
and $s$ is the error due to the imperfection of the optical system.
In the following, we choose $\alpha =1.0 \times 10^{-3}$, $p_0 = 4.0 \times 10^{-7}$, $s=0.03$.
Then, we consider the key generation rate with finite-length:
\begin{eqnarray}
R_{2,f}:=\frac{M_{s,2}-S- \eta h(\frac{s_{1,+}}{p_{1,+}})M_{s,2}  }{N_{s,2}},
\Label{12-8-1}
\end{eqnarray}
where $S$ is the sacrifice bit-length and $\eta= 1.1$.

Since 
the required value $2^{-80}$ is too small and the sizes $N_0,N_1,N_2,M_s$ are too large,
the exact calculations of the values 
$X_{\per}^{\pm}(N,p,\alpha)$, 
$X_{\est}^{\pm}(N,k,\alpha)$, and $p_{\est}^{\pm}(N,k,\alpha)$ spend too much time.
So, instead of the exact calculation, 
we employ the bounds of these values based on
Chernoff bound, which require a smaller amount of calculations
and are summarized in Appendices A and B. 
Indeed, when $N$ is large enough,
the exact values of 
$X_{\per}^{\pm}(N,p,\alpha)$, 
$X_{\est}^{\pm}(N,k,\alpha)$, and $p_{\est}^{\pm}(N,k,\alpha)$
are close to the values based on Chernoff bound sufficiently for our purpose
because the difference between the exact values $X_{\per}^{\pm}(N,p,\alpha)$, $X_{\est}^{\pm}(N,k,\alpha)$
and their values based on Chernoff bound 
behaves with the order $\log N$.

As is illustrated in Figs. \ref{f1} and \ref{f2} with $\mu_1=0.1$, $\alpha=1/1000$, $p_0=0.0000004$, $\eta=1.1$,
the key generation rate
is close to the asymptotic key generation rate $R_2(\mu_1,\mu_2)$
when the length of the code $M_{s,2}$ is increasing.
As is shown in \cite{HM2013},
the asymptotic key generation rate is monotonically decreasing with respect to $\mu_1$. 
However, as is illustrated in Fig. \ref{f3} with $\alpha=1/1000$, $p_0=0.0000004$, $\eta=1.1$,
the key generation rate is not monotonically decreasing with respect to $\mu_1$
when the length of the code $M_{s,2}$ is not sufficiently large.
That is, too small $\mu_1$ does not give a good key generation rate.
This is because smaller $\mu_1$ yields a larger estimation error.

\section{Sacrifice bit-length when the intensities are not fixed with the finite-length case}\Label{s12}
\subsection{Derivation of modified formula}
%\subsection{Case when the intensities $\mu_1$ and $\mu_2$ obey the independent and identical distribution}\Label{s12-1}
Unfortunately, many realized quantum key distribution systems have
fluctuation for the intensities.
The formulas of the secure sacrifice bit-length given in Section \ref{s5-5}
can guarantee the security (\ref{35i})
when the partitions of $N_1$ pulses and $N_2$ pulses
obey the Poisson distribution with a fixed intensity.
However, when the intensities have fluctuation,
we have to derive 
the sacrifice bit-length 
by taking into account this factor.
That is, we need to discuss
the distribution for $\vec{\bm{N}}$
in the different way. 
%from that in Section \ref{s9}. 
In this section, we discuss the sacrifice bit-length 
by taking into account the statistical fluctuation for the intensities.
%while we will discuss the asymptotic key generation rate with a similar setting in Subsection \ref{s83}.
Since the definition of $\rho_2$ given in (\ref{1-23-1})
depends on the intensity $\mu_1$,
we need to modify the definition of $\rho_2$ properly.

%Indeed, the definition of $\hat{\phi}_3(\bm{M},\vec{\bm{N}})$ does not depend on the distribution for $\vec{\bm{N}}$. Hence, since it is not needed to change the definition of $\hat{\phi}_3(\bm{M},\vec{\bm{N}})$, the relation (\ref{1-24-6a}) holds without any modification. Thus, we need to modify the definitions of $\hat{\phi}_4(\bm{M})$ and the set $\Omega_1$ giving the fluctuation of $\vec{\bm{N}}$ so that the relation (\ref{1-25-1}) holds.

\subsubsection{Modifications of $\rho_2$, $\rho_3$, $\omega_2$, and $\omega_3$}
In the following, 
we assume that the intensities $\mu_1$ and $\mu_2$ 
independently obey independent and identical distributions
of the distributions $P_1$ and $P_2$
satisfying the following condition.
For any integer $n \ge 3$,
the relation
\begin{eqnarray*}
\rE[ e^{-\mu_2} \mu_2^n ]
\rE[ e^{-\mu_1} \mu_1^2 ]
\ge
\rE[ e^{-\mu_1} \mu_1^n ]
\rE[ e^{-\mu_2} \mu_2^2 ]
\end{eqnarray*}
holds, where
$\rE$ denotes the expectation under the distributions $P_1$ and $P_2$.
Under the above assumption,
we have expansions for two kinds of pulses.
\begin{eqnarray}
\fl \sum_{n=0}^\infty  \frac{\rE [e^{-\mu_1} \mu_1^n] }{n!}
|n \rangle\langle n|
=&
\rE [e^{-\mu_1}] |0\rangle\langle 0|
+\rE [e^{-\mu_1}\mu_1 ] |1\rangle\langle 1| 
+\rE [e^{-\mu_1}\mu_1^2 ] \omega_2\rho_2 \Label{1-23-1-c}, \\
\fl \sum_{n=0}^\infty \frac{\rE [e^{-\mu_2} \mu_2^n ] }{n!}
|n \rangle\langle n|
=&
\rE [e^{-\mu_2}]|0\rangle\langle 0|
+\rE [e^{-\mu_2}\mu_2 ]|1\rangle\langle 1| %\nonumber \\
+\rE [e^{-\mu_2}\mu_2^2 ] \omega_2 \rho_2
+ \omega_3' \rho_3,
\Label{4-9-4-c}
\end{eqnarray}
where
\begin{eqnarray}
\rho_2 &:= 
\frac{1}{\omega_2}
\sum_{n=2}^\infty 
\frac{\rE [e^{-\mu_1} \mu_1^n ]}{n! \rE [e^{-\mu_1} \mu_1^2 ] }
|n \rangle\langle n| \Label{10-5-1}\\
\rho_3 &:= 
\frac{1}{\omega_3'}
\sum_{n=3}^\infty 
\frac{
\rE[ e^{-\mu_2} \mu_2^n ]
\rE[ e^{-\mu_1} \mu_1^2 ]
-
\rE[ e^{-\mu_1} \mu_1^n ]
\rE[ e^{-\mu_2} \mu_2^2 ]
}{n! \rE[ e^{-\mu_1} \mu_1^2 ]}
|n \rangle\langle n| \Label{10-5-2}\\
\omega_2 &:=
\sum_{n=2}^\infty 
\frac{\rE [e^{-\mu_1} \mu_1^n ]}{n! \rE [e^{-\mu_1} \mu_1^2 ] }
%=\frac{\rE [1-(1+\mu_1 ) e^{-\mu_1} ]}{\rE [e^{-\mu_1} \mu_1^2 ] }
\Label{10-5-3} \\
\omega_3' &:=
\sum_{n=3}^\infty 
\frac{
\rE[ e^{-\mu_2} \mu_2^n ]
\rE[ e^{-\mu_1} \mu_1^2 ]
-
\rE[ e^{-\mu_1} \mu_1^n ]
\rE[ e^{-\mu_2} \mu_2^2 ]
}{n! \rE[ e^{-\mu_1} \mu_1^2 ]}.
\Label{10-5-4}
\end{eqnarray}
Indeed, our analysis in the previous sections uses 
the 
expansions (\ref{1-23-1}) and (\ref{4-9-4})
and their coefficients.
Hence, replacing expansions (\ref{1-23-1}) and (\ref{4-9-4}) by expansions (\ref{1-23-1-c}) and (\ref{4-9-4-c}),
we can apply the discussion with suitable modifications 
in the following way.
(A similar idea was used in Wang \cite{wang2,wang3}.)

%In the following, we treat the case when the signal intensity is $\mu_1$ and the decoy intensity is $\mu_2$.
\subsubsection{Modifications of the set $\Omega_1$ and the estimate $\hat{\bm{N}}_1$ and $\hat{\bm{N}}_2$}
%Then, we obtain the same argument with the following replacement:
%We replace the definition of $\rho_2$ in the above way, and
We redefine the set $\Omega_1$ as the set of 
$\vec{\bm{N}}$ satisfying 
\begin{eqnarray}
%{N}^{(0)}&\in [X_{\per}^{-}(N,\rE [e^{-\mu_1}],2^{-\beta-6}),X_{\per}^{+}(N,\rE [e^{-\mu_1}],2^{-\beta-6})]
%\Label{10-5-10}
%\\
%{N}^{(1)}&\in [X_{\per}^{-}(N,\rE [\mu_1 e^{-\mu_1}],2^{-\beta-6}), X_{\per}^{+}(N,\rE[ \mu_1 e^{-\mu_1}],2^{-\beta-6})]
%\Label{10-5-11}
%\\
{N}^{(0)}_1&\in [X_{\per}^{-}(N_1,\rE [e^{-\mu_1}],2^{-\beta-6}),X_{\per}^{+}(N_1,\rE [e^{-\mu_1}],2^{-\beta-6})] \nonumber \\
{N}^{(1)}_1&\in [X_{\per}^{-}(N_1,\rE [\mu_1 e^{-\mu_1}],2^{-\beta-6}),X_{\per}^{+}(N_1,\rE [\mu_1 e^{-\mu_1}], 2^{-\beta-6})] \nonumber \\
{N}^{(0)}_2&\in [X_{\per}^{-}(N_2,\rE [e^{-\mu_2}],2^{-\beta-6}),X_{\per}^{+}(N_2, \rE[ e^{-\mu_2}],2^{-\beta-6})] \nonumber \\
{N}^{(1)}_2&\in [X_{\per}^{-}(N_2,\rE [\mu_2 e^{-\mu_2}],2^{-\beta-6}),X_{\per}^{+}(N_2,\rE[ \mu_2 e^{-\mu_2}],2^{-\beta-6})] \nonumber \\
{N}^{(2)}_2&\in [X_{\per}^{-}(N_2,\rE [e^{-\mu_2}\mu_2^2 ] \omega_2,2^{-\beta-6}), 
X_{\per}^{+}(N_2,\rE [e^{-\mu_2}\mu_2^2 ] \omega_2,2^{-\beta-6})] .\nonumber 
\end{eqnarray}
We also redefine $\hat{\bm{N}}_1$ and $\hat{\bm{N}}_2$ in the following way.
\begin{eqnarray*}
\hat{N}^{(0)}_1&:=X_{\per}^{-}(N_1,\rE [e^{-\mu_1}],2^{-\beta-6})\\
\hat{N}^{(1)}_1&:=X_{\per}^{-}(N_1,\rE [\mu_1 e^{-\mu_1}],2^{-\beta-6})\\
\hat{N}^{(0)}_2&:=X_{\per}^{-}(N_2,\rE [e^{-\mu_2}],2^{-\beta-6})\\
\hat{N}^{(1)}_2&:=X_{\per}^{-}(N_2,\rE [\mu_2 e^{-\mu_2}],2^{-\beta-6})\\
\hat{N}^{(2)}_2&:=X_{\per}^{-}(N_2,\rE [e^{-\mu_2}\mu_2^2 ] \omega_2,2^{-\beta-6}) .
\end{eqnarray*}
%Then, we define $\hat{\bm{N}}$ by
%\begin{eqnarray*}
%\hat{N}^{(0)}&:=X_{\per}^{-}(N,\rE [e^{-\mu_1}],2^{-\beta-6})
%\Label{10-5-12}\\
%\hat{N}^{(1)}&:=X_{\per}^{-}(N,\rE [\mu_1 e^{-\mu_1}],2^{-\beta-6})
%\Label{10-5-13}
%\end{eqnarray*}

\subsubsection{Modifications of Conditions \ref{c6}, \ref{c5}, and \ref{c15}}
Under the above modification, we change Condition \ref{c6} as follows.
\setcounter{condition}{0}
\begin{condition}
Any element $\vec{\bm{N}}$
in the modified set $\Omega_1$
satisfies
\begin{eqnarray*}
\fl & X_{\per}^{-}(N_1,\rE [\mu_1 e^{-\mu_1}],2^{-\beta-6})
X_{\per}^{-}(N_2, \omega_2 \rE [ \mu_2^2 e^{-\mu_2}],2^{-\beta-6}) \\
\fl >& 
(N_1- X_{\per}^{-}(N_1,\rE [ e^{-\mu_1}],2^{-\beta-6}) 
-X_{\per}^{-}(N_1,\rE [\mu_1 e^{-\mu_1}],2^{-\beta-6}))
X_{\per}^{+}(N_2,\rE [\mu_2 e^{-\mu_2}],2^{-\beta-6}),\\
\fl & X_{\per}^{-}(N_2,\omega_2 \rE [\mu_2^2 e^{-\mu_2}],2^{-\beta-6})
X_{\per}^{-}(N,\rE [e^{-\mu_1}],2^{-\beta-6}) \\
\fl > &
(N_1- X_{\per}^{-}(N_1,\rE [e^{-\mu_1}],2^{-\beta-6}) 
-X_{\per}^{-}(N_1,\rE [\mu_1 e^{-\mu_1}],2^{-\beta-6}))
X_{\per}^{+}(N_2,\rE [e^{-\mu_2}],2^{-\beta-6}),\\
\fl &
\frac{X_{\per}^{-}(N_2,\omega_2 \rE [ \mu_2^2 e^{-\mu_2}],2^{-\beta-6})}
{N_1- X_{\per}^{-}(N_1,\rE [e^{-\mu_1}],2^{-\beta-6}) -X_{\per}^{-}(N_1,\rE [\mu_1 e^{-\mu_1}],2^{-\beta-6})}
+
\frac{X_{\per}^{-}(N_2,\rE [e^{-\mu_2}],2^{-\beta-6})}{X_{\per}^{+}(N_1,\rE [e^{-\mu_1}],2^{-\beta-6})} \\
\fl  >&
\frac{2 X_{\per}^{+}(N_2,\rE [\mu_2 e^{-\mu_2}],2^{-\beta-6})}{X_{\per}^{-}(N_1,\rE [\mu_1 e^{-\mu_1}],2^{-\beta-6})}.
\end{eqnarray*}
\end{condition}

Conditions \ref{c5} and \ref{c15} are redefined in the term of $\Omega_1$
%, $\hat{\bm{N}}_1$, and $\hat{\bm{N}}_2$
defined above.

\begin{condition}
For any element $\vec{\bm{N}}$ in the modified set $\Omega_1$, 
all of 
$A^{(0)}_1$, $A^{(1)}_1$, $A^{(0)}_2$, $A^{(1)}_2$, and $A^{(2)}_2$
are negative.
%the conditions in original Condition \ref{c5}. 
\end{condition}

\begin{condition}
Any element $\vec{\bm{N}}$ in the modified set $\Omega_1$ satisfies 
the conditions in original Condition \ref{c15}. 
\end{condition}

\subsubsection{Modifications of sacrifice bit-length $S$}
Next, in order to modify the sacrifice bit-length $S$, 
we modify $\hat{J}^{(0)}$, $\hat{J}^{(1)}$, $\hat{r}^{(1)}_{\times}$,
and $\hat{\phi}_1$
as follows.
\begin{eqnarray*}
\fl \hat{J}^{(0)}(\bar{q}^{(0)},N_s) 
&:=& X_{\per}^-(N_s,\rE [e^{-\mu_s}] \bar{q}^{(0)},2^{-2\beta-8}) 
\\
\fl \hat{J}^{(1)}(\bar{a}^{(1)}_{\times},\bar{b}^{(1)}_{\times},N_s) 
&:=& X_{\per}^-(N_s,\rE [e^{-\mu_s} \mu_s] (\bar{a}^{(1)}_{\times}+\bar{b}^{(1)}_{\times}),2^{-2\beta-8})
\\ 
\fl \hat{r}^{(1)}_{\times}(\bar{a}^{(1)}_{\times},\bar{b}^{(1)}_{\times},N_s)  
&:=& p_{\per}^+(\hat{J}^{(1)}(\bar{a}^{(1)}_{\times},\bar{b}^{(1)}_{\times},N_s) , 
\frac{\bar{b}^{(1)}_{\times}}{\bar{a}^{(1)}_{\times}+\bar{b}^{(1)}_{\times}}
, 2^{-2\beta-8}) 
\\
\fl
\hat{\phi}_1(\bar{q}^{(0)},\bar{a}^{(1)}_{\times},\bar{b}^{(1)}_{\times},N_s) 
&:=& M_s-\hat{J}^{(0)}(\bar{q}^{(0)},N_s)  \\
\fl & &- \hat{J}^{(1)}(\bar{a}^{(1)}_{\times},\bar{b}^{(1)}_{\times},N_s)  
(1- h(\min \{\hat{r}^{(1)}_{\times}(\bar{a}^{(1)}_{\times},\bar{b}^{(1)}_{\times},N_s)  ,1/2\}) ) .
%\nonumber 
\end{eqnarray*}
Then, using the same functions 
$\hat{q}^{(0)}$, $\hat{a}^{(1)}_{\times}$, and $\hat{b}^{(1)}_{\times}$,
we define
$\hat{\phi}_2$ 
by (\ref{10-26-1i}).

Finally, we define the sacrifice bit-length 
$S$ by (\ref{3-28-6i}) 
with modified Conditions \ref{c6}, \ref{c5}, and \ref{c15}.
Then, the relation (\ref{35i}) holds.
This fact can be shown by replacing the definitions of 
$\rho_2$ and $\rho_3$ and related parameters in the security proofs given in 
Sections \ref{s4} - \ref{s9-5}.
%Denoting the final state by $\rho_{A,E}$, we obtain
%\begin{eqnarray}\| \rho_{A,E}-\rho_{\ideal} \|_1 \le 2^{-\beta}.\end{eqnarray}

%Even when the signal intensity is $\mu_2$ and the decoy intensity is $\mu_1$, the above arguments are still valid by modifying 
%(\ref{10-5-10}), (\ref{10-5-11}), (\ref{10-5-12}), and (\ref{10-5-13}) in the following way
%and replacing $N$, $N^{(0)}$, and $N^{(1)}$ by $N_{s,2}$, ${N^{(0)}}'$, and ${N^{(1)}}'$.
%\begin{eqnarray}
%{N^{(0)}}' &\in [
%X_{\per}^{-}(N_{s,2},\rE [e^{-\mu_2}],2^{-\beta-6}),
%X_{\per}^{+}(N_{s,2},\rE [e^{-\mu_2}],2^{-\beta-6})]
%\Label{10-5-10-a}
%\\
%{N^{(1)}}'& \in [
%X_{\per}^{-}(N_{s,2},\rE [\mu_2 e^{-\mu_2}],2^{-\beta-6}), 
%X_{\per}^{+}(N_{s,2},\rE [\mu_2 e^{-\mu_2}],2^{-\beta-6})]
%\Label{10-5-11-a} \\
%\hat{N}^{(0)}&:=X_{\per}^{-}(N_{s,2},\rE [e^{-\mu_2}],2^{-\beta-6})\nonumber \\
%\hat{N}^{(1)}&:=X_{\per}^{-}(N_{s,2},\rE [\mu_2 e^{-\mu_2}],2^{-\beta-6}).\nonumber 
%\end{eqnarray}
%Fig. \ref{Gauss} gives the numerical calculation of the key generation rate when the intensities $\mu_1$ and $\mu_2$ obey the Gaussian distribution with this setting.

\subsubsection{Extension to Case when the distributions of $\mu_2$ and $\mu_1$ are unknown}
Next, we treat the case
when there are several candidates for the distribution of
$\mu_2$ and $\mu_1$
while $\mu_2$ and $\mu_1$
obey independent and identical distributions.
The possible distributions is denoted by
$P_{\theta,1}$ and $P_{\theta,2}$,
and the expectation is written by $\rE_\theta$.
Then, we denote the set $\Omega_1 $ under the distribution $P_\theta$ by $\Omega_{1,\theta}$.

In this case, Conditions \ref{c6} and \ref{c5} are needed to be satisfied
for any $\theta$.
Hence, Condition \ref{c6} is redefined as follows.
That is, the following relations hold for any $\theta$.
\begin{eqnarray*}
\fl & 
X_{\per}^{-}(N_1,\rE_{\theta} [\mu_1 e^{-\mu_1}],2^{-\beta-6})
X_{\per}^{-}(N_2,\omega_{2|\theta} \rE_{\theta} [\mu_2^2 e^{-\mu_2}],2^{-\beta-6}) \\
\fl >& 
(N_1- X_{\per}^{-}(N_1,\rE_{\theta} [ e^{-\mu_1}],2^{-\beta-6}) 
-
X_{\per}^{-}(N_1,\rE_{\theta} [\mu_1 e^{-\mu_1}],2^{-\beta-6}))
X_{\per}^{+}(N_2,\rE_{\theta} [\mu_2 e^{-\mu_2}],2^{-\beta-6}),\\
\fl & 
X_{\per}^{-}(N_2,\omega_{2|\theta} \rE_{\theta} [\mu_2^2 e^{-\mu_2}],2^{-\beta-6})
X_{\per}^{-}(N,\rE_{\theta} [e^{-\mu_1}],2^{-\beta-6}) \\
\fl >& 
(N_1- X_{\per}^{-}(N_1,\rE_{\theta} [e^{-\mu_1}],2^{-\beta-6}) 
-
X_{\per}^{-}(N_1,\rE_{\theta} [\mu_1 e^{-\mu_1}],2^{-\beta-6}))
X_{\per}^{+}(N_2,\rE_{\theta} [e^{-\mu_2}],2^{-\beta-6}),\\
\fl &
\frac{X_{\per}^{-}(N_2,\omega_{2|\theta} \rE_{\theta} [ \mu_2^2 e^{-\mu_2}],2^{-\beta-6})}
{N_1- X_{\per}^{-}(N_1,\rE_{\theta} [e^{-\mu_1}],2^{-\beta-6}) -X_{\per}^{-}(N_1,\rE_{\theta} [\mu_1 e^{-\mu_1}],2^{-\beta-6})}
+
\frac{X_{\per}^{-}(N_2,\rE_{\theta} [e^{-\mu_2}],2^{-\beta-6})}{X_{\per}^{+}(N_1,\rE_{\theta} [e^{-\mu_1}],2^{-\beta-6})} \\
\fl  >&
\frac{2 X_{\per}^{+}(N_2,\rE_{\theta} [\mu_2 e^{-\mu_2}],2^{-\beta-6})}{X_{\per}^{-}(N_1,\rE_{\theta} [\mu_1 e^{-\mu_1}],2^{-\beta-6})},
\end{eqnarray*}
where
$\omega_{2|\theta}$
is $\omega_2$ with the distribution $P_{\theta,1}$.

Further,
we redefine Condition \ref{c5}
as the condition that 
all of 
$A^{(0)}_1$, $A^{(1)}_1$, $A^{(0)}_2$, $A^{(1)}_2$, and $A^{(2)}_2$
are negative for $\vec{\bm{N}} \in \cup_{\theta}\Omega_{1,\theta}$.
We define 
$\hat{\phi}_{2,\theta}$ 
to be 
$\hat{\phi}_2$ given in (\ref{10-26-1i})
when the true distributions are $P_{\theta,1}$ and $P_{\theta,2}$.
Finally, we define the sacrifice bit-length 
$S$ 
by $\sup_\theta \hat{\phi}_{2,\theta}+2\beta+5$
when modified Conditions \ref{c6}, \ref{c5}, and \ref{c15} hold.
Otherwise, we set $S$ to be $\dim C_1$.
Then, letting 
$\rho_{A,E|\theta}$ be
the final state with the true distributions
$P_{\theta,1}$ and $P_{\theta,2}$, $\rho_{\ideal|\theta}$ be the ideal state,
we obtain
\begin{eqnarray}
\| \rho_{A,E|\theta}-\rho_{\ideal|\theta} \|_1
\le 2^{-\beta}.
\end{eqnarray}
That is, the inequality holds for any $\theta$.

In the following, we consider the case when the pulses are generated with the mixture of the plural independent and identical
distributions $P_{\theta,1}$ and $P_{\theta,2}$, respectively.
In this case, we define 
Conditions \ref{c6}, \ref{c5}, and \ref{c15} in the above way.
Then,
the intensities of $N_{s,1}+N_1$ pulses
are described by
$(\mu_{1,1}, \ldots, \mu_{1,N_{s,1}+N_1})$
and 
are subject to the distribution
$\sum_{\theta} \lambda_{\theta} P_{\theta,1}^{\times (N_{s,1}+N_1)}$,
where
$P^{\times N_{s,1}}$ is the $N_{s,1}$-fold independent and identical distribution
of $P$.
Similarly,
the intensities of $N_{s,2}+N_2$ pulses
are described by
$(\mu_{2,1}, \ldots, \mu_{2,N_{s,2}+N_2})$
and 
are subject to the distribution
$\sum_{\theta} \lambda_{\theta} P_{\theta,2}^{\times (N_{s,2}+N_2)}$.
Then, we choose the sacrifice bit-length 
$S$ 
to be $\sup_\theta \hat{\phi}_{2,\theta}(\bm{M})+2\beta+5$
when modified Conditions \ref{c6}, \ref{c5}, and \ref{c15} hold.
Otherwise, we set $S$ to be $\dim C_1$.
Since the final state is $\sum_\theta \lambda_\theta \rho_{A,E|\theta}$,
we obtain
\begin{eqnarray}
\fl \|
(\sum_\theta \lambda_\theta \rho_{A,E|\theta}) -
(\sum_\theta \lambda_\theta \rho_{\ideal|\theta})
\|_1
\le
\sum_\theta \lambda_\theta
\| \rho_{A,E|\theta} -\rho_{\ideal|\theta}
\|_1
\le 2^{-\beta}.
\end{eqnarray}
Hence, 
the universal composability criterion is upper bounded by
$2^{-\beta}$.

\subsection{Numerical analysis with Gaussian distribution}\Label{s13-2}
Next, we treat numerical analysis
%the improved formula of the sacrifice-bit length
when two intensities $\mu_1$ and $\mu_2$
independently and identically
obey the Gaussian distributions with the averages 
$\bar{\mu_1}$ and $\bar{\mu_2}$
and the standard deviations 
$\bar{\mu_1}t$ and $\bar{\mu_2}t$, respectively
because 
these fluctuations usually are caused  by the thermal noise.
That is, we assume the value $t$ is independent of the intensity.
This assumption holds, if the weak pulses are obtained
from strong light pulses with a well-calibrated attenuator;
the error originates mainly from the intensity fluctuation of the light source.
In the following, we consider only the case when
the signal intensity is $\mu_2$, the decoy intensity is $\mu_1$, 
i.e., 
$\mu_s=\mu_2$, $N_s=N_{s,2}$, and $M_s=M_{s,2}$.
We also choose the parameters as
$N_0=N_1=N_2=N_{s,2}/10$, and $\beta=80$, i.e., the trace norm is less than $2^{-80}$.

\begin{figure}[htbp]
	\begin{tabular}{cc}
	\begin{minipage}{.5\textwidth}
\centering
\includegraphics[width=8.00cm, clip]{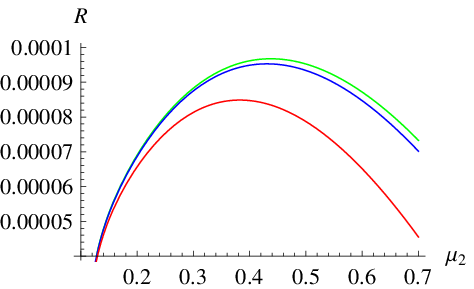}
\caption{All graphs give the key generation rates $R_{2,f}$ 
when the bit-length of raw keys $M_{s,2}$ is $10^7$ and the decoy intensity $\mu_1$ is $0.1$
and is smaller than the signal intensity $\mu_2$.
The horizontal axis describes the signal intensity $\mu_2$.
The green line is the rate $R_{2,f}$ with $t=0\%$.
The blue line is the rate $R_{2,f}$ with $t=10\%$.
The red line is the rate $R_{2,f}$ with $t=30\%$.}
\Label{Gauss1}
	\end{minipage}
		\begin{minipage}{.5\textwidth}
\centering
\includegraphics[width=8.00cm, clip]{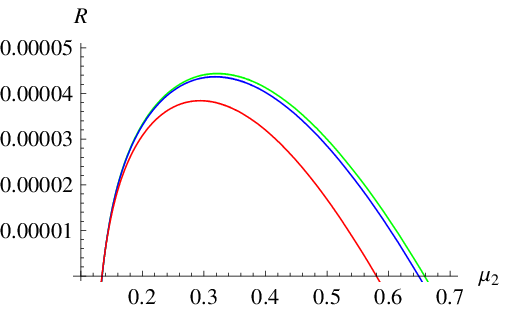}
\caption{
All graphs give the key generation rate $R_{2,f}$ with the bit-length of raw keys $M_{s,2}=10^6$ 
when the decoy intensity $\mu_1$ is $0.1$ and is smaller than the signal intensity $\mu_2$.
The horizontal axis describes the signal intensity $\mu_2$.
The green line is the rate $R_{2,f}$ with $t=0\%$.
The blue line is the rate $R_{2,f}$ with $t=10\%$.
The red line is the rate $R_{2,f}$ with $t=30\%$.}
\Label{Gauss2}
	\end{minipage}
	\end{tabular}
\end{figure}

%Next, we treat the case when two intensities $\mu_1$ and $\mu_2$ cannot be fixed.

In order to calculate the sacrifice bit-length given above,
%Section \ref{s12}, 
we need $\rE[e^{\mu_i}]$, $\rE[\mu_i e^{\mu_i}]$, $\rE[\mu_i^{2} e^{\mu_i}]$, and $\omega_2$,
which can be easily calculated from the formulas given in \ref{as4}.
In this case, due to (\ref{D6}), it is natural to assume that
the measured values $M_0$, $M_1$, $M_2$, and $M_3$ are
given by (\ref{12-8-2}) and (\ref{12-8-3}) 
when 
$p_{i,+}$, $p_{i,\times}$, $s_{i,+}$, and $s_{i,\times}$ are given as
\begin{eqnarray}
p_{i,+}&= p_{i,\times}= 1-\rE [e^{-\alpha \mu_i}] + p_0 
=1-e^{\frac{2\alpha \bar{\mu_i}-\alpha^{2} t^{2}\bar{\mu_1}^{2}}{2}} + p_0 
\Label{12-8-4}
\\
s_{i,+}&= s_{i,\times}= s(1-\rE [e^{-\alpha \mu_i}]) + \frac{p_0}{2} 
=s(1-e^{\frac{2\alpha \bar{\mu_i}-\alpha^{2} t^{2}\bar{\mu_1}^{2}}{2}}) + \frac{p_0}{2} 
\Label{12-8-5}.
\end{eqnarray}
Hence, 
we choose $N_{s,2}$ to be $M_{s,2}/p_{2,+}$.
Under this assumption,
substituting the sacrifice bit-length given above %in Section \ref{s12}
into the key generation rate $R_{f,2}$ given in (\ref{12-8-1}),
we obtain the numerical calculation in Figs. \ref{Gauss1} and \ref{Gauss2}.
These numerical results suggest that
when the variance is less than 10\% of the average,
the fluctuations of intensities do not cause serious decrease of the key generation rate.
Here, similar to Subsection \ref{s13-1},
Here, similar to Subsection \ref{s13-1},
we employ the bounds of $X_{\per}^{\pm}(N,p,\alpha)$, 
$X_{\est}^{\pm}(N,k,\alpha)$, and $p_{\est}^{\pm}(N,k,\alpha)$
given in Appendices A and B.

%\newpage

\section{Preparation for behavior of random variables}\Label{s6}
In this section, we explain that 
we can use the binomial distribution even when 
the true distribution is the hypergeometric distribution.
In this paper, we also treat the hypergeometric distribution
$HG(L,K,N)$ with $N$ draws and $L$ samples containing $K$ success.
In fact, 
the outcome obeys the binary distribution 
in the case of sampling with replacement,
and 
the outcome obeys the hypergeometric distribution
in the case of sampling without replacement.

Then, we study the stochastic behavior of the measured values
$\bm{M}=(M_s, M_0, M_1, M_2, M_3)$ 
under the assumption that the parameters $\bar{q}^{(0)}, \bar{\bm{a}}$, and $\bar{\bm{b}}$ 
are unknown, but are fixed to certain values.
For this purpose, we introduce 
the random variables $\check{M}_s, \check{M}_0, \check{M}_1, \check{M}_2, \check{M}_3$
subject to the binary distributions 
with the same draws and the same successful probabilities 
as $M_s, M_0, M_1, M_2, M_3$ by sampling with replacement.
The number of vacuum pulses is $N_0+N^{(0)}_1+N^{(0)}_2+N_s^{(0)}$.
The detection rate in Bob's side among 
$N_0+N^{(0)}_1+N^{(0)}_2+N_s^{(0)}$ vacuum pulses
is fixed.
$N_0$ vacuum pulses are randomly chosen from $N_0+N^{(0)}_1+N^{(0)}_2+N_s^{(0)}$ vacuum pulses.
Then,
the number $M_0$ of detected pulses among these $N_0$ vacuum pulses
obeys the hypergeometric distribution $HG(N_0+N^{(0)}_1+N^{(0)}_2+N_s^{(0)},\bar{q}^{(0)}(N_0+N^{(0)}_1+N^{(0)}_2+N_s^{(0)}),N_0)$.
For a real number $R > \bar{q}^{(0)}$,
the probability $\Pr \{ \frac{M_0}{N_0} >R \}$ 
%under $M_0 \sim HG(N_0+N^{(0)}_1+N^{(0)}_2,\bar{q}^{(0)}(N_0+N^{(0)}_1+N^{(0)}_2),N_0)$
is smaller than
the probability $\Pr \{ \frac{\check{M}_0}{N_0} >R \}$. 
%under $M_0 \sim Bin(N_0,\bar{q}^{(0)})$.
The reason is as follows.
Let $L$ be an arbitrary integer less than $N_0-1$.
If the observed detection rate of the initial $L$ transmitted pulses
is greater than $R$,
the detection probability of the $L+1$-th pulse
is less than $\bar{q}^{(0)}$
in the case of sampling without replacement.
Thus, we obtain
\begin{eqnarray}
\Pr_{\bm{M},\bm{J}|\bar{q}^{(0)}, \bar{\bm{a}},\bar{\bm{b}},\vec{\bm{N}}} 
\{ M_0 < X_{\per}^-(N_0,\bar{q}^{(0)}, \epsilon) \} \le \epsilon \Label{10-13-1}\\
\Pr_{\bm{M},\bm{J}|\bar{q}^{(0)}, \bar{\bm{a}},\bar{\bm{b}},\vec{\bm{N}}} 
\{ M_0 > X_{\per}^+(N_0,\bar{q}^{(0)}, \epsilon) \} \le \epsilon \Label{10-13-2},
\end{eqnarray}
where
$\Pr_{\bm{M},\bm{J}|\bar{q}^{(0)}, \bar{\bm{a}},\bar{\bm{b}},\vec{\bm{N}}} $
is the distribution of the random variables $\bm{M},\bm{J}$
when $\bar{q}^{(0)}, \bar{\bm{a}},\bar{\bm{b}}$ and $\vec{\bm{N}}$ are fixed.
That is,
we obtain 
\begin{eqnarray}
\Pr_{\bm{M},\bm{J}|\bar{q}^{(0)}, \bar{\bm{a}},\bar{\bm{b}},\vec{\bm{N}}} 
\{ \bar{M}_0 < X_{\est}^-(N_0,M_0, \epsilon) \} \le \epsilon \Label{10-13-10}\\
\Pr_{\bm{M},\bm{J}|\bar{q}^{(0)}, \bar{\bm{a}},\bar{\bm{b}},\vec{\bm{N}}} 
\{ \bar{M}_0 > X_{\est}^+(N_0,M_0, \epsilon) \} \le \epsilon ,
\end{eqnarray}
where $\bar{M}_0$ is the expectation of $M_0$, which equals $\bar{q}^{(0)} N_0$.

\begin{rem}
Here, we should remark that
the above analysis does not imply that
the non-replacement case can be reduced to the replacement case perfectly.
Let $M_1^{(0)}$ be the number of detected pulses among $N_1^{(0)}$ transmitted vacuum pulses
and
$\check{M}_1^{(0)}$ be the random variable
subject to the binary distribution
with the same draws and the same successful probability 
as $M_1^{(0)}$ by sampling with replacement.
Since 
$\check{M}_0$ and $\check{M}_1^{(0)}$ are
independent of each other due to sampling with replacement,
we have
\begin{eqnarray}
& \Pr \{
\check{M}_1^{(0)} > N_1^{(0)}\frac{\check{M}_0}{N_0}+ a
\}
\le
\sum_{k=0}^{N_1^{(0)}}
\Pr 
(\{ \check{M}_1^{(0)} \ge k \} \cap \{ N_1^{(0)}\frac{\check{M}_0}{N_0} \le  a-k \})
\nonumber
\\
= &
\sum_{k=0}^{N_1^{(0)}}
\Pr \{ \check{M}_1^{(0)} \ge k \} \cdot \Pr \{ N_1^{(0)}\frac{\check{M}_0}{N_0} \le  a-k \} .
\Label{3-2-1}
\end{eqnarray}
However, 
since ${M}_0$ and ${M}_1^{(0)}$ have no replacement,
${M}_0$ and ${M}_1^{(0)}$ are not independent of each other.
Hence, the relation (\ref{3-2-1}) does not hold for ${M}_0$ and ${M}_1^{(0)}$.
That is,
the probability $\Pr \{
{M}_1^{(0)} > N_1^{(0)}\frac{{M}_0}{N_0}+ a
\}$ cannot be bounded by RHS of (\ref{3-2-1}).
Instead of RHS of (\ref{3-2-1}), we have a weaker evaluation,
\begin{eqnarray}
\fl & \Pr \{
{M}_1^{(0)} > N_1^{(0)}\frac{{M}_0}{N_0}+ a
\}
\le
\Pr
\{{M}_1^{(0)} \ge N_1^{(0)} \bar{q}^{(0)} + \frac{a}{2}\}
+
\Pr
\{N_1^{(0)} \bar{q}^{(0)} - \frac{a}{2} \ge N_1^{(0)}\frac{{M}_0}{N_0} \} 
\nonumber
\\
\fl \le &
\Pr
\{\check{M}_1^{(0)} \ge N_1^{(0)} \bar{q}^{(0)} + \frac{a}{2}\}
+
\Pr
\{N_1^{(0)} \bar{q}^{(0)} - \frac{a}{2} \ge N_1^{(0)}\frac{\check{M}_0}{N_0} \} .
\Label{3-2-2}
\end{eqnarray}
because
$
\{{M}_1^{(0)} > N_1^{(0)}\frac{{M}_0}{N_0}+ a\}
\subset 
\{{M}_1^{(0)} \ge N_1^{(0)} \bar{q}^{(0)} + \frac{a}{2}\}
\cup
\{N_1^{(0)} \bar{q}^{(0)} - \frac{a}{2} \ge N_1^{(0)}\frac{{M}_0}{N_0} \}$.
That is,
the above discussion cannot yield a better bound (RHS of (\ref{3-2-1}))
but can yield a weaker bound (RHS of (\ref{3-2-2})).
\end{rem}

Next, we consider a more complicated case, i.e.,
focus on $N_1$ $\mu_1$-intensity pulses,
which contain
$N^{(0)}_1$ vacuum pulses,
$N^{(1)}_1$ pulses with the single-photon state,
and
$N^{(2)}_1$ pulses with the state $\rho_2$.
Then, the expectation $\bar{M}_1$ of $M_1$ is 
$\frac{\bar{q}^{(0)}}{2} N^{(0)}_1+ \bar{a}^{(1)}_{\times} N^{(1)}_1+ \bar{a}^{(2)}_{\times} N^{(2)}_1$.
Assume that $N^{(2)}_1=0$ and $\frac{\bar{q}^{(0)}}{2} <\bar{a}^{(1)}_{\times}$.
For a real number $R > \bar{M}_1/N_1$,
the probability $\Pr \{ \frac{M_1}{N_1} >R \}$ 
is smaller than 
the probability $\Pr \{ \frac{\check{M}_1}{N_1} >R \}$. 
%under $M_1 \sim Bin(N_1,\bar{M}_1/N_1)$.
This fact can be shown as follows.
Assume that the detection rate among the initial $L$ pulses
is greater than $R$.
Under the above condition,
the rate of the single-photon pulses among initial $L$ pulses
is higher than
$\frac{N^{(1)}_1}{N_1}$ with probability more than $1/2$.
Conversely,
under the above condition,
the rate of the single-photon pulses among remaining $N_1-L$ pulses is smaller than
$\frac{N^{(1)}_1}{N_1}$ with probability more than $1/2$.
Hence, 
the detecting probability of the $L+1$-th pulse
is less than $\bar{M}_1/N_1$.
Therefore,
\begin{eqnarray}
\Pr_{\bm{M},\bm{J}|\bar{q}^{(0)}, \bar{\bm{a}},\bar{\bm{b}},\vec{\bm{N}}} 
\{ M_1 < X_{\per}^-(N_1,\frac{\bar{M}_1}{N_1}, \epsilon) \} \le \epsilon \\
\Pr_{\bm{M},\bm{J}|\bar{q}^{(0)}, \bar{\bm{a}},\bar{\bm{b}},\vec{\bm{N}}} 
\{ M_1 > X_{\per}^+(N_1,\frac{\bar{M}_1}{N_1}, \epsilon) \} \le \epsilon ,
\end{eqnarray}
which implies that
\begin{eqnarray}
\Pr_{\bm{M},\bm{J}|\bar{q}^{(0)}, \bar{\bm{a}},\bar{\bm{b}},\vec{\bm{N}}} 
\{ \bar{M}_1 < X_{\est}^-(N_1,M_1, \epsilon) \} \le \epsilon \Label{10-13-11} \\
\Pr_{\bm{M},\bm{J}|\bar{q}^{(0)}, \bar{\bm{a}},\bar{\bm{b}},\vec{\bm{N}}} 
\{ \bar{M}_0 > X_{\est}^+(N_1,M_1, \epsilon) \} \le \epsilon .
\end{eqnarray}
Repeating a similar discussion,
we can show the above relations without the condition $N^{(2)}_1=0$.

Similarly, 
the expectations 
$\overline{M}_2$ and $\overline{M}_3$
of $M_2$ and $M_3$
are calculated to
$\frac{\bar{q}^{(0)}}{2} N^{(0)}_2+ \bar{a}^{(1)}_{\times} N^{(1)}_2+ \bar{a}^{(2)}_{\times} N^{(2)}_2+ \bar{a}^{(3)}_{\times} N^{(3)}_2$
and
$\frac{\bar{q}^{(0)}}{2} N^{(0)}_1+ \bar{b}^{(1)}_{\times} N^{(1)}_1+ \bar{b}^{(2)}_{\times} N^{(2)}_1$,
and are denoted by $\bar{M}_2$ and $\bar{M}_3$,
respectively.
Then, we obtain 
\begin{eqnarray}
\Pr_{\bm{M},\bm{J}|\bar{q}^{(0)}, \bar{\bm{a}},\bar{\bm{b}},\vec{\bm{N}}} 
\{ \bar{M}_2 < X_{\est}^-(N_2,M_2, \epsilon) \} \le \epsilon \\
\Pr_{\bm{M},\bm{J}|\bar{q}^{(0)}, \bar{\bm{a}},\bar{\bm{b}},\vec{\bm{N}}} 
\{ \bar{M}_2 > X_{\est}^+(N_2,M_2, \epsilon) \} \le \epsilon \Label{10-13-12}\\
\Pr_{\bm{M},\bm{J}|\bar{q}^{(0)}, \bar{\bm{a}},\bar{\bm{b}},\vec{\bm{N}}} 
\{ \bar{M}_3 < X_{\est}^-(N_1,M_3, \epsilon) \} \le \epsilon \\
\Pr_{\bm{M},\bm{J}|\bar{q}^{(0)}, \bar{\bm{a}},\bar{\bm{b}},\vec{\bm{N}}} 
\{ \bar{M}_3 > X_{\est}^+(N_1,M_3, \epsilon) \} \le \epsilon \Label{10-13-13}.
\end{eqnarray}
For a detail discussion, see \cite{Hari}.

\section{Outlines of security proof}\Label{s4}
\subsection{Requirement for a function estimating the leaked information $\phi$ from $\bm{M}$ and $\vec{\bm{N}}$}
In this section, we give the outline of the security proof of the formula of the sacrifice bit-length given 
in Subsection \ref{s-non-imp}
while their partitions will be given in latter sections.
%In the following, we treat the case when the intensity of the signal pulse (the signal intensity) is $\mu_1$ and the intensity of the decoy pulse (the decoy intensity) is $\mu_2$.
First, we fix the partition $\vec{\bm{N}}$ of transmitted pulses.
The aim of this subsection is to give our requirement for a function estimating 
the leaked information $\phi$ from the measured values $\bm{M}=(M_s, M_0, M_1, M_2, M_3)$ and the partition $\vec{\bm{N}}$.
For this purpose,
we introduce three conditions for the partition $\vec{\bm{N}}$.
\begin{condition}\Label{c1}
\begin{eqnarray}
N^{(1)}_1 N^{(2)}_2- N^{(1)}_2 N^{(2)}_1 >0.
\end{eqnarray}
\end{condition}
\begin{condition}\Label{c2}
\begin{eqnarray}
N^{(2)}_1 N^{(0)}_2-N^{(0)}_1 N^{(2)}_2 <0.
\end{eqnarray}
\end{condition}

\begin{condition}\Label{c10}
\begin{eqnarray}
%-N^{(0)}-\frac{N^{(1)}(N_1^{(2)} N_2^{(0)}-N_2^{(2)} N_1^{(0)})}{2(N_1^{(1)} N_2^{(2)}-N_1^{(2)} N_2^{(1)})}-\frac{N^{(1)} N_1^{(0)}}{N_1^{(1)}}
-\frac{N_1^{(2)} N_2^{(0)}-N_2^{(2)} N_1^{(0)}}{2(N_1^{(1)} N_2^{(2)}-N_1^{(2)} N_2^{(1)})}-\frac{N_1^{(0)}}{N_1^{(1)}}
<0.
\end{eqnarray}
\end{condition}

If the all values take their expectation,
the left hand side of Condition \ref{c1} is 
$\frac{1}{2}e^{-\mu_1-\mu_2}\mu_1\mu_2(\mu_2-\mu_1)\omega_2 N_1 N_2 $, 
and is positive.
In the same assumption,
the left hand side of Condition \ref{c2} is 
$\frac{1}{2}N_1 N_2 e^{-(\mu_1+\mu_2)} \omega_2(\mu_1^2-\mu_2^2)$,
and is negative.

Condition \ref{c10} is equivalent to
\begin{eqnarray*}
\fl N_2^{(2)}-\frac{N_1^{(2)} N_2^{(0)}}{N_1^{(0)}}
=
\frac{N_2^{(2)} N_1^{(0)}- N_1^{(2)} N_2^{(0)}}{N_1^{(0)}}
<
\frac
{2(N_1^{(1)} N_2^{(2)}-N_1^{(2)} N_2^{(1)})}
{N_1^{(1)}} 
=2 N_2^{(2)}
-2\frac{N_1^{(2)} N_2^{(1)}}{N_1^{(1)}}.
\end{eqnarray*}
The above condition is equivalent to
\begin{eqnarray}
2 \frac{N_1^{(2)} N_2^{(1)}}{N_1^{(1)}}
-
\frac{N_1^{(2)} N_2^{(0)}}{N_1^{(0)}}
<
N_2^{(2)}.
\end{eqnarray}
Then, this condition is converted to 
\begin{eqnarray}
2 \frac{N_2^{(1)}}{N_1^{(1)}}
<
\frac{N_2^{(0)}}{N_1^{(0)}} +\frac{N_2^{(2)}}{N_1^{(2)}}.
\end{eqnarray}
When the all values take their expectation,
the left hand side is 
$2 \frac{\mu_2}{\mu_1}e^{-\mu_2+\mu_1} $,
and
the right hand side is 
$(\frac{\mu_2^2}{\mu_1^2}+1)e^{-\mu_2+\mu_1} $.
Then, the above condition holds.
Hence, these three assumptions are natural.

In the following, in order to give an upper bound of $\phi(\bm{J})$,
for a given real number $\beta>0$, 
we assume that there exists a function $\hat{\phi}_a$
of $\vec{\bm{N}}$ and measured values $\bm{M}=(M_s, M_0, M_1, M_2, M_3)$
satisfying that
\begin{eqnarray}
\Pr_{\bm{M},\bm{J}|\bar{q}^{(0)}, \bar{\bm{a}},\bar{\bm{b}},\vec{\bm{N}}} 
\{ \hat{\phi}_a(\bm{M},\vec{\bm{N}}) < \phi(\bm{J}) \}
\le 8 \cdot 2^{-2\beta-8}
\Label{1-24-6a}.
\end{eqnarray}
Indeed, in Section \ref{s7},
we will give its concrete example $\hat{\phi}_3$.
The relation (\ref{1-24-6a}) implies the relations
\begin{eqnarray*}
& \Pr_{\bm{M},\bm{J}|\vec{\bm{N}}} 
\{ \hat{\phi}_a(\bm{M},\vec{\bm{N}}) < 
\phi(\bm{J}) \} \\
= &
\sum_{\bar{q}^{(0)}, \bar{\bm{a}},\bar{\bm{b}}}
Q_e(\bar{q}^{(0)}, \bar{\bm{a}},\bar{\bm{b}}|\vec{\bm{N}})
\Pr_{\bm{M},\bm{J}|\bar{q}^{(0)}, \bar{\bm{a}},\bar{\bm{b}},\vec{\bm{N}}} 
\{ \hat{\phi}_a(\bm{M},\vec{\bm{N}}) < 
\phi(\bm{J}) \} 
\le 8 \cdot 2^{-2\beta-8},
\end{eqnarray*}
where $Q_e$ is the conditional distribution of 
$\bar{q}^{(0)}, \bar{\bm{a}},\bar{\bm{b}}$ conditioned with the partition$\vec{\bm{N}}$.

\subsection{Requirement for a function estimating the leaked information $\phi$ from $\bm{M}$}
This subsection has three aims.
The first aim is to give a requirement for a function estimating the leaked information $\phi$ from 
the measured values $\bm{M}=(M_s, M_0, M_1, M_2, M_3)$.
The second aim is to give the sacrifice bit-length by using a function satisfying this requirement.
The final aim is to figure out the structure of our formula of the sacrifice bit-length,
which gives the detail of Figs. \ref{outline-zu2} and \ref{outline-zu}.

Now, we remember the definition of the set $\Omega_1$ and Condition \ref{c6}.
Condition \ref{c6} is equivalent to the condition that
any element $\vec{\bm{N}}\in\Omega_1$ satisfies Conditions \ref{c1}, \ref{c2}, and \ref{c10}.
In the following, we assume that there exists a real-valued 
function $\hat{\phi}_b$ of the measured value $\bm{M}$
that satisfies that
\begin{eqnarray}
%\Pr_{\bm{M},\bm{J},\vec{\bm{N}} %|\bm{q},\bm{b}} 
\{
\hat{\phi}_b(\bm{M}) < \hat{\phi}_a(\bm{M},\vec{\bm{N}}) \}
\subset \Omega_1^c
%\cup \Omega_1
%\le 2^{-\beta-5}.
\Label{1-25-1}
\end{eqnarray}
under Condition \ref{c6}.
In Section \ref{s9}, 
we will give a concrete function 
$\hat{\phi}_4$ satisfying the above condition.
Note that the value $\hat{\phi}_b(\bm{M})$ does not depend on the partition $\vec{\bm{N}}$.
Then,
we can show the following theorem.
\begin{theorem}\Label{th1-15}
When Condition \ref{c6} holds
and the function $\hat{\phi}_b$ satisfies (\ref{1-25-1}),
we obtain
\begin{eqnarray}
\Pr_{\bm{J},\bm{M}
%|\bar{q}^{(0)}, \bar{\bm{a}},\bar{\bm{b}}
}
\{
\hat{\phi}_b(\bm{M})
< \phi(\bm{J})
\}
\le 3 \cdot 2^{-2\beta-5}.
\end{eqnarray}
\end{theorem}

\begin{proof}
The definition of $\Omega_1$ yields that
\begin{eqnarray*}
\Pr_{\vec{\bm{N}}} \Omega_1^c 
\le 2 \cdot 5 \cdot 2^{-2\beta-8} .
\end{eqnarray*}
Since
\begin{eqnarray*}
%\Pr_{\bm{M},\bm{J},\vec{\bm{N}}| \bm{a},\bm{b}} 
\fl & \{
\hat{\phi}_b(\bm{M})
< \phi(\bm{J})
\}
\subset
\{\hat{\phi}_a(\bm{M},\vec{\bm{N}}) < \phi(\bm{J}) \}
\cup
\{\hat{\phi}_b(\bm{M}) < \hat{\phi}_a(\bm{M},\vec{\bm{N}}) \} \\
\fl \subset &
\{\hat{\phi}_a(\bm{M},\vec{\bm{N}}) < \phi(\bm{J}) \}
\cup \Omega_1^c 
\subset 
(\{\hat{\phi}_a(\bm{M},\vec{\bm{N}}) < \phi(\bm{J}) \}
\cap \Omega_1)
\cup \Omega_1^c,
\end{eqnarray*}
we have
\begin{eqnarray*}
\fl & \Pr_{\bm{J},\bm{M}
%|\bar{\bm{a}},\bar{\bm{b}}
}
\{
\hat{\phi}_b(\bm{M})
< \phi(\bm{J})
\} \le
\Pr_{\bm{M},\bm{J},\vec{\bm{N}}} 
(\{\hat{\phi}_a(\bm{M},\vec{\bm{N}}) < \phi(\bm{J}) \}
\cap \Omega_1)
+\Pr_{\bm{M},\bm{J},\vec{\bm{N}}} 
 \Omega_1^c \\
\fl 
\le & 8 \cdot 2^{-2\beta-8} + 10 \cdot 2^{-2\beta-8} 
\le  24 \cdot 2^{-2\beta-8}
=3 \cdot 2^{-2\beta-5},
\end{eqnarray*}
which implies the desired argument.
\end{proof}

Therefore, 
when 
$\rho_{A,E}$ is the final state 
with the sacrifice bit-length
\begin{eqnarray}
S(\bm{M}):= \hat{\phi}_b(\bm{M})+2\beta+ 5,
\Label{eq-1-15-1}
\end{eqnarray}
(\ref{9-17-5}) implies that
\begin{eqnarray*}
P_{ph} 
 &\le 
& 2^{-2\beta-5}
\Pr\{ \hat{\phi}_b(\bm{M})+2\beta+ 5 < \phi(\bm{J}) +2\beta+ 5\}^c \\
 &&+
\Pr
\{ \hat{\phi}_b(\bm{M})+2\beta+ 5 < 
\phi(\bm{J}) +2\beta+ 5\} \\
 &=& 2^{-2\beta-5}+3 \cdot 2^{-2\beta-5}
= 2^{-2\beta-3}.
\end{eqnarray*}
Thus, the relation (\ref{9-17-6}) implies
\begin{eqnarray}
\| \rho_{A,E}-\rho_{\ideal} \|_1
\le 2\sqrt{2} 2^{(-2\beta-3)/2} =2^{-\beta}\Label{35}.
\end{eqnarray}

%When the signal intensity is $\mu_2$ and the decoy intensity is $\mu_1$,
%the above arguments hold with modifying (\ref{10-5-10-d2}) and (\ref{10-5-11-d2}) in the following way.
%\begin{eqnarray}
%{N^{(0)}}'&\in [ X_{\per}^{-}(N',e^{-\mu_2},2^{-2\beta-8}), X_{\per}^{+}(N',e^{-\mu_2},2^{-2\beta-8})] \Label{10-5-10-c}
%\\
%{N^{(1)}}'&\in [ X_{\per}^{-}(N',\mu_2 e^{-\mu_2},2^{-2\beta-8}), X_{\per}^{+}(N',\mu_2 e^{-\mu_2},2^{-2\beta-8})]
%\Label{10-5-11-c} .
%\end{eqnarray}
In summary, since Theorem \ref{th1-15} requires Condition \ref{c6},
we need to choose the parameters
$\mu_1$, $\mu_2$, $N_0$, $N_1$, and $N_2$
so that Condition \ref{c6} holds.
That is, we need to choose sufficiently large
integers $N_0$, $N_1$, and $N_2$.
Otherwise, 
we cannot apply Theorem \ref{th1-15}, i.e., we cannot guarantee the security.

The latter sections give a formula of the sacrifice bit-length $S$ 
as a function of $\beta,\mu_s,\mu_1,\mu_2,N_s,N_0,N_1,N_2$, and $\bm{M}$
by giving a concrete example of $\hat{\phi}_b$.
In order to apply interval estimation and percent point,
we have to decide which upper or lower bound to be used in the respective steps.
These decisions will be done based on derivatives for respective variables.
Hence, the calculations of these derivatives are the main issues in the latter sections.

\section{Derivation of upper bound of leaked information}\Label{s7}
\subsection{Case when the channel parameters are given}\Label{s7-1}
The purpose of this section is 
to derive an upper bound $\hat{\phi}_3(\bm{M}, \vec{\bm{N}})$ of the leaked information $\phi$
as an example of $\hat{\phi}_a$.
%based on the channel parameters $\bar{q}^{(0)}$, $\bar{a}^{(1)}_{\times}$, $\bar{b}^{(1)}_{\times}$ and the partition $\bm{N}$ of $N$ pulses.
For this purpose, we describe the leaked information $\phi$
as a function of $J^{(0)}$, $J^{(1)}$, and $r^{(1)}:=J^{(1)}_e/J^{(1)}$:
\begin{eqnarray}
\phi=M_s-J^{(0)}- J^{(1)}(1- h(\min\{r^{(1)},1/2\}) ).
\end{eqnarray}
That is, $\phi$ is monotonically decreasing with respect to $J^{(0)}$ and $J^{(1)}$, and monotonically increasing with respect to $r^{(1)}$.
%In the following, we denote the number of transmitted pulses for generating raw keys by $N$, and denote the signal intensity by $\mu_s$.
Due to the same reason as (\ref{10-13-1}),
the channel parameters
$\bar{q}^{(0)}$, $\bar{a}^{(1)}_{\times}$,
and $\bar{b}^{(1)}_{\times}$ satisfy
\begin{eqnarray}
\Pr_{\bm{J}|\bar{q}^{(0)}, \bar{\bm{a}},\bar{\bm{b}},\vec{\bm{N}}} 
\{ J^{(0)} \le X_{\per}^-(N_s,e^{-\mu_s} \bar{q}^{(0)},2^{-2\beta-8}) \} \le 2^{-2\beta-8} \Label{10-13-5}\\
\Pr_{\bm{J}|\bar{q}^{(0)}, \bar{\bm{a}},\bar{\bm{b}},\vec{\bm{N}}} 
\{ J^{(1)} \le X_{\per}^-(N_s,e^{-\mu_s} \mu_s (\bar{a}^{(1)}_{\times}+\bar{b}^{(1)}_{\times}),2^{-2\beta-8}) \} \le 2^{-2\beta-8} \Label{10-13-3}.
\end{eqnarray}
Using this fact, we give estimates of $J^{(0)}$ and $J^{(1)}$ as
\begin{eqnarray}
\hat{J}^{(0)}(\bar{q}^{(0)},N_s,\mu_s) 
&:= X_{\per}^-(N_s,e^{-\mu_s} \bar{q}^{(0)},2^{-2\beta-8}) \Label{1-24-1}\\
\hat{J}^{(1)}(\bar{a}^{(1)}_{\times},\bar{b}^{(1)}_{\times},N_s,\mu_s) 
&:= X_{\per}^-(N_s,e^{-\mu_s} \mu_s (\bar{a}^{(1)}_{\times}+\bar{b}^{(1)}_{\times}),2^{-2\beta-8}) \Label{1-24-2},
\end{eqnarray}
which provide Step (4) in Figs. \ref{outline-zu2} and \ref{outline-zu}.
Since $p_{\per}^+(j^{(1)},\frac{\bar{b}^{(1)}_{\times}}{\bar{a}^{(1)}_{\times}+\bar{b}^{(1)}_{\times}},2^{-2\beta-8})
\le
p_{\per}^+(\tilde{j}^{(1)},\frac{\bar{b}^{(1)}_{\times}}{\bar{a}^{(1)}_{\times}+\bar{b}^{(1)}_{\times}},2^{-2\beta-8})$
holds for $j^{(1)} \ge \tilde{j}^{(1)}$,
(\ref{10-13-3}) implies that
\begin{eqnarray}
\fl & \Pr_{\bm{J}|\bar{q}^{(0)}, \bar{\bm{a}},\bar{\bm{b}},\vec{\bm{N}}} 
\Bigl\{  
p_{\per}^+\Bigl( J^{(1)},\frac{\bar{b}^{(1)}_{\times}}{\bar{a}^{(1)}_{\times}+\bar{b}^{(1)}_{\times}},2^{-2\beta-8} \Bigr) 
\ge p_{\per}^+\Bigl(\hat{J}^{(1)}(\bar{a}^{(1)}_{\times},\bar{b}^{(1)}_{\times},N_s,\mu_s ) ,
\frac{\bar{b}^{(1)}_{\times}}{\bar{a}^{(1)}_{\times}+\bar{b}^{(1)}_{\times}},2^{-2\beta-8}\Bigr) \Bigr\} \nonumber \\
\fl =& \Pr_{\bm{J}|\bar{q}^{(0)}, \bar{\bm{a}},\bar{\bm{b}},\vec{\bm{N}}} 
\{ J^{(1)} \le X_{\per}^-(N_s,e^{-\mu_s} \mu_s (\bar{a}^{(1)}_{\times}+\bar{b}^{(1)}_{\times}),2^{-2\beta-8}) \} 
\le 2^{-2\beta-8} \nonumber .
\end{eqnarray}
Using the relation
\begin{eqnarray}
\Pr_{\bm{J}|\bar{q}^{(0)}, \bar{\bm{a}},\bar{\bm{b}},\vec{\bm{N}}} 
\Bigl\{ \frac{J^{(1)}_e}{J^{(1)}} \ge 
p_{\per}^+\Bigl(J^{(1)},\frac{\bar{b}^{(1)}_{\times}}{\bar{a}^{(1)}_{\times}+\bar{b}^{(1)}_{\times}},2^{-2\beta-8}
\Bigr) \Bigr\} \le 2^{-2\beta-8} 
\Label{10-13-6},
\end{eqnarray}
we obtain
\begin{eqnarray}
\fl &
\Bigl\{ \frac{J^{(1)}_e}{J^{(1)}} \ge 
p_{\per}^+\Bigl(\hat{J}^{(1)}(\bar{a}^{(1)}_{\times},\bar{b}^{(1)}_{\times},N_s,\mu_s),
\frac{\bar{b}^{(1)}_{\times}}{\bar{a}^{(1)}_{\times}+\bar{b}^{(1)}_{\times}},2^{-2\beta-8}\Bigr) \Bigr\} \nonumber \\
\fl \subset &
\Biggl(
\Bigl\{  p_{\per}^+\Bigl(J^{(1)},\frac{\bar{b}^{(1)}_{\times}}{\bar{a}^{(1)}_{\times}+\bar{b}^{(1)}_{\times}},2^{-2\beta-8}\Bigr) 
\ge p_{\per}^+\Bigl(\hat{J}^{(1)}(\bar{a}^{(1)}_{\times},\bar{b}^{(1)}_{\times},N_s,\mu_s) ,
\frac{\bar{b}^{(1)}_{\times}}{\bar{a}^{(1)}_{\times}+\bar{b}^{(1)}_{\times}},2^{-2\beta-8}\Bigr) \Bigr\} \nonumber \\
\fl & \quad
\cup
\Bigl\{ \frac{J^{(1)}_e}{J^{(1)}} \ge 
p_{\per}^+\Bigl(J^{(1)},\frac{\bar{b}^{(1)}_{\times}}{\bar{a}^{(1)}_{\times}+\bar{b}^{(1)}_{\times}},2^{-2\beta-8}\Bigr) 
\Bigr\} 
\Biggr)
\nonumber \\
\fl \subset &
\{ J^{(1)} \le X_{\per}^-(N_s,e^{-\mu_s}\mu_s (\bar{a}^{(1)}_{\times}+\bar{b}^{(1)}_{\times}),2^{-2\beta-8}) \} 
\cup
\Bigl\{ \frac{J^{(1)}_e}{J^{(1)}} \ge 
p_{\per}^+\Bigl(J^{(1)},\frac{\bar{b}^{(1)}_{\times}}{\bar{a}^{(1)}_{\times}+\bar{b}^{(1)}_{\times}},2^{-2\beta-8}\Bigr) 
\Bigr\}.\nonumber \\
\fl &
%\fl \le & 2 \cdot 2^{-2\beta-8} 
\Label{10-13-4}
\end{eqnarray}
Therefore, we give an estimate of $r^{(1)}:=\frac{J^{(1)}_e}{J^{(1)}}$ by
\begin{eqnarray}
\fl \hat{r}^{(1)}_{\times}(\bar{a}^{(1)}_{\times},\bar{b}^{(1)}_{\times},N_s,\mu_s)  
&:= p^+_{\per}\Bigl(\hat{J}^{(1)}(\bar{a}^{(1)}_{\times},\bar{b}^{(1)}_{\times},N_s,\mu_s) , 
\frac{\bar{b}^{(1)}_{\times}}{\bar{a}^{(1)}_{\times}+\bar{b}^{(1)}_{\times}}
, 2^{-2\beta-8}\Bigr).
\Label{1-24-3}
\end{eqnarray}

Using the above relations,
we give an estimate of $\phi$ by
\begin{eqnarray}
\fl &
\hat{\phi}_1(\bar{q}^{(0)},\bar{a}^{(1)}_{\times},\bar{b}^{(1)}_{\times},N_s,\mu_s) \nonumber \\
\fl
:= & M_s-\hat{J}^{(0)}(\bar{q}^{(0)},N_s,\mu_s)  
- \hat{J}^{(1)}(\bar{a}^{(1)}_{\times},\bar{b}^{(1)}_{\times},N_s,\mu_s)  
(1- h(\min \{\hat{r}^{(1)}_{\times}(\bar{a}^{(1)}_{\times},\bar{b}^{(1)}_{\times},N_s,\mu_s)  ,1/2\}) ) .\nonumber \\
\fl & \Label{9-24-3}
\end{eqnarray}
Due to (\ref{10-13-5}), (\ref{10-13-3}), (\ref{10-13-4}), and (\ref{10-13-6}),
the estimate $\hat{\phi}_1(\bar{q}^{(0)},\bar{a}^{(1)}_{\times},\bar{b}^{(1)}_{\times},N_s,\mu_s) $ 
satisfies
\begin{eqnarray}
\fl & 
\{\phi (\bm{J}) >
\hat{\phi}_1(\bar{q}^{(0)},\bar{a}^{(1)}_{\times},\bar{b}^{(1)}_{\times},N_s,\mu_s) 
\} \nonumber \\
\fl \subset & 
\Biggl ( \{ J^{(0)} \le X_{\per}^-(N_s,e^{-\mu_s} \bar{q}^{(0)},2^{-2\beta-8}) \} 
\cup
\{ J^{(1)} \le X_{\per}^-(N_s,e^{-\mu_s} \mu_s (\bar{a}^{(1)}_{\times}+\bar{b}^{(1)}_{\times}),2^{-2\beta-8}) \} \nonumber \\
\fl &\cup 
\Bigl\{ \frac{J^{(1)}_e}{J^{(1)}} \ge 
p_{\per}^+\Bigl(\hat{J}^{(1)}(\bar{a}^{(1)}_{\times},\bar{b}^{(1)}_{\times},N_s,\mu_s),
\frac{\bar{b}^{(1)}_{\times}}{\bar{a}^{(1)}_{\times}+\bar{b}^{(1)}_{\times}}
,2^{-2\beta-8}\Bigr) \Bigr\} \Biggr) \nonumber \\
\fl \subset & 
\Biggl(
\{ J^{(0)} \le X_{\per}^-(N_s,e^{-\mu_s} \bar{q}^{(0)},2^{-2\beta-8}) \} 
\cup
\{ J^{(1)} \le X_{\per}^-(N_s,e^{-\mu_s}\mu_s( \bar{a}^{(1)}_{\times}+\bar{b}^{(1)}_{\times}),2^{-2\beta-8}) \} \nonumber \\
\fl & \quad \cup 
\Bigl\{ \frac{J^{(1)}_e}{J^{(1)}} \ge 
p_{\per}^+\Bigl(J^{(1)},\frac{\bar{b}^{(1)}_{\times}}{\bar{a}^{(1)}_{\times}+\bar{b}^{(1)}_{\times}}
,2^{-2\beta-8}\Bigr) \Bigr\} 
\Biggr)
\Label{9-24-1-b}.
\end{eqnarray}
Hence, we obtain
\begin{eqnarray}
\Pr 
\{\phi (\bm{J}) >
\hat{\phi}_1(\bar{q}^{(0)},\bar{a}^{(1)}_{\times},\bar{b}^{(1)}_{\times},N_s,\mu_s)
\} 
\le 3\cdot 2^{-2\beta-8} \Label{9-24-1}.
\end{eqnarray}

\subsection{Estimation of channel parameters $\bar{q}^{(0)}$, $\bar{a}^{(1)}_{\times}$, and $\bar{r}^{(1)}_{\times}$}\Label{s5}
Next, in order to treat an upper bound of leaked information $\phi$,
we will give estimates of channel parameters
$\bar{q}^{(0)},\bar{a}^{(1)}_{\times}$, and $\bar{r}^{(1)}_{\times}$
based on the measured values $\bm{M}$ 
and the partition $\vec{\bm{N}}$ of pulses.
For this estimation, we employ the one-sided interval estimation.
We have to decide which 
the lower one-sided interval estimator or
the upper one-sided interval estimator is used for the respective channel parameters.
These decisions will be done by 
the signs of 
the partial derivatives of 
$\hat{\phi}_1(\bar{q}^{(0)},\bar{a}^{(1)}_{\times},\bar{b}^{(1)}_{\times},N_s,\mu_s)$
with respect to 
$\bar{q}^{(0)},\bar{a}^{(1)}_{\times}$, and $\bar{b}^{(1)}_{\times}$, respectively.
%we consider how the leaked information $\hat{\phi}_1(\bm{N},\bar{q}^{(0)},\bar{a}^{(1)}_{\times},\bar{b}^{(1)}_{\times}, 2^{-2\beta-8}) )$ reflects small changes of the channel parameters $\bar{q}^{(0)},\bar{a}^{(1)}_{\times}$, and $\bar{b}^{(1)}_{\times}$.

While the partial derivatives of 
$\hat{J}^{(0)}(\bar{q}^{(0)},N_s,\mu_s)$,
$\hat{J}^{(1)}(\bar{a}^{(1)}_{\times},\bar{b}^{(1)}_{\times},N_s,\mu_s)$,
and
$\hat{r}^{(1)}_{\times}(\bar{a}^{(1)}_{\times},\bar{b}^{(1)}_{\times},N_s,\mu_s) $ 
are needed, their calculations are not easy.
When $N_s$ is sufficiently large,
these values take the almost same values as
$\bar{q}^{(0)} e^{-\mu_s} N_s$,
$(\bar{a}^{(1)}_{\times}+\bar{b}^{(1)}_{\times}) e^{-\mu_s} \mu_s N_s$,
and 
$\frac{\bar{b}^{(1)}_{\times}}{\bar{a}^{(1)}_{\times}+\bar{b}^{(1)}_{\times}}$,
and 
the variations due to the fluctuations of
$\bar{q}^{(0)},\bar{a}^{(1)}_{\times}$, and $\bar{b}^{(1)}_{\times}$ are negligible.

Hence, we can regard the derivatives of
$\hat{J}^{(0)}(\bar{q}^{(0)},N_s,\mu_s)$,
$\hat{J}^{(1)}(\bar{a}^{(1)}_{\times},\bar{b}^{(1)}_{\times},N_s,\mu_s) $,
and 
$\hat{r}^{(1)}_{\times}(\bar{a}^{(1)}_{\times},\bar{b}^{(1)}_{\times},N_s,\mu_s) $
as 
the same as those of
$\bar{q}^{(0)} e^{-\mu_s} N_s $,
$(\bar{a}^{(1)}_{\times}+\bar{b}^{(1)}_{\times}) e^{-\mu_s} \mu_s N_s$,
and 
$\frac{\bar{b}^{(1)}_{\times}}{\bar{a}^{(1)}_{\times}+\bar{b}^{(1)}_{\times}}$.
Thus, we obtain
\begin{eqnarray*}
\fl
\frac{\partial \hat{J}^{(0)}(\bar{q}^{(0)},N_s,\mu_s)}{\partial \bar{q}^{(0)}}
&=e^{-\mu_s} N_s , &\\
\fl
\frac{\partial \hat{J}^{(1)}(\bar{a}^{(1)}_{\times},\bar{b}^{(1)}_{\times},N_s,\mu_s)}{\partial \bar{a}^{(1)}_{\times}}
&=e^{-\mu_s} \mu_s N_s  ,\quad
&\frac{\partial \hat{J}^{(1)}(\bar{a}^{(1)}_{\times},\bar{b}^{(1)}_{\times},N_s,\mu_s)}{\partial \bar{b}^{(1)}_{\times}}
=e^{-\mu_s} \mu_s N_s  ,\\
\fl
\frac{\partial \hat{r}^{(1)}_{\times}(\bar{a}^{(1)}_{\times},\bar{b}^{(1)}_{\times},N_s,\mu_s) }{\partial \bar{a}^{(1)}_{\times}}
&=-\frac{\bar{b}^{(1)}_{\times}}{(\bar{a}^{(1)}_{\times}+\bar{b}^{(1)}_{\times})^2}
,\quad
&\frac{\partial \hat{r}^{(1)}_{\times}(\bar{a}^{(1)}_{\times},\bar{b}^{(1)}_{\times},N_s,\mu_s) }{\partial \bar{b}^{(1)}_{\times}}
=\frac{\bar{a}^{(1)}_{\times}}{(\bar{a}^{(1)}_{\times}+\bar{b}^{(1)}_{\times})^2}.
\end{eqnarray*}
Under this assumption,
we have
\begin{eqnarray}
\fl \frac{\partial \hat{\phi}_1(\bar{q}^{(0)},\bar{a}^{(1)}_{\times},\bar{b}^{(1)}_{\times},N_s,\mu_s)}{\partial \bar{q}^{(0)}} 
&=  
\frac{\partial \phi}{\partial J^{(0)}}
\frac{\partial \hat{J}^{(0)}}{\partial \bar{q}^{(0)}}
= -e^{-\mu_s} N_s  < 0 \Label{10-25-1}
\\
\fl \frac{\partial \hat{\phi}_1(\bar{q}^{(0)},\bar{a}^{(1)}_{\times},\bar{b}^{(1)}_{\times},N_s,\mu_s)}{\partial \bar{a}^{(1)}_{\times}} 
&= 
\frac{\partial \phi}{\partial J^{(1)}}
\frac{\partial \hat{J}^{(1)}}{\partial \bar{a}^{(1)}_{\times}}
+
\frac{\partial \phi}{\partial \bar{r}^{(1)}_{\times}}
\frac{\partial \hat{r}^{(1)}_{\times}}{\partial \bar{a}^{(1)}_{\times}}\nonumber \\
&=
-e^{-\mu_s} \mu_s N_s  (1+\log \frac{\bar{a}^{(1)}_{\times}}{\bar{a}^{(1)}_{\times}+\bar{b}^{(1)}_{\times}})
< 0 \Label{10-25-2}
\\
\fl \frac{\partial \hat{\phi}_1(\bar{q}^{(0)},\bar{a}^{(1)}_{\times},\bar{b}^{(1)}_{\times},N_s,\mu_s)}{\partial \bar{b}^{(1)}_{\times}} 
&= 
\frac{\partial \phi}{\partial J^{(1)}}
\frac{\partial \hat{J}^{(1)}}{\partial \bar{b}^{(1)}_{\times}}
+
\frac{\partial \phi}{\partial \bar{r}^{(1)}_{\times}}
\frac{\partial \hat{r}^{(1)}_{\times}}{\partial \bar{b}^{(1)}_{\times}} \nonumber \\
&=
-e^{-\mu_s} \mu_s N_s  (1+\log \frac{\bar{b}^{(1)}_{\times}}{\bar{a}^{(1)}_{\times}+\bar{b}^{(1)}_{\times}})
> 0 \Label{10-25-3}
\end{eqnarray}
because 
$\frac{\bar{b}^{(1)}_{\times}}{\bar{a}^{(1)}_{\times}+\bar{b}^{(1)}_{\times}}<\frac{1}{2}$.
Therefore, smaller $\bar{q}^{(0)}$ and $\bar{a}^{(1)}_{\times}$ 
yield larger $\hat{\phi}_1$,
and larger $\bar{b}^{(1)}_{\times}$ yields larger $\hat{\phi}_1$.
That is,
it is needed to estimate 
$\bar{q}^{(0)}$ and $\bar{a}^{(1)}_{\times}$ to be smaller
and $\bar{b}^{(1)}_{\times}$ to be larger.

In the following,
we treat the estimation of $\bar{q}^{(0)}$, $\bar{a}^{(1)}_{\times}$, and $\bar{b}^{(1)}_{\times}$.
Here, we should remark that
we have two kinds of channel parameters.
The first kind of parameters are 
$\bar{q}^{(0)}$, $\bar{a}^{(1)}_{\times}$, and $\bar{b}^{(1)}_{\times}$,
which are directly linked to the eavesdropping
and cannot be measured directly.
The second kind of parameters are
the detection rates $p_{1,\times}$ and $p_{2,\times}$ of pulses 
of the phase basis with intensities $\mu_1$ and $\mu_2$,
which can be measured directly.
Similarly, 
as the second latter kind of parameters, we have the rates $s_{1,\times}$ and $s_{2,\times}$
of the detected pulses having phase error with intensities $\mu_1$ and $\mu_2$
to the transmitted pulses 
with intensities $\mu_1$ and $\mu_2$, respectively.
The rates $s_{1,\times}$ and $s_{2,\times}$ also can be measured directly.
Hence, 
the expectations 
$\bar{M}_0$,
$\bar{M}_1$,
$\bar{M}_2$,
and 
$\bar{M}_3$
%, and $\bar{M}_4$
of
$M_0$, $M_1$, $M_2$, and $M_3$
are characterized by
$\bar{M}_0=p_0 N_0$,
$\bar{M}_1=(p_{1,\times}-s_{1,\times}) N_1$,
$\bar{M}_2=(p_{2,\times}-s_{2,\times}) N_2$,
and 
$\bar{M}_3=s_{1,\times} N_1$. %, and $\bar{M}_4=s_{2,\times} N_4$.
Thus, we can regard 
$\bar{M}_0$,
$\bar{M}_1$,
$\bar{M}_2$,
and $\bar{M}_3$ %, and $\bar{M}_4$
as the second kind of channel parameters.
Since 
$\bar{M}_0$,
$\bar{M}_1$,
$\bar{M}_2$,
and $\bar{M}_3$ %, and $\bar{M}_4$
are easier to treat than 
$p_0$, $p_{1,\times}$, $p_{2,\times}$, $s_{1,\times}$, and $s_{2,\times}$,
we can estimate an upper bound of
the leaked information $\phi$
via the estimation of the channel parameters
$\bar{q}^{(0)}$, $\bar{a}^{(1)}_{\times}$, and $\bar{b}^{(1)}_{\times}$ 
when 
the channel parameters $\bar{\bm{M}}:=(\bar{M}_0$, $\bar{M}_1$, $\bar{M}_2$, $\bar{M}_3)$ %, $\bar{M}_4)$
and the partition $\vec{\bm{N}}$ of pulses
are given.
Then, using the expansion formula (\ref{4-9-4}), we obtain
\begin{eqnarray}
\bar{M}_0= & \bar{q}^{(0)} N_0 \Label{9-17-1}\\
\bar{M}_1= & \frac{\bar{q}^{(0)}}{2} N^{(0)}_1 + \bar{a}^{(1)}_{\times} N^{(1)}_1 + 
\bar{a}^{(2)}_{\times} N^{(2)}_1 \Label{9-17-2} \\
\bar{M}_2= & \frac{\bar{q}^{(0)}}{2} N^{(0)}_2 + \bar{a}^{(1)}_{\times} N^{(1)}_2 + 
\bar{a}^{(2)}_{\times} N^{(2)}_2 + \bar{a}^{(3)}_{\times} N^{(3)}_2,
\Label{9-17-3}
\end{eqnarray}
which imply the matrix equation
\begin{eqnarray}
\fl \left(
\begin{array}{l}
\bar{M}_0 \\
\bar{M}_1 \\
\bar{M}_2 - \bar{a}^{(3)}_{\times} N^{(3)}_2
\end{array}
\right)
=
\left(
\begin{array}{ccc}
N_0   & 0     & 0 \\
N^{(0)}_1/2 & N^{(1)}_1 & N^{(2)}_1\\
N^{(0)}_2/2 & N^{(1)}_2 & N^{(2)}_2 
\end{array}
\right)
\left(
\begin{array}{l}
\bar{q}^{(0)} \\
\bar{a}^{(1)}_{\times} \\
\bar{a}^{(2)}_{\times}
\end{array}
\right).
\end{eqnarray}
Solving the above, we obtain
\begin{eqnarray}
\fl \bar{q}^{(0)}=& \hat{q}^{(0)}(\bar{M}_0):=\frac{\bar{M}_0}{N_0} 
\Label{1-23-4} \\
\fl \bar{a}^{(1)}_{\times}=&
\tilde{a}^{(1)}_{\times}(\bar{\bm{M}},\vec{\bm{N}})+ A_1 \bar{a}^{(3)}_{\times}\\
\fl \tilde{a}^{(1)}_{\times}(\bar{\bm{M}},\vec{\bm{N}})
:=&
\frac{ N^{(2)}_2(\hat{M}_1 - \hat{q}^{(0)} (\bar{M}_0) \hat{N}^{(0)}_1 / 2 )
-\hat{N}^{(2)}_1 (\hat{M}_2 - \hat{q}^{(0)} (\bar{M}_0) \hat{N}^{(0)}_2 / 2 )
}{\hat{N}^{(1)}_1 \hat{N}^{(2)}_2- \hat{N}^{(1)}_2 \hat{N}^{(2)}_1} 
%\frac{ N^{(2)}_2(\bar{M}_1 - \bar{M}_0 N^{(0)}_1 / 2N_0)-N^{(2)}_1 (\bar{M}_2 - \bar{M}_0 N^{(0)}_2 / 2N_0)
%}{N^{(1)}_1 N^{(2)}_2- N^{(1)}_2 N^{(2)}_1} 
\Label{1-23-4-2}
\\
\fl A_1:=&
\frac{N^{(2)}_1  N^{(3)}_2}{N^{(1)}_1 N^{(2)}_2- N^{(1)}_2 N^{(2)}_1}.
\end{eqnarray}
Since $A_1$ and $\bar{a}^{(3)}_{\times}$ are non-negative,
we obtain a lower bound of $\bar{a}^{(1)}_{\times}$.
\begin{eqnarray}
\bar{a}^{(1)}_{\times}
\ge 
\hat{a}^{(1)}(\bar{\bm{M}},\vec{\bm{N}})
:=
[\tilde{a}^{(1)}(\bar{\bm{M}},\vec{\bm{N}})]_+
.
\end{eqnarray}

Similarly, we have
\begin{eqnarray}
\bar{M}_3= & \bar{q}^{(0)} N^{(0)}_1/2 + \bar{b}^{(1)}_{\times} N^{(1)}_1 
+ \bar{b}^{(2)}_{\times} N^{(2)}_1 \Label{9-17-2-2} .
\end{eqnarray}
Since $\bar{b}^{(2)}_{\times}$ is non-negative,
we obtain an upper bound of $\bar{b}^{(1)}_{\times}$
as
\begin{eqnarray}
\bar{b}^{(1)}_{\times}
\le
\hat{b}^{(1)}_{\times}(\bar{\bm{M}},\vec{\bm{N}}) 
:=
\left[\frac{\bar{M}_3 - \frac{1}{2} \hat{q}^{(0)} (\bar{M}_0) N^{(0)}_1}{N^{(1)}_1}
\right]_+
\Label{10-8-20}.
\end{eqnarray}
That is, (\ref{1-23-4}), (\ref{1-23-4-2}), and (\ref{10-8-20}) give Step (3) in Figs. \ref{outline-zu2} and \ref{outline-zu}.
Then, we can define an upper bound of $\hat{\phi}_1$ as
\begin{eqnarray}
\fl & \hat{\phi}_2(\hat{\bm{M}}(\bm{M}),\vec{\bm{N}})
:=
\hat{\phi}_1(\hat{q}^{(0)}(\bar{M}_0),
\hat{a}^{(1)}_{\times}(\hat{\bm{M}}(\bm{M}),\vec{\bm{N}}),
\hat{b}^{(1)}_{\times}(\hat{\bm{M}}(\bm{M}),\vec{\bm{N}}),
N_s,\mu_s) \nonumber\\
\fl 
=& M_s-\hat{J}^{(0)}(\hat{q}^{(0)}(\bar{M}_0),N_s,\mu_s) \nonumber\\
\fl &- \hat{J}^{(1)}(
\hat{a}^{(1)}_{\times}(\hat{\bm{M}}(\bm{M}),\vec{\bm{N}})
,
\hat{b}^{(1)}_{\times}(\hat{\bm{M}}(\bm{M}),\vec{\bm{N}})
,N_s,\mu_s) \nonumber\\
\fl & \quad \times (1- h(\min \{\hat{r}^{(1)}_{\times}(\hat{a}^{(1)}_{\times}(\hat{\bm{M}}(\bm{M}),\vec{\bm{N}}),
\hat{b}^{(1)}_{\times}(\hat{\bm{M}}(\bm{M}),\vec{\bm{N}}),N_s,\mu_s) ,1/2\}) ) ,%\nonumber \\
%\fl & 
\Label{10-26-1}
\end{eqnarray}
where $\hat{\bm{M}}(\bm{M})$ is the estimate of $\bar{\bm{M}}$ when $\bm{M}$ is observed.

Indeed, 
when 
$\bar{a}^{(1)}_{\times}
=\hat{a}^{(1)}_{\times}(\bar{\bm{M}},\vec{\bm{N}})$
and
$\bar{b}^{(1)}_{\times}
=\hat{b}^{(1)}_{\times}(\bar{\bm{M}},\vec{\bm{N}})$,
the relations
\begin{eqnarray}
\bar{M}_1= & \frac{\bar{M}_0}{2N_0} N^{(0)}_1 + \hat{a}^{(1)}_{\times} N^{(1)}_1 + 
\hat{a}^{(2)}_{\times} N^{(2)}_1 \Label{9-17-2-a} \\
\bar{M}_2= & \frac{\bar{M}_0}{2N_0} N^{(0)}_2 + \hat{a}^{(1)}_{\times} N^{(1)}_2 + 
\hat{a}^{(2)}_{\times} N^{(2)}_2 \Label{9-17-3-a} \\
\bar{M}_3= & \frac{\bar{M}_0}{2N_0} N^{(0)}_1
+ \hat{b}^{(1)}_{\times} N^{(1)}_1 
\Label{9-17-2-b}
\end{eqnarray}
hold.

\begin{rem}
When we extend the existing method\cite{Ma05,Wang05,H1,decoy2,decoy4} to our finite length setting,
we obtain the following evaluation.
In this case, we employ the parameter $\bar{q}^{(1)}$ instead of 
$\bar{a}^{(1)}_{\times}$.
Because smaller $\bar{q}^{(1)} $ yields larger $\hat{\phi}_1$,
similar to $\bar{a}^{(1)}_{\times}$, $\bar{q}^{(1)}$ can be estimated as
\begin{eqnarray}
\fl \hat{q}^{(1)}(\bar{\bm{M}},\vec{\bm{N}})
:=
\frac{ N^{(2)}_2(\bar{M}_1+\bar{M}_3 - \bar{M}_0 N^{(0)}_1 / N_0)
-N^{(2)}_1 (\bar{M}_2 +\bar{M}_4- \bar{M}_0 N^{(0)}_2 / N_0)
}{N^{(1)}_1 N^{(2)}_2- N^{(1)}_2 N^{(2)}_1} .
\end{eqnarray}
Then,
in the upper bound $\hat{\phi}_1$
of the sacrifice bit size,
$\hat{a}^{(1)}_{\times}+\hat{b}^{(1)}_{\times}$
is replaced by 
$\hat{q}^{(1)}(\bar{\bm{M}},\vec{\bm{N}})$.
That is, we obtain 
an upper bound 
\begin{eqnarray}
\hat{\phi}_1(\hat{q}^{(0)}(\bar{M}_0),
\hat{q}^{(1)}(\bar{\bm{M}},\vec{\bm{N}})-
\hat{b}^{(1)}_{\times}(\bar{\bm{M}},\vec{\bm{N}}),
\hat{b}^{(1)}_{\times}(\bar{\bm{M}},\vec{\bm{N}}),N_s,\mu_s),
\end{eqnarray}
which is larger than
$\hat{\phi}_2(\hat{\bm{M}}(\bm{M}),\vec{\bm{N}})$
because
$\hat{q}^{(1)}(\bar{\bm{M}},\vec{\bm{N}})-
\hat{b}^{(1)}_{\times}(\bar{\bm{M}},\vec{\bm{N}})
\le \hat{a}^{(1)}_{\times}(\bar{\bm{M}},\vec{\bm{N}})$.
\end{rem}

\subsection{Estimation of another kind of channel parameters $\bar{\bm{M}}$}\Label{s7-3}
In this subsection, we treat the estimation of the channel parameters $\bar{\bm{M}}$
that is required to estimate the channel parameters
$\bar{q}^{(0)}$, $\bar{a}^{(1)}_{\times}$, and $\bar{b}^{(1)}_{\times}$
when the partition $\vec{\bm{N}}=(\bm{N}_1$, $\bm{N}_2)$
of pulses is known.
That is, we consider the method to estimate
$\bar{M}_0$, $\bar{M}_1$, $\bar{M}_2$, and $\bar{M}_3$
from the measured value
$M_0,M_1,M_2,$ and $M_3$.

For this purpose, we introduce the following assumption for $\bar{\bm{M}}$.
%We also assume a condition for the estimates $\hat{a}^{(1)}_{\times}$ and $\hat{b}^{(1)}_{\times}$ of $\bar{a}^{(1)}_{\times}$ and $\bar{b}^{(1)}_{\times}$ as follows.
\begin{condition}\Label{c11}
Any element $\vec{\bm{N}} \in \Omega_1$ satisfies
\begin{eqnarray}
\frac{\hat{b}^{(1)}_{\times}(\bar{\bm{M}},\vec{\bm{N}})}{\hat{a}^{(1)}_{\times}(\bar{\bm{M}},\vec{\bm{N}})
+\hat{b}^{(1)}_{\times}(\bar{\bm{M}},\vec{\bm{N}})}
\le \frac{1}{8}.
\end{eqnarray}
\end{condition}

The partial derivatives of
$\hat{q}^{(0)}(\bar{M}_0)$,
$\hat{a}^{(1)}_{\times}(\bar{\bm{M}},\vec{\bm{N}})$, 
and $\hat{b}^{(1)}_{\times}(\bar{\bm{M}},\vec{\bm{N}})$
are calculated as
\begin{eqnarray}
\frac{\partial \hat{q}^{(0)}}{\partial \bar{M}_0}=&\frac{1}{N_0}>0
\Label{54} \\
\frac{\partial \hat{q}^{(0)}}{\partial \bar{M}_1}=&
\frac{\partial \hat{q}^{(0)}}{\partial \bar{M}_2}=
\frac{\partial \hat{q}^{(0)}}{\partial \bar{M}_3}=
\frac{\partial \hat{q}^{(0)}}{\partial \bar{M}_4}=0, 
\Label{60-2} \\
\frac{\partial \hat{a}^{(1)}_{\times}}{\partial \bar{M}_0}
=&\frac{(N_1^{(2)} N^{(0)}_2-N^{(2)}_2 N^{(0)}_1)}{2N_0(N^{(1)}_1 N^{(2)}_2-N^{(1)}_2 N_1^{(2)})}
\Label{55} \\
\frac{\partial \hat{a}^{(1)}_{\times}}{\partial \bar{M}_1}=&
\frac{N^{(2)}_2}{N^{(1)}_1 N^{(2)}_2-N^{(1)}_2 N_1^{(2)} }
\Label{58} \\
\frac{\partial \hat{a}^{(1)}_{\times}}{\partial \bar{M}_2}=&
-\frac{N_1^{(2)} }{N^{(1)}_1 N^{(2)}_2-N^{(1)}_2 N_1^{(2)}}
\Label{61} \\
\frac{\partial \hat{a}^{(1)}_{\times}}{\partial \bar{M}_3}=& 0 ,\\
\frac{\partial \hat{b}^{(1)}_{\times}}{\partial \bar{M}_0}=&
-\frac{N_1^{(0)}}{2 N_0 N^{(1)}_1} <0\\
\frac{\partial \hat{b}^{(1)}_{\times}}{\partial \bar{M}_3}=&
\frac{1}{N^{(1)}_1} >0 \\
\frac{\partial \hat{b}^{(1)}_{\times}}{\partial \bar{M}_1}=&
\frac{\partial \hat{b}^{(1)}_{\times}}{\partial \bar{M}_2}=0.
\end{eqnarray}
Hence, applying Conditions \ref{c1} and \ref{c2} to (\ref{55}),(\ref{58}),(\ref{61}),
we obtain
\begin{eqnarray*}
\frac{\partial \hat{a}^{(1)}_{\times}}{\partial \bar{M}_0}<0 ,~
\frac{\partial \hat{a}^{(1)}_{\times}}{\partial \bar{M}_1}>0 ,~
\frac{\partial \hat{a}^{(1)}_{\times}}{\partial \bar{M}_2}<0.
\end{eqnarray*}
Thus, due to (\ref{10-25-1}), (\ref{10-25-2}), (\ref{10-25-3}),
and Conditions \ref{c10} and \ref{c11},
we obtain
\begin{eqnarray*}
\fl \frac{\partial \hat{\phi}_1}{\partial \bar{M}_1} &<&0 ,~
\frac{\partial \hat{\phi}_1}{\partial \bar{M}_2}>0 ,~
\frac{\partial \hat{\phi}_1}{\partial \bar{M}_3}>0 ,\\
\fl \frac{\partial \hat{\phi}_1}{\partial \bar{M}_0}
&=& 
-\frac{N_s e^{-\mu_s}}{N_0}
-\frac{N_s e^{-\mu_s} \mu_s(N_1^{(2)} N_2^{(0)}-N_2^{(2)} N_1^{(0)})(1+\log \frac{\hat{a}^{(1)}_{\times}}{\hat{a}^{(1)}_{\times}+\hat{b}^{(1)}_{\times}})}{2N_0(N_1^{(1)} N_2^{(2)}-N_1^{(2)} N_2^{(1)})}\\
\fl &&+\frac{N_s e^{-\mu_s} \mu_s N_1^{(0)} (1+\log \frac{\hat{b}^{(1)}_{\times}}{\hat{a}^{(1)}_{\times}+\hat{b}^{(1)}_{\times}})}{2N_0 N_1^{(1)}} \\
\fl &\le &
-\frac{N_s e^{-\mu_s} }{N_0}
-\frac{N_s e^{-\mu_s} \mu_s (N_1^{(2)} N_2^{(0)}-N_2^{(2)} N_1^{(0)})
%(1+\log \frac{\hat{a}^{(1)}_{\times}}{\hat{a}^{(1)}_{\times}+\hat{b}^{(1)}_{\times}})
}{2N_0(N_1^{(1)} N_2^{(2)}-N_1^{(2)} N_2^{(1)})}
-\frac{N_s e^{-\mu_s} \mu_s N_1^{(0)} }{N_0 N_1^{(1)}} 
<0.
\end{eqnarray*}
Hence, we need to estimate 
$\bar{M}_0$ and $\bar{M}_1$ to be smaller, and  
$\bar{M}_2$ and $\bar{M}_3$ to be larger.

In the following, we employ
\begin{eqnarray}
\hat{M}_0(M_0) &:=X_{\est}^-(N_0,M_0,2^{-2\beta-8}),\Label{10-4-0} \\
\hat{M}_1(M_1) &:=X_{\est}^-(N_1,M_1,2^{-2\beta-8}),\Label{10-4-1} \\
\hat{M}_2(M_2) &:=X_{\est}^+(N_2,M_2,2^{-2\beta-8}),\Label{10-4-2} \\
\hat{M}_3(M_3) &:=X_{\est}^+(N_1,M_3,2^{-2\beta-8}),\Label{10-4-3} 
\end{eqnarray}
as estimates of 
$\bar{M}_0$,
$\bar{M}_1$,
$\bar{M}_2$,
and 
$\bar{M}_3$, which give Step (1) in Figs. \ref{outline-zu2} and \ref{outline-zu}.
Then,
we define the function $\hat{\phi}_3$
\begin{eqnarray*}
\hat{\phi}_3(\bm{M},\vec{\bm{N}})
:=
\hat{\phi}_2(\hat{\bm{M}}(\bm{M}),\vec{\bm{N}}),\Label{3-7-1}
\end{eqnarray*}
which satisfies the condition (\ref{1-24-6a}) for 
$\hat{\phi}_a$,
as is guaranteed by the following theorem.
\begin{theorem}\Label{thm5-1-1}
When the partition $\vec{\bm{N}}$ belongs to $\Omega_1$ and satisfies Conditions \ref{c1} and \ref{c2},
and when $\hat{\bm{M}}(\bm{M})$ satisfies Condition \ref{c11},
the relation
\begin{eqnarray*}
\Pr_{\bm{J},\bm{M}|\bar{q}^{(0)}, \bar{\bm{a}},\bar{\bm{b}},\vec{\bm{N}}} 
\{ \hat{\phi}_3(\bm{M},\vec{\bm{N}}) < \phi(\bm{J}) \}
\le 8 \cdot 2^{-2\beta-8}
%\Label{1-24-6a}.
\end{eqnarray*}
holds.
\end{theorem}

\begin{proof}
The definition of $\hat{\phi}_2$ 
given in (\ref{10-26-1})
yields
\begin{eqnarray}
\fl & \{ \hat{\phi}_3(\bm{M},\vec{\bm{N}}) <\phi(\bm{J}) \}
\subset
\{ \hat{\phi}_2(\bar{\bm{M}},\vec{\bm{N}}) < \phi (\bm{J}) \}
\cup
\{ \hat{\phi}_3(\bm{M},\vec{\bm{N}}) < \hat{\phi}_2(\bar{\bm{M}},\vec{\bm{N}}) \} \nonumber \\
\fl \subset &
\{ \hat{\phi}_1(\bar{q}^{(0)},\bar{a}^{(1)}_{\times},\bar{b}^{(1)}_{\times},N_s,\mu_s)< \phi(\bm{J}) \} 
\cup
\{ \hat{\phi}_3(\bm{M},\vec{\bm{N}}) < \hat{\phi}_2(\bar{\bm{M}},\vec{\bm{N}}) \} 
\Label{10-13-8}.
\end{eqnarray}
Hence,
the above calculations of the partial derivatives
and Conditions \ref{c1} and \ref{c2} imply
\begin{eqnarray*}
\fl & \{ \hat{\phi}_3(\bm{M},\vec{\bm{N}}) < \hat{\phi}_2(\bar{\bm{M}},\vec{\bm{N}}) \} \\
\fl \subset &
\{ \bar{q}_0 < \hat{q}^{(0)}(M_0) \}
\cup \{ \bar{M}_1 < \hat{M}_1(M_1) \}
\cup \{ \hat{M}_2(M_2) < \bar{M}_2 \}
 \cup \{ \hat{M}_3(M_3) < \bar{M}_3 \}.
\Label{11-13-2}
%\cup \{ \bar{M}_4 < \hat{M}_4(M_4) \}.
\end{eqnarray*}
Thus,
it follows from the relations
(\ref{10-13-10}),
(\ref{10-13-11}),
(\ref{10-13-12}),
and (\ref{10-13-13})
%and (\ref{10-13-14})
that
\begin{eqnarray}
\Pr_{\bm{J},\bm{M}|\bar{q}^{(0)}, \bar{\bm{a}},\bar{\bm{b}},\vec{\bm{N}}} 
\{ \hat{\phi}_3(\bm{M},\vec{\bm{N}}) < \hat{\phi}_2(\bar{\bm{M}},\vec{\bm{N}}) \} 
\le 4\cdot 2^{-2\beta-8}.
\Label{5-1-1}
\end{eqnarray}
Using (\ref{10-13-8}) and (\ref{9-24-1}),
we obtain
\begin{eqnarray*}
\Pr_{\bm{J},\bm{M}|\bar{q}^{(0)}, \bar{\bm{a}},\bar{\bm{b}},\vec{\bm{N}}} 
\{ \hat{\phi}_3(\bm{M},\vec{\bm{N}}) < \phi(\bm{J}) \}
\le 
3\cdot 2^{-2\beta-8}+4\cdot 2^{-2\beta-8}
\le 8 \cdot 2^{-2\beta-8}.
\end{eqnarray*}
\end{proof}

%The reader who is not interested in the asymptotic case can proceed to Section \ref{s9}.

\begin{rem}\Label{r5-15}
When the vacuum pulse has a possibility to contain a non-vacuum state,
we adjust the estimate $\hat{q}^{(0)}(M_0)$ as follows
\begin{eqnarray}
\fl
\hat{q}^{(0)}(M_0):=
p_{\est}^+(N_0-X_{\per}^+(N_0, q, 2^{-2\beta-8}), M_0-X_{\per}^+(N_0, q, 2^{-2\beta-8}), 2^{-2\beta-8})
\Label{5-1-2}
\end{eqnarray}
where
the vacuum pulse becomes a non-vacuum state with a probability $q$.
Let $N_0^{(1)}$ be the number of non-vacuum pulses among $N_0$ pulses.
Then,
\begin{eqnarray*}
\fl
\{ \bar{q}_0 < \hat{q}^{(0)}(M_0) \}
\subset 
\{ N_0^{(1)} > X_{\per}^+(N_0, q, 2^{-2\beta-8})\}
\cup 
\{
\bar{q}_0 <
p_{\est}^+(N_0-N_0^{(1)}, M_0-N_0^{(1)}, 2^{-2\beta-8})\}.
\end{eqnarray*}
Hence, the probability of 
$\{ \bar{q}_0 < \hat{q}^{(0)}(M_0) \}$ is less than 
$2\cdot 2^{-2\beta-8}$.
Thus, 
in the proof of Theorem \ref{thm5-1-1},
we replace the right hand side of (\ref{5-1-1}) 
by $5\cdot 2^{-2\beta-8}$.
Then, 
since $3\cdot 2^{-2\beta-8}+5\cdot 2^{-2\beta-8} = 8 \cdot 2^{-2\beta-8}$,
we obtain Theorem \ref{thm5-1-1} in this adjustment.
\end{rem}

\section{Derivation of upper bound $\hat{\phi}_b$ of leaked information}\Label{s9}
\subsection{Characterizations of Conditions \ref{c5} and \ref{c15}}
%for $\mu_1$, $\mu_2$, and $\bm{M}$}
In this section, we define the upper bound $\hat{\phi}_b(\bm{M})$ satisfying (\ref{1-25-1}) as a function of the measured value $\bm{M}$.
Note that the upper bound $\hat{\phi}_b(\bm{M})$ does not depend on the partition $\vec{\bm{N}}$.
%In the following, we treat the case when the signal intensity is $\mu_1$ and the decoy intensity is $\mu_2$.
%However, the discussions in this section except for Subsection \ref{s92} are still valid with replacing $N$, $N^{(0)}$, and $N^{(1)}$ by $N'$, ${N^{(0)}}'$, and ${N^{(1)}}'$ even when the signal intensity is $\mu_2$ and the decoy intensity is $\mu_1$.
In this subsection, for this purpose,
we recall Conditions \ref{c5} and \ref{c15},
which are conditions 
for $\mu_1$, $\mu_2$, $N_0, N_1, N_2$ and the observed data $\bm{M}$.
Condition \ref{c15} plays an alternative role of Condition \ref{c11}.

Substituting $\bar{\bm{M}}$ into $\hat{\bm{M}}$,
and applying the relations
(\ref{9-17-2-a}), (\ref{9-17-3-a}), and 
(\ref{9-17-2-b}),
we can calculate the above values as
\begin{eqnarray*}
\fl A^{(0)}_1 = (\hat{a}^{(1)}_{\times}-\hat{q}^{(0)}/2) N^{(2)}_2 ,\quad
A^{(1)}_1 = (\hat{a}^{(2)}_{\times}-\hat{a}^{(1)}_{\times}) N^{(2)}_2 \\
%\frac{q^{(0)}(2 b^{(1)}-q^{(1)})}{2 q^{(1)}} \\
%\frac{2 b^{(1)}-q^{(1)}}{2 q^{(1)}} \\
\fl A^{(1)}_2 = \hat{a}^{(1)}_{\times} N^{(2)}_1  ,\quad
A^{(2)}_2 = \hat{a}^{(2)}_{\times} N^{(2)}_2 ,\quad
B^{(1)}_1 = \hat{b}^{(1)}_{\times} N^{(1)}_1.
\end{eqnarray*}
%As is explained latter, 
Now, we show that these values take positive values, naturally.
Assume that there is no eavesdrop, i.e., 
That is, we adopt the following model \cite{LB,GRTZ}
\begin{eqnarray}
&p_{i,+}=p_{i,\times}= 1-e^{-\alpha \mu_i} + p_0 ,\nonumber \\
&s_{i,+} =s_{i,\times}= s(1-e^{-\alpha \mu_i}) + \frac{p_0}{2} ,
\Label{10-8-14}
\end{eqnarray}
where $\alpha$ is the total transmission including quantum efficiency of the detector,
and $s$ is the error due to the imperfection of the optical system.
This model implies $p_{i,\times}-s_{i,\times}= (1-s)(1-e^{-\alpha \mu_i}) + \frac{p_0}{2}$.
As was shown in the beginning of Section \ref{s4},
when all of $\vec{\bm{N}}$ are close to their expectations,
$\hat{a}^{(2)}_{\times}-\hat{q}^{(0)}/2$ and 
$\hat{a}^{(2)}_{\times}-\hat{a}^{(1)}_{\times}$ are
positive.
Thus, Condition \ref{c5} 
holds under 
the condition (\ref{10-8-14}).
Similarly, Condition \ref{c15} 
holds under the condition (\ref{10-8-14}) with small $s$.
Since the condition (\ref{10-8-14}) is a natural assumption,
Conditions \ref{c5} and \ref{c15} can be regarded as natural assumptions.
Hence, we need to choose
the initial parameters $\mu_1$, $\mu_2$, $N_0$, $N_1$, $N_2$
so that
Conditions \ref{c5} and \ref{c15} hold with high probability without the presence of an eavesdropper.

Here, we need to pay attention to the difference between
Condition \ref{c6}
and 
Conditions \ref{c5} and \ref{c6}.
Conditions \ref{c5} and \ref{c15} are assumptions for 
the initial parameters $\mu_1$, $\mu_2$, $N_0$, $N_1$, and $N_2$.
On the other hand, Condition \ref{c6} is an assumption for 
the measured values $\bm{M}$ and the initial parameters $\mu_1$, $\mu_2$, $N_0$, $N_1$, and $N_2$.
This is because the estimates $\hat{\bm{M}}$ 
are determined from the measured values $\bm{M}$
via the relations
(\ref{10-4-0}), (\ref{10-4-1}), (\ref{10-4-2}), 
and (\ref{10-4-3}).

\subsection{Derivation of $\vec{\bm{N}}=(\bm{N}_1$, $\bm{N}_2)$}\Label{s93}
Next, in order to derive the estimates of $\bm{N}_1$ and $\bm{N}_2$ 
giving an upper bound of $\hat{\phi}_3(\bm{M})$,
we calculate the partial derivatives of $\hat{\phi}_1$
with respect to $\bm{N}_1$ and $\bm{N}_2$.
For this purpose,
we calculate the partial derivatives of $\hat{a}^{(1)}_{\times}(\bar{\bm{M}},\vec{\bm{N}})$ 
with respect to $N^{(0)}_1,N^{(1)}_1,N^{(0)}_2,N^{(1)}_2,N^{(2)}_2$ as follows.
\begin{eqnarray*}
\fl & \frac{\partial \hat{a}^{(1)}_{\times}(\hat{\bm{M}},\hat{q}^{(0)},\vec{\bm{N}})}{\partial N^{(0)}_1} \\
\fl =&
\frac{ \hat{M}_2 - \hat{q}^{(0)} (N^{(0)}_2 +N^{(2)}_2)/ 2}
{N^{(1)}_1 N^{(2)}_2- N^{(1)}_2 N^{(2)}_1} 
- N_2^{(1)} \frac{ N^{(2)}_2(\hat{M}_1 - \hat{q}^{(0)} N^{(0)}_1 / 2)
-N^{(2)}_1 (\hat{M}_2 - \hat{q}^{(0)} N^{(0)}_2 / 2)
}{(N^{(1)}_1 N^{(2)}_2- N^{(1)}_2 N^{(2)}_1)^2} ,\\
\fl & \frac{\partial \hat{a}^{(1)}_{\times}(\hat{\bm{M}},\hat{q}^{(0)},\vec{\bm{N}})}{\partial N^{(1)}_1} \\
\fl =&
\frac{ \hat{M}_2 - \hat{q}^{(0)} N^{(0)}_2 / 2}{N^{(1)}_1 N^{(2)}_2- N^{(1)}_2 N^{(2)}_1} 
- (N_2^{(2)}+N_2^{(1)}) \frac{ N^{(2)}_2(\hat{M}_1 - \hat{q}^{(0)} N^{(0)}_1 / 2)
-N^{(2)}_1 (\hat{M}_2 - \hat{q}^{(0)} N^{(0)}_2 / 2)
}{(N^{(1)}_1 N^{(2)}_2- N^{(1)}_2 N^{(2)}_1)^2} , \\
\fl & \frac{\partial \hat{a}^{(1)}_{\times}(\hat{\bm{M}},\hat{q}^{(0)},\vec{\bm{N}})}{\partial N^{(0)}_2} 
=
\frac{
N^{(2)}_1 \hat{q}^{(0)} /2 
}{N^{(1)}_1 N^{(2)}_2- N^{(1)}_2 N^{(2)}_1} , \\
\fl & \frac{\partial \hat{a}^{(1)}_{\times}(\hat{\bm{M}},\hat{q}^{(0)},\vec{\bm{N}})}{\partial N^{(1)}_2} 
 =
\frac{
N^{(2)}_1(
N^{(2)}_2
(\hat{M}_1 - \frac{\hat{q}^{(0)}}{2} N^{(0)}_1)
-
N^{(2)}_1 
(\hat{M}_2 - \frac{\hat{q}^{(0)}}{2} N^{(0)}_2) 
)
}{(N^{(1)}_1 N^{(2)}_2- N^{(1)}_2 N^{(2)}_1)^2}, \\
\fl & \frac{\partial \hat{a}^{(1)}_{\times}(\hat{\bm{M}},\hat{q}^{(0)},\vec{\bm{N}})}{\partial N^{(2)}_2} \\
\fl =&
\frac{ \hat{M}_1 - \frac{\hat{q}^{(0)}(M_0)}{2} N^{(0)}_1}{N^{(1)}_1 N^{(2)}_2- N^{(1)}_2 N^{(2)}_1} 
-
\frac{
N^{(1)}_1(
N^{(2)}_2
(\hat{M}_1 - \frac{\hat{q}^{(0)}}{2} N^{(0)}_1)
-
N^{(2)}_1 
(\hat{M}_2 - \frac{\hat{q}^{(0)}}{2} N^{(0)}_2) 
)
}{(N^{(1)}_1 N^{(2)}_2- N^{(1)}_2 N^{(2)}_1)^2} .
\end{eqnarray*}
Here, we should remark that $N_1^{(2)}=N_1-N_1^{(0)}-N_1^{(1)}$.
That is, the variable $N_1^{(2)}$ is a dependent variable.
Due to Condition \ref{c5}, all of the above values are positive.

Next, under Condition \ref{c5},
we calculate the partial derivatives of
$\hat{b}^{(1)}_{\times}(\bar{\bm{M}},\vec{\bm{N}})$ 
with respect to $N^{(0)}_1,N^{(1)}_1,N^{(0)}_2,N^{(1)}_2,N^{(2)}_2$ as follows.
\begin{eqnarray*}
\fl & \frac{\partial \hat{b}^{(1)}_{\times}(\hat{\bm{M}},\hat{q}^{(0)},\vec{\bm{N}})}{\partial N^{(0)}_1} 
 =
- 
\frac{\hat{q}^{(0)}}{2 N^{(1)}_1}<0
\\
\fl & \frac{\partial \hat{b}^{(1)}_{\times}(\hat{\bm{M}},\hat{q}^{(0)},\vec{\bm{N}})}{\partial N^{(1)}_1} 
=
-\frac{ \hat{M}_3 - \hat{q}^{(0)} N^{(0)}_1 /2 }{(N_1^{(1)})^2}<0
\\
\fl & \frac{\partial \hat{b}^{(1)}_{\times}(\hat{\bm{M}},\hat{q}^{(0)},\vec{\bm{N}})}{\partial N^{(0)}_2} 
= \frac{\partial \hat{b}^{(1)}_{\times}(\hat{\bm{M}},\hat{q}^{(0)},\bm{N}_1,\bm{N}_2)}{\partial N^{(1)}_2} 
= \frac{\partial \hat{b}^{(1)}_{\times}(\hat{\bm{M}},\hat{q}^{(0)},\bm{N}_1,\bm{N}_2)}{\partial N^{(2)}_2} =0.
\end{eqnarray*}

Since
$\frac{\partial \hat{\phi}_1}{\partial \bar{a}^{(1)}_{\times}} <0$ 
and
$\frac{\partial \hat{\phi}_1}{\partial \bar{b}^{(1)}_{\times}} 
>0$,
due to (\ref{10-25-2}) and (\ref{10-25-3}),
any element $\vec{\bm{N}}\in \Omega_1$ satisfies
\begin{eqnarray*}
\frac{\partial \hat{\phi}_1}{\partial N^{(0)}_1} <0 ,\quad
\frac{\partial \hat{\phi}_1}{\partial N^{(1)}_1} <0 ,\quad
\frac{\partial \hat{\phi}_1}{\partial N^{(0)}_2} <0 ,\quad
\frac{\partial \hat{\phi}_1}{\partial N^{(1)}_2} <0 ,\quad
\frac{\partial \hat{\phi}_1}{\partial N^{(2)}_2} <0 .
\end{eqnarray*}
Thus, since $N^{(0)}_1,N^{(1)}_1,N^{(0)}_2,N^{(1)}_2,$ and $N^{(2)}_2$ 
obey the binomial distribution,
we decide 
$\hat{\bm{N}}_1=(\hat{N}^{(0)}_1,\hat{N}^{(1)}_1)$
and
$\hat{\bm{N}}_2=(\hat{N}^{(0)}_2,\hat{N}^{(1)}_2,\hat{N}^{(2)}_2)$
in the following way:
\begin{eqnarray}
\hat{N}^{(0)}_1&:=X_{\per}^{-}(N_1,e^{-\mu_1},2^{-2\beta-8})\Label{3-28-1}\\
\hat{N}^{(1)}_1&:=X_{\per}^{-}(N_1,\mu_1 e^{-\mu_1},2^{-2\beta-8})\Label{3-28-2}\\
\hat{N}^{(0)}_2&:=X_{\per}^{-}(N_2,e^{-\mu_2},2^{-2\beta-8})\Label{3-28-3}\\
\hat{N}^{(1)}_2&:=X_{\per}^{-}(N_2,\mu_2 e^{-\mu_2},2^{-2\beta-8})\Label{3-28-4}\\
\hat{N}^{(2)}_2&:=X_{\per}^{-}(N_2,\omega_2 \mu_2^2 e^{-\mu_2},2^{-2\beta-8}) ,\Label{3-28-5}
\end{eqnarray}
which give Step (2) in Figs. \ref{outline-zu2} and \ref{outline-zu}.
Then,
%when $\hat{\bm{M}}(\bm{M})$ satisfies Conditions \ref{c11} and \ref{c5},
we define
\begin{eqnarray}
\fl
\hat{\phi}_4(\bm{M}):= 
\left\{
\begin{array}{ll}
\hat{\phi}_2(\hat{\bm{M}}(\bm{M}),\vec{\hat{\bm{N}}}) & \hbox{if Conditions \ref{c6}, \ref{c5} and \ref{c15} hold} \\
M & \hbox{otherwise.}
\end{array}
\right.
\Label{10-4-a}
\end{eqnarray}

Due to the definition,
any element $\vec{\bm{N}}\in \Omega_1$ satisfies
\begin{eqnarray}
\hat{\phi}_4(\bm{M}) \ge \hat{\phi}_3(\bm{M},\vec{\bm{N}}),\Label{11-13-1}
\end{eqnarray}
i.e., 
the function $\hat{\phi}_4$ satisfies the condition (\ref{1-25-1}) for $\hat{\phi}_b$
when $\hat{\phi}_a$ is $\hat{\phi}_3$.
In summary,
when the parameters $\mu_1$, $\mu_2$, $N_s$, $N_0$, $N_1$, and $N_2$
satisfy Condition \ref{c6},
and when
we choose the sacrifice bit-length
\begin{eqnarray}
\fl S(\bm{M})
&= \hat{\phi}_4(\bm{M}) +2\beta+ 5 \nonumber \\
\fl &=
\left\{
\begin{array}{ll}
\hat{\phi}_2(\hat{\bm{M}}(\bm{M}),\vec{\hat{\bm{N}}}) +2\beta+ 5& \hbox{if Conditions \ref{c6}, \ref{c5} and \ref{c15} hold} \\
M+2\beta+ 5 & \hbox{otherwise,}
\end{array}
\right.
\Label{3-28-6}
\end{eqnarray}
%with the above choice of $\hat{\phi}_4(\bm{M})$,
due to (\ref{35}),
we obtain 
$\| \rho_{A,E}-\rho_{\ideal} \|_1 \le 2^{-\beta}$.
Note that $\hat{\phi}_2(\hat{\bm{M}}(\bm{M}),\vec{\hat{\bm{N}}})$ is given in (\ref{10-26-1}), which is the same as $\hat{\phi}_2 $ given in (\ref{10-26-1i}).
Since $M+2\beta+ 5$ is greater than $\dim C_1$,
the formula (\ref{3-28-6}) implies the abort of the protocol when one of
Conditions \ref{c6}, \ref{c5} and \ref{c15} does not hold.
Hence, we obtain (\ref{35i}) under the formula given in Subsection \ref{s-non-imp}.

\section{Security proof of improved formula}\Label{s9-5}
Up to the previous section,
based on (\ref{35}) given in Section \ref{s4},
we evaluate the universal composability criterion with the finite-length setting.
However, the above given evaluation can be improved
by removing the square root for a part of probabilities.
The improved formula for the sacrifice bit-length given in Subsection \ref{s-imp}
is derived by removing the square root for a part of probabilities.
The purpose of this section is to show the following theorem.

\begin{theorem}\Label{th5-1-2}
The improved formula for the sacrifice bit-length given in Subsection \ref{s-imp}
satisfies 
\begin{eqnarray}
\| \rho_{A,E}-\rho_{\ideal } \|_1
\le 2^{-\beta}.
\end{eqnarray}
\end{theorem}

Now, we will show the above theorem.
For this purpose, we discuss the formula (\ref{9-17-6}) more deeply.
Let $\bm{s}=(s_1, \ldots, s_{N_s+N_0+N_1+N_2})$ be the indicators of the kinds of initial states of  
pulses received by Bob.
The indicators are decided as follows.
If the $i$-th received state is the vacuum state,
$s_i$ is $0$.
If the $i$-th received state is the single-photon state,
$s_i$ is $1$.
If the $i$-th received state is the state $\rho_2$,
$s_i$ is $2$.
Otherwise, 
$s_i$ is $3$.
The information $\bm{s}$ contains all of information for 
$\bm{N}$ and $(J^{(0)},J^{(1)},J^{(2)})$.
That is, $\bm{s}$ decides $\bm{N}$ and $(J^{(0)},J^{(1)},J^{(2)})$.
However, it cannot decide $J^{(1)}_e$.

Once we apply (\ref{9-17-6}) to the case when $\bm{s}$ is fixed,
we obtain
\begin{eqnarray}
\| \rho_{A,E|\bm{s}}-\rho_{\ideal |\bm{s}} \|_1
\le 2\sqrt{2}\sqrt{P_{ph|\bm{s}}}
\Label{9-17-6-c},
\end{eqnarray}
where
$\rho_{A,E|\bm{s}}$, $\rho_{\ideal|\bm{s}}$,
and
$P_{ph|\bm{s}}$ are 
the final true composite state, the ideal final state,
and the averaged virtual decoding phase error probability
conditioned with $\bm{s}$.
Hence, 
the final true composite state $\rho_{A,E}$ and
the ideal final state $\rho_{\ideal} $
are written as
\begin{eqnarray*}
\rho_{A,E}= \sum_{\bm{s}}
P(\bm{s})|\bm{s}\rangle \langle \bm{s}|
\otimes \rho_{A,E|\bm{s}},\\
\rho_{\ideal} = \sum_{\bm{s}}
P(\bm{s})|\bm{s}\rangle \langle \bm{s}|
\otimes \rho_{\ideal|\bm{s}} .
\end{eqnarray*}
Hence,
for a set $\Omega$ of $\bm{s}$, we have
\begin{eqnarray}
\fl \| \rho_{A,E}-\rho_{\ideal } \|_1
&=
\sum_{\bm{s}}
P(\bm{s})
\| \rho_{A,E|\bm{s}}-\rho_{\ideal |\bm{s}} \|_1 
\le 
\sum_{\bm{s}}
P(\bm{s})
\min \{2\sqrt{2}\sqrt{P_{ph|\bm{s}}},2\} \nonumber \\
\fl &\le 
2 P(\Omega^c)
+
2\sqrt{2}
\sqrt{
\sum_{\bm{s}\in \Omega}
P(\bm{s})P_{ph|\bm{s}}}
\Label{9-17-6-d}.
\end{eqnarray}
We choose the set $\Omega$ as
\begin{eqnarray*}
\fl \Omega:=& \Omega_1 \cap
\{ J^{(0)} \ge X_{\per}^-(N^{(0)},\bar{q}^{(0)},2^{-\beta-6}) \} 
\cap
\{ J^{(1)} \ge X_{\per}^-(N^{(1)},\bar{a}^{(1)}_{\times}+\bar{b}^{(1)}_{\times},2^{-\beta-6}) \}  \\
\fl & \cap
%\{\bar{M}_0 \ge \hat{M}_0 (M_0)\}
\{ \bar{q}_0 < \hat{q}^{(0)}(\hat{M}_0(M_0)) \}.
\end{eqnarray*}
Hence, using (\ref{9-17-6-d}), we obtain
\begin{eqnarray}
\fl & \| \rho_{A,E}-\rho_{\ideal } \|_1 
\le 
\max_{\rho,\sigma}\|\rho-\sigma \|_1
 \cdot (5 \cdot 2+ 3) \cdot  2^{-\beta-6}
+
2\sqrt{2}
\sqrt{
\sum_{\bm{s}\in \Omega}
P(\bm{s})P_{ph|\bm{s}}} \nonumber \\
\fl =&
2 \cdot 13 \cdot  2^{-\beta-2}
+
2\sqrt{2}
\sqrt{
\sum_{\bm{s}\in \Omega}
P(\bm{s})P_{ph|\bm{s}}} 
\Label{9-17-6-e}.
\end{eqnarray}

%Using (\ref{10-13-4}), we obtain
%\begin{eqnarray}
%\fl &
%\Omega\cap \{ \frac{J^{(1)}_e}{J^{(1)}} \ge p^+(\hat{J}^{(1)}(\bar{a}^{(1)}_{\times}+\bar{b}^{(1)}_{\times},N^{(1)}),\bar{a}^{(1)}_{\times}+\bar{b}^{(1)}_{\times},2^{-\beta-8}) \} \nonumber \\
%\fl \subset &
%\Omega\cap
%\{ \frac{J^{(1)}_e}{J^{(1)}} \ge p^+(J^{(1)},\bar{a}^{(1)}_{\times}+\bar{b}^{(1)}_{\times},2^{-\beta-8}) \}
%\end{eqnarray}
Since $\hat{\phi}_4(\bm{M}) \ge 
\hat{\phi}_3(\hat{\bm{M}}(\bm{M}),\vec{\bm{N}})$ 
for $\vec{\bm{N}} \in \Omega_1$, as was shown in (\ref{11-13-1}), 
the relations (\ref{10-13-8}), (\ref{11-13-2}), and (\ref{9-24-1-b})
guarantee that
\begin{eqnarray*}
\fl \Omega \cap 
\{ \hat{\phi}_4(\bm{M}) < \phi(\bm{J}) \}
\subset 
\Omega \cap 
\{ \hat{\phi}_3(\bm{M},\vec{\bm{N}}) < \phi(\bm{J}) \} \nonumber \\
\fl \subset 
\Omega \cap 
\Bigl(
\{ \hat{\phi}_1(\bar{q}^{(0)},\bar{a}^{(1)}_{\times},\bar{b}^{(1)}_{\times},N_s,\mu_s) < \phi (\bm{J})\} 
\cup
\{ \hat{\phi}_3(\bm{M},\vec{\bm{N}}) < \hat{\phi}_2(\bar{\bm{M}},\vec{\bm{N}}) \} \Bigr) \nonumber \\
\fl \subset 
\Omega \cap 
\Bigl(
\{ J^{(0)} \le X_{\per}^-(N_s,e^{-\mu_s} \bar{q}^{(0)},2^{-\beta-6}) \} 
\cup
\{ J^{(1)} \le X_{\per}^-(N_s,e^{-\mu_s}\mu_s( \bar{a}^{(1)}_{\times}+\bar{b}^{(1)}_{\times}),2^{-\beta-6}) \} \nonumber \\
\fl \quad \cup 
\Bigl\{ p_{\per}^+\Bigl(J^{(1)},\frac{\bar{b}^{(1)}_{\times}}{\bar{a}^{(1)}_{\times}+\bar{b}^{(1)}_{\times}}
,2^{-2\beta-7}\Bigr) \le \frac{J^{(1)}_e}{J^{(1)}} \Bigr\} \nonumber \\
\fl \quad \cup 
\{ \bar{q}_0 < \hat{q}^{(0)}(\hat{M}_0(M_0)) \} %\{ \bar{M}_0 < \hat{M}_0(M_0) \}
\cup \{ \bar{M}_1 < \hat{M}_1(M_1) \}
\cup \{ \bar{M}_2 > \hat{M}_2(M_2) \}
\cup \{ \bar{M}_3 > \hat{M}_3(M_3) \}
\Bigr ) \nonumber \\
\fl \subset 
\Omega \cap 
\Bigl(
\{ \frac{J^{(1)}_e}{J^{(1)}} \ge p_{\per}^+(J^{(1)},\frac{\bar{b}^{(1)}_{\times}}{\bar{a}^{(1)}_{\times}+\bar{b}^{(1)}_{\times}}
,2^{-2\beta-7}) \} \nonumber \\
\fl \quad 
\cup \{ \bar{M}_1 < \hat{M}_1(M_1) \}
\cup \{ \bar{M}_2 > \hat{M}_2(M_2) \}
\cup \{ \bar{M}_3 > \hat{M}_3(M_3) \}
\Bigr ) \nonumber\\
\fl \subset 
\Bigl\{ \frac{J^{(1)}_e}{J^{(1)}} \ge 
p_{\per}^+\Bigl(J^{(1)},\frac{\bar{b}^{(1)}_{\times}}{\bar{a}^{(1)}_{\times}+\bar{b}^{(1)}_{\times}}
,2^{-2\beta-7}\Bigr) \Bigr\} 
\cup \{ \bar{M}_1 < \hat{M}_1(M_1) \} \nonumber \\
\fl \quad
\cup \{ \bar{M}_2 > \hat{M}_2(M_2) \}
\cup \{ \bar{M}_3 > \hat{M}_3(M_3) \}.
\end{eqnarray*}
%where $\hat{J}^{(1)}=X_{\per}^-(N_s,e^{-\mu_s}\mu_s( \bar{a}^{(1)}_{\times}+\bar{b}^{(1)}_{\times}),2^{-\beta-6})$.
Hence, due to (\ref{3-28-6}),
\begin{eqnarray*}
\sum_{\bm{s}\in \Omega}
P(\bm{s})P_{ph|\bm{s}} 
\le
\Pr
\Omega \cap 
\{ \hat{\phi}_4(\bm{M}) < \phi(\bm{J}) \}
+2^{-2\beta-5} \\
\le
4\cdot 2^{-2\beta-7}+2^{-2\beta-5} 
=2^{-2\beta-4} .
\end{eqnarray*}
Thus, using (\ref{9-17-6-e}), we obtain
\begin{eqnarray}
\| \rho_{A,E}-\rho_{\ideal } \|_1 
\le 
2 \cdot 13 \cdot  2^{-\beta-6}
+
2\sqrt{2}\cdot 2^{\frac{-2\beta-4}{2}}
\le 2^{-\beta}\Label{5-1-3}
\end{eqnarray}
because $2 \sqrt{2}+\frac{13}{16}(\cong 3.64)\le 4$.
Indeed, in order to put out a probability from the square root,
the event corresponding to the probability must be defined by $\bm{s}$,
i.e.,
the probability conditioned with $\bm{s}$ must take the value $1$ or $0$.
Hence,
the probabilities corresponding to the sets
$\{ \frac{J^{(1)}_e}{J^{(1)}} \ge p_{\per}^+(J^{(1)},
\frac{\bar{b}^{(1)}_{\times}}{\bar{a}^{(1)}_{\times}+\bar{b}^{(1)}_{\times}}
,2^{-2\beta-7}) \}, 
\{ \bar{M}_1 < \hat{M}_1(M_1) \},
\{ \bar{M}_2 > \hat{M}_2(M_2) \}$, and $\{ \bar{M}_3 > \hat{M}_3(M_3) \}$
cannot be put out from the square root.

%Fig. \ref{finite1} gives the numerical calculation of the key generation rate with the finite-block length based on the formula given in this section.

In summary,
when the parameters $\mu_1$, $\mu_2$, $N_0$, $N_1$, and $N_2$
satisfy Condition \ref{c6} modified in Subsection \ref{s-imp},
and when we choose the sacrifice bit-length
$S(\bm{M})= \hat{\phi}_4(\bm{M})
+2\beta+ 5$ 
by using the choice of $\hat{\phi}_4(\bm{M})$ given in (\ref{10-4-a})
with the modification given in Subsection \ref{s-imp},
we obtain 
$\| \rho_{A,E}-\rho_{\ideal} \|_1 \le 2^{-\beta}$.

\begin{rem}
When the vacuum pulse has a possibility to contain a non-vacuum state,
we adjust the estimate $\hat{q}^{(0)}(\hat{M}_0(M_0))$
as (\ref{5-1-2ii}).
Then,
\begin{eqnarray*}
\fl
\{ \bar{q}_0 < \hat{q}^{(0)}(M_0) \}
\subset 
\{ N_0^{(1)} > X_{\per}^+(N_0, q, 2^{-\beta-6})\}
\cup 
\{
\bar{q}_0 <
p_{\est}^+(N_0-N_0^{(1)}, M_0-N_0^{(1)}, 2^{-\beta-6})\}.
\end{eqnarray*}
Hence, the probability of 
$\{ \bar{q}_0 < \hat{q}^{(0)}(M_0) \}$ is less than 
$2\cdot 2^{-\beta-6}$.
Hence, 
in the above proof,
we replace the right hand side of (\ref{9-17-6-e}) by 
$2 \cdot 14 \cdot  2^{-\beta-2}
+
2\sqrt{2}\sqrt{\sum_{\bm{s}\in \Omega}P(\bm{s})P_{ph|\bm{s}}}$.
Then, we replace (\ref{5-1-3}) by
\begin{eqnarray}
\| \rho_{A,E}-\rho_{\ideal } \|_1 
\le 
2 \cdot 14 \cdot  2^{-\beta-6}
+
2\sqrt{2}\cdot 2^{\frac{-2\beta-4}{2}}
\le 2^{-\beta}
\end{eqnarray}
because $2 \sqrt{2}+\frac{14}{16}(\cong 3.70)\le 4$.
Thus, we obtain Theorem \ref{th5-1-2} in the adjustment (\ref{5-1-2ii}).
\end{rem}

\section{Conclusion and further improvement}
In this paper, under the BB84 protocol with the decoy method, 
based on several observed values,
we have derived the required sacrifice bit-length 
$S(\bm{M})= \hat{\phi}_2(\bm{M}) +2\beta+ 5$,
where $\hat{\phi}_2(\bm{M})$ is given in {\bf Step (6)}.
Under the above sacrifice bit-length, we have shown that
the final keys satisfy the security condition 
$\| \rho_{A,E}-\rho_{\ideal} \|_1 \le 2^{-\beta}$
when the parameters $\mu_1$, $\mu_2$, $N_s$, $N_0$, $N_1$, and $N_2$
satisfy Condition \ref{c6}.
Hence, in order to apply our formula, 
we need to choose 
the parameters
$\mu_1$, $\mu_2$, $N_s$, $N_0$, $N_1$, and $N_2$
so that Condition \ref{c6} holds.
This is a definitive requirement for our analysis.
However, when we choose 
sufficiently large integers $N_s$, $N_0$, $N_1$, and $N_2$ for 
the two values $\mu_1$ and $\mu_2-\mu_1$,
Condition \ref{c6} holds.
Indeed, when the two positive values $\mu_1$ and $\mu_2-\mu_1$ are quite small,
we need to choose quite large integers $N_s$, $N_0$, $N_1$, and $N_2$.
As the second requirement,
we need to choose 
the parameters
$\mu_1$, $\mu_2$, $N_s$, $N_0$, $N_1$, and $N_2$
so that Conditions \ref{c5} and \ref{c15} hold with a high probability
when there is no eavesdropper.
This requirement is also satisfied when
the integers $N_s$, $N_0$, $N_1$, and $N_2$ are
sufficiently large and 
the noise in the channel is sufficiently small.
Indeed, it is not so difficult to realize sufficiently large $N_s$, $N_0$, $N_1$, and $N_2$ for these requirements
because a universal$_2$ hash function (or an $\varepsilon$-almost dual universal$_2$ hash function)
with a large size can be implemented with a small cost \cite{AT11}.

Since the decoy method has so many parameters,
it is quite difficult to derive tight evaluation.
The proposed method might be improved by modifying several points.
However, such a modification might make the protocol more complex.
For example, 
while we treat the decoding phase error probability
and the estimation error probability, separately,
The paper \cite{finite} treated them jointly.
In order to keep the simplicity, it is better to treat these terms separately.
Further, in Section \ref{s6}, we proposed to treat the probability based on the hypergeometric distribution by using the binomial distribution.
If we treat the probabilities given in Section \ref{s7} 
with the hypergeometric distribution, we obtain a better evaluation,
but our analysis becomes much harder.

%In order to treat this problem, 
Therefore, we have to consider the trade-off
between the complexity and the tightness of our evaluation.
This kind of trade-off cannot be ignored from an industrial view point. 
If the protocol is more complex, the cost for maintenance becomes higher.
In particular, when we change the arrangement of the total system
or we change the parameter of the system,
we have to rewrite the program for calculating the sacrifice bit-length.
If the protocol is simple, the change can be easily done.
Otherwise, it spends some additional cost.
Hence, we have to take into account this trade-off.
This paper has treated this trade-off heuristically.
%Due to the property, our treatment is heuristic.

However, its systematic treatment might be possible partially in the following sense.
Assume that we employ the Renner's formalism instead of the phase error correction formalism.
If we parametrize the channel with more parameters to be estimated,
the asymptotic key generation rate becomes better.
One might consider that, 
if the number of parameters describing the model increases,
we obtain a better estimation of the model.
However, it is considered that it is not true in statistics.
This is because if we do not have enough data to characterize so many parameters,
we obtain a larger error.
In order to resolve this problem, we have to treat the trade-off
between the error and the number of parameters.
Such a problem is called the model selection.
In order to treat this problem quantitatively,
we can use several information criteria, e.g.,
Akaike information criterion (AIC)\cite{AIC},
Takeuchi information criterion (TIC)\cite{TIC}, and 
minimum description length principle (MDL)\cite{MDL}.
If we employ the Renner's formalism,
and increase the number of channel parameters for precise description of channel,
we need to consider this kind of trade-off.
Currently, it is not known that what kind of information criterion is suitable for the above our trade-off.

\section*{Acknowledgment}
MH thanks 
Prof. Masahide Sasaki, Prof. Akihisa Tomita, Dr. Toyohiro Tsurumaru,
Prof. Ryutaroh Matsumoto, Dr. Kiyoshi Tamaki, and Dr. Wataru Kumagai 
for valuable comments.
He is partially supported by a MEXT Grant-in-Aid for Scientific Research (A) No. 23246071.
He is also partially supported by the National Institute of Information and Communication Technology (NICT), Japan.
The Center for Quantum Technologies is funded by the Singapore
Ministry of Education and the National Research Foundation
as part of the Research Centres of Excellence programme.

\appendix

\section{Chernoff inequality}\Label{as1}
In this section, we derive a lower bound of the lower percent point
$X_{\per}^-(N,p,\alpha)$ with probability $\alpha$
by using Chernoff inequality.
When the random variable $X$ obeys the 
binomial distribution $Bin(N,p)$,
Chernoff inequality
\begin{eqnarray}
P_p\{X \le Nq\} \le \exp(-ND(q\|p))
\end{eqnarray}
holds with $q<p$, 
where the relative entropy $D(q\|p)$
is defined as $q \log \frac{q}{p}+(1-q) \log \frac{1-q}{1-p}$,
where
$P_p$ is the distribution when the success probability with one trial is 
$p$.

Hence,  
letting $q^-$ be the solution of 
the equation $D(q\|p)=-\frac{\log \alpha}{N}$ with respect to $q$ with $q<p$,
we obtain
\begin{eqnarray}
P_p\{X \le Nq^{-}\} \le \exp(-ND(q^{-}\|p))=\alpha.
\end{eqnarray}
That is, we obtain
$X_{\per}^-(N,p,\alpha) \ge N q^{-}$.
Similarly,
letting $q^+$ be the solution of 
the equation $D(q\|p)=-\frac{\log \alpha}{N}$ with respect to $q$ with $q>p$,
we obtain
$X_{\per}^+(N,p,\alpha) \le N q^{+}$.

Further, combining 
Pinsker inequality
$D(q\|p)\ge 2 (\log e) (p-q)^2 $,
we obtain 
\begin{eqnarray}
P_p\{X \le Nq\} \le \exp(-2 (\log e) N(p-q)^2 ).
\end{eqnarray}
Hence,  
solving 
the equation $2 (\log e) (p-q)^2=-\frac{\log \alpha}{N}$ with respect to $q$,
we obtain two solutions
$\tilde{q}^-:=p- \sqrt{\frac{-\log \alpha}{2 (\log e)N}}$
and
$\tilde{q}^+:=p+ \sqrt{\frac{-\log \alpha}{2 (\log e)N}}$.
Then, we obtain
$X_{\per}^-(N,p,\alpha) \ge N \tilde{q}^{-}$
and $X_{\per}^+(N,p,\alpha) \le N \tilde{q}^{+}$.

Using the information geometry,
we have a better evaluation than
Pinsker inequality as follows.
The relative entropy can be written with an integral form as follows\cite{A-N}.
\begin{eqnarray}
\frac{D(q\|p)}{\log e}=\int_q^p \frac{t-p}{t(1-t)}dt.\Label{q-1-1}
\end{eqnarray}
We consider only the case $p < 1/2$.
When $q < p < 1/2 $,
we have
\begin{eqnarray}
\frac{D(q\|p)}{\log e}\ge \frac{(p-q)^2}{2p(1-p)} .
\end{eqnarray}
Hence,  
solving 
the equation $\frac{(p-q)^2}{2p(1-p)}=-\frac{\log \alpha}{N (\log e)}$ with respect to $q$,
we obtain the smaller solution
$\bar{q}^-:=p- \sqrt{\frac{-2 (\log \alpha)p(1-p)}{(\log e)N}}$.
Then, we obtain
$X_{\per}^-(N,p,\alpha) \ge N \bar{q}^{-}$.

The treatment for
$X_{\per}^+(N,p,\alpha)$ is a little complex.
When $p < q \le 1/2 $,
we have
\begin{eqnarray}
\frac{D(q\|p)}{\log e}\ge \frac{(p-q)^2}{2q(1-q)} .
\end{eqnarray}
Hence,  
solving 
the equation $\frac{(p-q)^2}{2q(1-q)}=-\frac{\log \alpha}{N (\log e)}$ with respect to $q$,
we obtain the larger solution
$\bar{q}^+:=
\frac{
p-\log \alpha/(N\log e)
+ \sqrt{
(-p^{2}+p -\log \alpha/(2N\log e))\cdot 
(-2 \log \alpha)/(N\log e)
}
}{1-2 \log \alpha/(N\log e) }$.
Then, 
when $\bar{q}^+\le 1/2$,
we obtain
$X_{\per}^+(N,p,\alpha) \le N \bar{q}^{+}$.
Indeed, 
since $\bar{q}^+$ is complicated, 
we introduce a simpler upper bound.
Since $\sqrt{a+b}\le \sqrt{a}+\sqrt{b}$,
\begin{eqnarray*}
\fl \bar{q}^+
\le
\hat{q}^+ &:=
\frac{
p-\log \alpha/(N\log e)
+ \sqrt{
(-p^{2}+p)(-2 \log \alpha)/(N\log e)
}
+
\sqrt{ (\log \alpha/(N\log e))^2}
}{1-2 \log \alpha/(N\log e) }\\
\fl &= \frac{
p-2 \log \alpha/(N\log e)
+ \sqrt{
p(1-p)(-2 \log \alpha)/(N\log e)
}
}{1-2 \log \alpha/(N\log e) }.
\end{eqnarray*}
Then,
when $\hat{q}^+\le 1/2$,
we obtain
$X_{\per}^+(N,p,\alpha) \le N \hat{q}^{+}$.

\section{One-sided interval estimation}\Label{as2}
\subsection{One-sided interval estimation based of F distribution}
We consider 
lower one-sided interval estimation with the confidential level $1-\alpha$
when we observe the value $k$ subject to the binomial distribution $Bin(N,p)$ with $N$ trials and probability $p$.

For this purpose, 
when we fix an integer $k$ and define the constants
\begin{eqnarray}
n_1:=2(N-k+1),~ n_2:=2 k,~
f_1:=\frac{n_2}{n_1}\frac{(1-p)}{p},
\end{eqnarray} 
it is known that the random variable $F(n_1,n_2)$ subject to 
F distribution with the freedom $(n_1,n_2)$ satisfies
\begin{eqnarray}
P \{F(n_1,n_2)>f_1\}=
P_p \{ X \ge k\}
=\sum_{i=k}^{N}{N \choose i} p^{i}(1-p)^{N-i}.
\end{eqnarray}
Our task is solving 
$P_p\{ X \ge k\}=1-\alpha$
with respect to $p$ with $p< \frac{k}{N}$ for a given $k$.
Define $f_1^*$ to be the solution of
$P\{F(n_1,n_2)>f_1\}=1-\alpha$ with respect to $f_1$.
Then, 
the solution $p= \frac{n_2}{n_1 f_1^* +n_2} $
satisfies the equation 
$\frac{n_2}{n_1}\frac{(1-p)}{p}=f_1^*$.
Thus, we obtain
\begin{eqnarray}
P_{\frac{n_2}{n_1 f_1^* +n_2}} \{ X \ge k\}
=1-\alpha.
\end{eqnarray}
That is, 
$\frac{n_2}{n_1 f_1^* +n_2}$ is 
the 
lower confidence limit
$p_{\est}^{-}(N,k,\alpha)$ 
of the lower one-sided interval estimation with the confidential level $1-\alpha$
when we observe the value $k$.

Similarly, when
we fix an integer $k$ and define the constants
\begin{eqnarray}
m_1=2(k+1),~m_2=2(N-k),~
f_2=\frac{m_1}{m_2}\frac{p}{(1-p)},
\end{eqnarray} 
it is known that the random variable $F(m_1,m_2)$ subject to 
F distribution with the freedom $(m_1,m_2)$ satisfies
\begin{eqnarray}
P\{ F(m_1,m_2)>f_2\}
=
P_p \{ X \ge k\}.
%\sum_{i=k}^{N}{N \choose i} p^{i}(1-p)^{N-i}.
\end{eqnarray}
Our task is solving 
$P_p \{ X \ge k\}
%\sum_{i=k}^{N}{N \choose i} p^{i}(1-p)^{N-i}
=\alpha$
with respect to $p$ with $p< \frac{k}{N}$ for a given $k$.
Define $f_2^*$ to be the solution of
$P(F(m_1,m_2)>f_2)=\alpha$ with respect to $f_2$.
Then, 
the solution $p= \frac{m_1 f_2}{m_1 f_2+m_2} $
satisfies the equation 
$\frac{m_1}{m_2}\frac{p}{(1-p)}=f_2^*$.
Thus, we obtain
\begin{eqnarray}
P_{\frac{m_2}{m_1 f_2^* +m_2}} \{ X \ge k\}
%\sum_{i=k}^{N}{N \choose i} (\frac{m_2}{m_1 f_2^* +m_2})^{i}
(1-\frac{m_2}{m_1 f_2^* +m_2})^{N-i}
=\alpha.
\end{eqnarray}
That is, 
$\frac{m_2}{m_1 f_2^* +m_2}$ is the 
upper confidence limit
$p_{\est}^{+}(N,k,\alpha)$ 
of the upper one-sided interval estimation with the confidential level $1-\alpha$
when we observe the value $k$.

\subsection{Application of Chernoff inequality}
Assume that we observe the random variable $X$ subject to the binomial distribution $Bin(N,p)$ with $N$ trials and probability $p$.
For a fixed integer $k$,
we have
\begin{eqnarray}
P_p\{\frac{X}{N} \le  \frac{k}{N} \} \le \exp(-ND(\frac{k}{N}\|p))
\end{eqnarray}
with $\frac{k}{N}<p$.
Hence,  
letting $p^-$ be the solution of 
the equation $D(\frac{k}{N}\|p)=-\frac{\log \alpha}{N}$ with respect to $p$ with $\frac{k}{N}<p$,
we obtain
\begin{eqnarray}
P_{p^-}\{\frac{X}{N} \le \frac{k}{N} \} 
\le \exp(-ND(\frac{k}{N} \|p^{-}))=\alpha.
\end{eqnarray}
Thus, $p^-\le p_{\est}^{-}(N,k,\alpha)$. 
Similarly,
letting $q^+$ be the solution of 
the equation $D(\frac{k}{N}\|p)=-\frac{\log \alpha}{N}$ with respect to $p$ with 
$\frac{k}{N}>p$,
we obtain
$p^+ \ge p_{\est}^{+}(N,k,\alpha)$. 

Further, combining Pinsker inequality
$D(q\|p)\ge 2 (\log e) (p-q)^2 $,
we obtain 
\begin{eqnarray}
P_p\{\frac{X}{N} \le  \frac{k}{N} \} \le \exp(-2 (\log e) N(p-\frac{k}{N})^2 ).
\end{eqnarray}
Hence,  
solving 
the equation $2 (\log e) (p-\frac{k}{N})^2=-\frac{\log \alpha}{N}$ with respect to $p$,
we obtain two solutions
$\tilde{p}^-:=\frac{k}{N}- \sqrt{\frac{-\log \alpha}{2 (\log e)N}}$
and
$\tilde{p}^+:=\frac{k}{N}+ \sqrt{\frac{-\log \alpha}{2 (\log e)N}}$.
Then, we obtain
$\tilde{p}^- \le p_{\est}^{-}(N,k,\alpha)$
and
$\tilde{p}^+ \ge p_{\est}^{+}(N,k,\alpha)$. 

Using the relation (\ref{q-1-1}),
we consider better bounds only for the case $\frac{k}{N} < 1/2$.
Solving 
the equation $\frac{(p-\frac{k}{N})^2}{2\frac{k}{N}(1-\frac{k}{N})}=-\frac{\log \alpha}{N (\log e)}$ with respect to $p$
with $p < \frac{k}{N} < 1/2 $,
we obtain the smaller solution
$\bar{p}^-:=\frac{k}{N}- \sqrt{\frac{-2 (\log \alpha)\frac{k}{N}(1-\frac{k}{N})}{(\log e)N}}$.
Then, we obtain
$p_{\est}^{-}(N,k,\alpha) \ge \bar{q}^{-}$.
The treatment for
$p_{\est}^{+}(N,k,\alpha)$ is a little complex.
When $\frac{k}{N} < p \le 1/2 $,
we have
\begin{eqnarray}
\frac{D(\frac{k}{N}\|p)}{\log e}\ge \frac{(p-q)^2}{2q(1-q)} .
\end{eqnarray}
Hence,  
solving 
the equation $\frac{(p-\frac{k}{N})^2}{2p(1-p)}=-\frac{\log \alpha}{N (\log e)}$ with respect to $p$,
we obtain the larger solution
$\bar{p}^+:=
\frac{
k/N-\log \alpha/(N\log e)
+ \sqrt{
(-(k/N)^2+k/N -\log \alpha/(2N\log e))
(-2 \log \alpha)/(N\log e)
}
}{1-2 \log \alpha/(N\log e) }$.
Then, 
when $\bar{p}^+\le 1/2$,
we obtain
$p_{\est}^{+}(N,k,\alpha) \le \bar{p}^{+}$.
Indeed, 
since $\bar{p}^+$ is complicated, 
we introduce a simpler upper bound:
\begin{eqnarray*}
\fl \bar{p}^+
\le
\hat{p}^+ &:=
\frac{
k/N-2 \log \alpha/(N\log e)
+ \sqrt{
k/N(1-k/N)(-2 \log \alpha)/(N\log e)
}
}{1-2 \log \alpha/(N\log e) }.
\end{eqnarray*}
Then,
when $\hat{p}^+\le 1/2$,
we obtain
$p_{\est}^{+}(N,k,\alpha) \le \hat{p}^{+}$.

\section{Calculation with the Gaussian case}\Label{as3}
In order to calculate the sacrifice bit-length given in Section \ref{s12}, 
we need $\rE[e^{\mu_i}]$, $\rE[\mu_i e^{\mu_i}]$, $\rE[\mu_i^{2} e^{\mu_i}]$, and $\omega_2$.
For this purpose, we calculate $e^{-\frac{1}{2\sigma^2}(x-(\mu-\sigma^2))^2}$ as follows.
\begin{eqnarray}
\fl \frac{d e^{-\frac{1}{2\sigma^2}(x-(\mu-\sigma^2))^2}}{d x}&=&-\frac{1}{\sigma^2}(x-(\mu-\sigma^2))e^{-\frac{1}{2\sigma^2}(x-(\mu-\sigma^2))^2}  \\
\fl \frac{d^{2} e^{-\frac{1}{2\sigma^2}(x-(\mu-\sigma^2))^2}}{d x^{2}}&=&\frac{1}{\sigma^4}(x^{2} e^{-\frac{1}{2\sigma^2}(x-(\mu-\sigma^2))^2}-2(\mu-\sigma^2)xe^{-\frac{1}{2\sigma^2}(x-(\mu-\sigma^2))^2} \nonumber \\
&&+((\mu-\sigma^2)^2-\sigma^2)e^{-\frac{1}{2\sigma^2}(x-(\mu-\sigma^2))^2}).
\end{eqnarray}
Hence, $xe^{-\frac{1}{2\sigma^2}(x-(\mu-\sigma^2))^2},x^{2}e^{-\frac{1}{2\sigma^2}(x-(\mu-\sigma^2))^2}$
can be written by using $e^{-\frac{1}{2\sigma^2}(x-(\mu-\sigma^2))^2}$ and its first and second derivatives as follows.
\begin{eqnarray}
\fl xe^{-\frac{1}{2\sigma^2}(x-(\mu-\sigma^2))^2}&=&-\sigma^2\frac{d e^{-\frac{1}{2\sigma^2}(x-(\mu-\sigma^2))^2}}{d x}+(\mu-\sigma^2)e^{-\frac{1}{2\sigma^2}(x-(\mu-\sigma^2))^2} \Label{D3}\\
\fl x^{2}e^{-\frac{1}{2\sigma^2}(x-(\mu-\sigma^2))^2}&=&\sigma^4\frac{d^{2} e^{-\frac{1}{2\sigma^2}(x-(\mu-\sigma^2))^2}}{d x^{2}}
\nonumber \\
\fl 
&& +(2(\mu-\sigma^2)x-((\mu-\sigma^2)^2-\sigma^2))e^{-\frac{1}{2\sigma^2}(x-(\mu-\sigma^2))^2} .
\Label{D4}
\end{eqnarray}
We also prepare the following formula for $e^{-x}e^{-\frac{(x-\mu)^2}{2\sigma^2}}$.
\begin{eqnarray}
\fl e^{-x}e^{-\frac{(x-\mu)^2}{2\sigma^2}}&=&e^{-\frac{(x-\mu)^2}{2\sigma^2}-x}
=e^{-\frac{1}{2\sigma^2}(x^{2}-2(\mu-\sigma^2) x+\mu^{2})} \nonumber \\
\fl &=&e^{-\frac{1}{2\sigma^2}(x-(\mu-\sigma^2))^2} e^{-\frac{\mu^{2}}{2\sigma^2}+\frac{(\mu-\sigma^2)^2}{2\sigma^2}}
=e^{-\frac{1}{2\sigma^2}(x-(\mu-\sigma^2))^2}e^{\frac{(\sigma^2-2\mu)}{2}} \Label{D5}.
\end{eqnarray}

When $X$ obeys the Gaussian distribution with the average $\mu$ and the variance $\sigma^2$,
using (\ref{D3}), (\ref{D4}), and (\ref{D5}), 
we can calculate 
the expectations of $e^{-x},xe^{-x}$, and $x^{2}e^{-x}$ as follows.
\begin{eqnarray}
\fl \rE[e^{-x}]&=&\frac{1}{\sqrt{2\pi\sigma^2}}\int^{\infty}_{-\infty}e^{-x}e^{-\frac{(x-\mu)^2}{2\sigma^2}}dx
=\frac{e^{\frac{(\sigma^2-2\mu)}{2}}}{\sqrt{2\pi\sigma^2}}\int^{\infty}_{-\infty}e^{-\frac{1}{2\sigma^2}(x-(\mu-\sigma^2))^2}dx \nonumber \\
\fl &=&e^{\frac{(\sigma^2-2\mu)}{2}}\Label{D6} \\
\fl \rE[x e^{-x}]&=&\frac{1}{\sqrt{2\pi\sigma^2}}\int^{\infty}_{-\infty}x e^{-x}e^{-\frac{(x-\mu)^2}{2\sigma^2}}dx
=\frac{e^{\frac{(\sigma^2-2\mu)}{2}}}{\sqrt{2\pi\sigma^2}}\int^{\infty}_{-\infty}x e^{-\frac{1}{2\sigma^2}(x-(\mu-\sigma^2))^2}dx \nonumber \\
\fl &=&e^{\frac{(\sigma^2-2\mu)}{2}}(\mu-\sigma^2)\Label{D7}\\ 
\fl \rE[x^{2}e^{-x}]&=&\frac{1}{\sqrt{2\pi\sigma^2}}\int^{\infty}_{-\infty}x^{2} e^{-x}e^{-\frac{(x-\mu)^2}{2\sigma^2}}dx
=\frac{e^{\frac{(\sigma^2-2\mu)}{2}}}{\sqrt{2\pi\sigma^2}}\int^{\infty}_{-\infty}x^{2} e^{-\frac{1}{2\sigma^2}(x-(\mu-\sigma^2))^2}dx \nonumber \\
\fl &=&\frac{e^{\frac{(\sigma^2-2\mu)}{2}}}{\sqrt{2\pi\sigma^2}}(\sigma^4 \int^{\infty}_{-\infty} \frac{d^{2} e^{-\frac{1}{2\sigma^2}(x-(\mu-\sigma^2))^2}}{d x^{2}} dx+2(\mu-\sigma^2)\int^{\infty}_{-\infty} xe^{-\frac{1}{2\sigma^2}(x-(\mu-\sigma^2))^2} dx  \nonumber \\ 
\fl &&-((\mu-\sigma^2 )^2-\sigma^2 ) \int^{\infty}_{-\infty} e^{-\frac{1}{2\sigma^2}(x-(\mu-\sigma^2))^2} dx) \nonumber \\
\fl &=&e^{\frac{(\sigma^2-2\mu)}{2}}((\mu-\sigma^2)^2+\sigma^2 ).
\Label{D8} 
\end{eqnarray}
Next, we calculate the real number 
$\omega_2$ when 
$\mu_1$ obeys the Gaussian distribution with the average $\mu$ and the variance $\sigma^2$.
\begin{eqnarray}
\fl\omega_2 &:=&\sum^{\infty}_{n=2} \frac{\rE[e^{-\mu_1}\mu_1^n]}{n! \rE[e^{-\mu_1}\mu_1^2]}=\frac{1}{\rE[e^{-\mu_1}\mu_1^2]} \rE[e^{-\mu_1}\sum^{\infty}_{n=2}\frac{1}{n!}\mu_1^n]\nonumber\\
\fl &=&\frac{1}{\rE[e^{-\mu_1}\mu_1^2]}\rE[e^{-\mu_1}((\sum^{\infty}_{n=0}\frac{1}{n!}\mu_1^n) -1-\mu_1)] 
=\frac{1}{\rE[e^{-\mu_1}\mu_1^2]}\rE[e^{-\mu_1}(e^{\mu_1} -1-\mu_1)] \nonumber\\
\fl &=& \frac{1}{\rE[e^{-\mu_1}\mu_1^2]}(1-\rE[e^{-\mu_1}]-\rE[\mu_1 e^{-\mu_1}])
=\frac{e^{-\frac{(\sigma^2-2\mu)}{2}}-(\mu-\sigma^2) -1}{(\mu-\sigma^2)^2+\sigma^2 }.
\Label{D9}
\end{eqnarray}

\section{Relation with Eve's success probability}\Label{as4}
We consider the state $\rho_{AE}:=\sum_m P(m) |m \rangle \langle m| \otimes \rho_{AE|m}$, where $\rho_{AE|m}$
is the composite state on $(\C^{2})^{\otimes m} \otimes {\cal H}_E$.
Now, we consider a function $f$ from $\cup_m \{0,1\}^m$ to $\{0,1\}$.
Then, we have the state $\rho_{f(A),E}=\sum_m P(m) \rho_{f(A)E|m}$ on $\C^{2} \otimes {\cal H}_E$.
Due to the monotonicity of the trace norm,
the state $\rho_{f(A),E}$ satisfies 
\begin{eqnarray}
\| \rho_{f(A),E}- \rho_{f(A)} \otimes \rho_E \|_1
\le
\| \rho_{A,E}- \rho_{\ideal} \|_1.
\end{eqnarray}

When $\rho_{f(A),E}
= p_0 |0\rangle \langle 0| \otimes \rho_{0,E}+
p_1 |1\rangle \langle 1| \otimes \rho_{1,E}$,
due to the monotonicity of the trace norm,
any two-valued POVM $\{T,I-T\}$ on ${\cal H}_E$
satisfies
\begin{eqnarray*}
\fl & \| \rho_{f(A),E}- \rho_{f(A)} \otimes \rho_E \|_1 \\
\fl \ge &
p_0 
(| \Tr \rho_{0,E} T -\Tr (p_0 \rho_{0,E}+p_1 \rho_{1,E}) T|+
| \Tr \rho_{0,E} (I-T) -\Tr (p_0 \rho_{0,E}+p_1 \rho_{1,E}) (I-T)|) \\
\fl &+
p_1
(| \Tr \rho_{1,E} T -\Tr (p_0 \rho_{0,E}+p_1 \rho_{1,E}) T|+
| \Tr \rho_{1,E} (I-T) -\Tr (p_0 \rho_{0,E}+p_1 \rho_{1,E}) (I-T)|) \\
\fl =&
4 p_0 p_1 
| \Tr \rho_{0,E} T -\Tr \rho_{1,E} T|.
\end{eqnarray*}
When 
$T$ supports $f(A)=0$ and $I-T$ supports $f(A)=1$,
the success probability is bounded by
\begin{eqnarray*}
\fl & p_0 \Tr \rho_{0,E}T+
p_1 \Tr \rho_{1,E}(I-T) 
\le 
\max(p_0,p_1)
(\Tr \rho_{0,E}(I-T)+ \Tr \rho_{1,E}T ) \\
\fl =&
\max(p_0,p_1)
(1+| \Tr \rho_{0,E} T -\Tr \rho_{1,E} T|) 
\le
\min(p_0,p_1)
(1+\frac{1}{4 p_0 p_1 }
\| \rho_{f(A),E}- \rho_{f(A)} \otimes \rho_E \|_1)
.
\end{eqnarray*}

%\bibliographystyle{sieicej}
%\bibliography{myrefs}

\section*{References}

\begin{thebibliography}{99}

\bibitem{BB84}
Bennett C H and Brassard G 1984 
{\em Proc. IEEE Int. Conf. on Computers Systems and Signal Processing
(Bangalore, India)} (New York: IEEE) pp 175--179

\bibitem{SP00}
Shor P W and Preskill J 2000
%``Simple Proof of Security of the BB84 Quantum Key Distribution Protocol," 
{\em Phys. Rev. Lett.} {\bf 85} 441

\bibitem{M01}
Mayers D 2001
%Unconditional security in Quantum Cryptography
{\em Journal of the ACM} {\bf 48} 351

\bibitem{WMU06}
Watanabe S, Matsumoto R, and Uyematsu T 2006
%``Noise Tolerance of the BB84 Protocol with Random Privacy Amplification,"
{\em Int. J. Quant. Infor.} {\bf 4} 935
%International Journal of Quantum Information, Vol.4, No.6, pp.935--946, 2006.

\bibitem{Hayashi3}
Hayashi  M 
%``Practical Evaluation of Security for Quantum Key Distribution," 
2006 {\em Phys. Rev. A} {\bf 74} 022307

\bibitem{GLLP}
Gottesman D, Lo H-K, L\"{u}tkenhaus N, and Preskill J 2004
%Security of quantum key distribution with imperfect devices
{\em Quant. Inf. Comput.} {\bf 5} 325 - 360

\bibitem{decoy1}
Hwang W-Y 2003 {\em Phys. Rev. Lett.} {\bf 91} 057901

\bibitem{decoy2}
Lo H-K, Ma X, and Chen K 2005
{\em Phys. Rev. Lett.} {\bf 94} 230504

\bibitem{decoy3}
Wang X-B 2005
{\em Phys. Rev. Lett.} {\bf 94} 230503

\bibitem{Ma05}
Ma X-F, Qi B, Zhao Y and Lo H-K 2005 %Practical decoy state for quantum key distribution 
Phys. Rev. A {\bf 72}
012326
\bibitem{Wang05}
Wang X-B 2005 %A decoy-state protocol for quantum cryptography with four intensities of coherent states
Phys. Rev. A {\bf 72} 012322


\bibitem{H1}
Hayashi M 2007
%General theory for decoy-state quantum key distribution with an arbitrary number of intensities," 
{\em New J. Phys.} {\bf 9} 284

\bibitem{decoy4}
Tsurumaru T, Soujaeff A, and Takeuchi S 2008
%``Exact minimum and maximum of yield with a finite number of decoy light intensities,''
{\em Phys. Rev. A} {\bf 77} 022319

\bibitem{Lo2}
Curty M, Moroder T, Ma X, Lo H-K, and L\"{u}tkenhaus N 2009
%Upper bounds for the secure key rate of decoy state quantum key distribution
{\em Phys. Rev. A} {\bf 79} 032335
%arXiv:0901.4669

\bibitem{wang2}
Wang X-B, Yang L, Peng C-Z, and Pan J-W 2009
%Decoy-state quantum key distribution with both source errors and statistical fluctuations
{\em New. J. Phys.} {\bf 11} 075006
    
\bibitem{wang3}
Wang X-B, Peng C-Z, Zhang J, Yang L, and Pan J-W 2008
%General theory of decoy-state quantum cryptography with source errors
{\em Phys. Rev. A} {\bf 77} 042311

\bibitem{Renner} 
Renner R 2005
{\em Security of Quantum Key Distribution}
PhD thesis, Dipl. Phys. ETH, Switzerland;
(eprint arXiv:quantph/0512258)

\bibitem{Hayashi2}
Hayashi M 2011
%, ``Exponential decreasing rate of leaked information in universal random privacy amplification,"
{\em IEEE Trans. Inform. Theory} {\bf 57} 3989

\bibitem{Hayashi5}
Hayashi M 2011
%``Tight exponential evaluation for universal composablity with privacy amplification and its applications,''
eprint arXiv:1010.1358

\bibitem{cq-security}
Hayashi M 2012
%``Classical and quantum security analysis via smoothing of R\'{e}nyi entropy order 2''
eprint arXiv:1202.0322 

\bibitem{WC81}
Wegman M N and Carter J L 1981
%``New Hash Functions and Their Use in Authentication and Set Inequality,"
{\em J. Comput. System Sci.} {\bf 22} 265

\bibitem{Miyadera}
Miyadera T 2006
%``Information-Disturbance Theorem for Mutually Unbiased Observables," 
{\em Phys. Rev. A} {\bf 73} 042317

\bibitem{H2}
Hayashi M 
%"Upper bounds of eavesdropper's performances in finite-length code with the decoy method," Physical Review A, Vol.76, 012329 (2007); 
2007 {\em Phys. Rev. A} {\bf 76} 012329;
Hayashi M 2009 {\em Phys. Rev. A} {\bf 79} 019901(E)

\bibitem{Koashi}
Koashi M 2009
% ``Simple security proof of quantum key distribution based on complementarity," 
{\em New J. Phys.} {\bf 11} 045018 

\bibitem{Renes10}
Renes J M 2011
%``Duality of privacy amplification against quantum adversaries and data compression with quantum side information,"
{\em Proc. R. Soc. A} {\bf 467} 1604

\bibitem{TH11}
Tsurumaru T and Hayashi M 2013
%``Dual universality of hash functions and its applications to quantum cryptography,''
{\em IEEE Trans. Inform. Theory}, {\bf 59} 4700?4717

\bibitem{strassen}
Strassen V 1962
Asymptotische Absch\"{a}tzugen in Shannon's Informationstheorie
In {\em Transactions of the Third Prague Conference on Information Theory etc},
Czechoslovak Academy of Sciences, Prague, pp. 689-723

\bibitem{Hayashi4}
Hayashi M
2008
%``Second-Order Asymptotics in Fixed-Length Source Coding and Intrinsic Randomness," 
{\em IEEE Trans. Inform. Theory} {\bf 54} 4619

\bibitem{Hay1}
Hayashi M 2009
%, ``Information Spectrum Approach to Second-Order Coding Rate in Channel Coding,'' 
{\em IEEE Trans. Inform. Theory} {\bf 55} 4947

\bibitem{Pol}
Polyanskiy Y, Poor H V, and Verd\'{u} S 2010 
%``Channel coding rate in the finite blocklength regime,'' 
{\em IEEE Trans. Inform. Theory} {\bf 56} 2307

\bibitem{SR08}
Scarani V and Renner R 2008
{\em Phys. Rev. Lett.} {\bf 100} 200501

\bibitem{MU10}
Sano Y, Matsumoto R, and Uyematsu T 2010
{\em J. Phys. A} {\bf 43} 495302

\bibitem{TWGR}
Tomamichel M, Lim C C W, Gisin N, and Renner R 2012
%Tight Finite-Key Analysis for Quantum Cryptography
{\em Nat. Commun.} {\bf 3} 634

\bibitem{finite}
Hayashi M and Tsurumaru T
2012
%``Concise and Tight Security Analysis of the Bennett-Brassard 1984 Protocol with Finite Key Lengths,''
{\em New J. Phys.} {\bf 14} 093014

\bibitem{FFBLSTW}
Furrer F, Franz T, Berta M, Leverrier A, Scholz V B, Tomamichel M, and 
Werner R F 2012
%Continuous Variable Quantum Key Distribution: Finite-Key Analysis of Composable Security against Coherent Attacks,
{\em Phys. Rev. Lett.} {\bf 109} 100502

\bibitem{Sasa}
Sasaki M, Fujiwara M, Ishizuka H, Klaus W, Wakui K, Takeoka M, Tanaka A, Yoshino K, Nambu Y, 
Takahashi S, Tajima A, Tomita A, Domeki T, Hasegawa T, Sakai Y, Kobayashi H, Asai T, Shimizu K, 
Tokura T, Tsurumaru T, Matsui M, Honjo T, Tamaki K, Takesue H, Tokura Y, Dynes J F, Dixon A R, 
Sharpe A W, Yuan Z L, Shields A J, Uchikoga S, Legr\'{e} M, Robyr S, Trinkler P, Monat L, Page J-B, Ribordy G,
Poppe A, Allacher A, Maurhart O, L\"{a}nger T, Peev M and Zeilinger A 2011
%``Field test of quantum key distribution in the Tokyo QKD Network,'' 
{\em Opt. Express} {\bf 19} 10387

\bibitem{Stucki}
Stucki D, Legr\'{e} M, Buntschu F, Clausen B, Felber N, Gisin N, Henzen L, Junod P, Litzistorf G, Monbaron P,
Monat L, Page J-B, Perroud D, Ribordy G, Rochas A, Robyr S, Tavares J, Thew R, Trinkler P, Ventura S, Voirol R, Walenta N and Zbinden H
2011 
%``Long-term performance of the SwissQuantum quantum key distribution network in a field environment,'' 
{\em New J. Phys.} {\bf 13} 123001

\bibitem{uni2}
Renner R and Konig R 2005 
Universally composable privacy amplification against quantum adversaries
{\em TCC: Theory of Cryptography: 2nd Theory of Cryptography Conference}, 
Lecture Notes in Computer Science vol 3378 ed J Kilian (Berlin: Springer) pp 407-25

\bibitem{uni1}
Ben-Or M, Horodecki M, Leung D W, Mayers D and Oppenheim J 2005 
The universal composable security of quantum key distribution
{\em Theory of Cryptography: 2nd Theory of Cryptography Conf., TCC 2005} 
(Lecture Notes in Computer Science vol 3378) ed J Kilian (Berlin: Springer) pp 386-406

\bibitem{FMC10}
Fung C-H F, Ma X and Chau H F 2010 
{\em Phys. Rev. A} {\bf 81} 012318

\bibitem{S91}
Stinson D R 1992
Universal hashing and authentication codes,
in J. Feigenbaum (Ed.): {\em Advances in Cryptology - CRYPTO '91}, 
LNCS 576, pp.62-73

\bibitem{Hari}
Haran R 
{\em Chernoff Bounds for Binomial and Hypergeometric Distributions},
http://www.hariharan-ramesh.com/ppts/chernoff.pdf.

\bibitem{AIC}
Akaike H 1973 
Information theory and an extension of the maximum likelihood principle,
In
B. N. Petrov and F. Csaki (Eds.), {\em Second international symposium on information theory} (pp. 267-281).
Budapest: Academiai Kiado.

\bibitem{TIC}
Takeuchi K 1976 Distribution of information statistics and a criterion of model fitting
{\em Suri-Kagaku}
(Mathematical Sciences) {\bf 153} 12--18 [In Japanese]

\bibitem{MDL}
Rissanen J 1978 %Modeling by shortest data description 
{\em Automatica} {\bf 14} 465--471

%\bibitem{two-way1}
%Abdul Khir M F, Bahari I, Ali S, and Shaari S 2011
%Weak+Vacuum and One Decoy State with Two Way Quantum Key Distribution Protocol
%arXiv:1108.4756

%\bibitem{two-way2}
%J.S. Shaari, Iskandar Bahari, Sellami Ali     
%Decoy States and Two Way Quantum Key Distribution Schemes
%arXiv:1006.1693
    
\bibitem{A-N}
Amari S and Nagaoka H 2000
{\it Methods of Information Geometry}, 
(AMS \& Oxford University Press)

\bibitem{koashi}
Koashi K 2009 New J. Phys. {\bf 11} 045018

\bibitem{koashi2}
Koashi K 2006 arXiv:quant-ph/0609180
%Masato Koashi, "Efficient quantum key distribution with practical sources and detectors "

\bibitem{AT11}
Asai T and Tsurumaru T 2011
%``Efficient Privacy Amplification Algorithms for Quantum Key Distribution" (in Japanese), 
{\em IEICE technical report}, ISEC2010-121 (in Japanese)

\bibitem{HM2013}
Hayashi M 2013
%Optimal decoy intensity for decoy quantum key distribution
arXiv:1311.3003

\bibitem{LB}
Levine B F and Bethea C G 1984
%"Single photon detection at 1.3 μm using a gated avalanche photodiode."
{\em Appl. Phys. Lett.} {\bf 44} 553 
%DOI: 10.1063/1.94800.

\bibitem{GRTZ}
Gisin N, Ribordy G, Tittel W, and Zbinden H 2002
%"Quantum cryptography." 
{\em Rev. Mod. Phys.} {\bf 74} 145


\end{thebibliography}

\end{document}